\documentclass[
superscriptaddress,
twocolumn,
nofootinbib,
amsmath,amssymb,
longbibliography,
floatfix,
]{revtex4-2}

\makeatletter 
\renewcommand\onecolumngrid{
\do@columngrid{one}{\@ne}
\def\set@footnotewidth{\onecolumngrid}
\def\footnoterule{\kern-6pt\hrule width 1.5in\kern6pt}
}
\renewcommand\twocolumngrid{
        \def\footnoterule{
        \dimen@\skip\footins\divide\dimen@\thr@@
        \kern-\dimen@\hrule width.5in\kern\dimen@}
        \do@columngrid{mlt}{\tw@}
}
\makeatother

\newcommand{\Var}{{\rm Var}}
\newcommand{\Cov}{{\rm Cov}}
\newcommand{\Ebb}{\mathbb{E}}

\usepackage{graphicx}

\usepackage{comment}
\usepackage{dcolumn}
\usepackage{bm,bbm}
\usepackage{physics}
\usepackage{svg}
\usepackage{amsmath}
\usepackage{mathtools}
\usepackage{amsthm}
\usepackage{amsfonts}
\usepackage{xcolor}
\usepackage{algorithm}
\usepackage{algorithmicx}
\usepackage{algpseudocode}
\newtheorem{theorem}{Theorem}

\newtheorem{lemma}[theorem]{Lemma}
\newtheorem{definition}{Definition}

\usepackage{tikz}

\newcommand{\be}{\begin{equation}}
		\newcommand{\ee}{\end{equation}}

\newcommand{\1}{\mathbbm{1}}

\newcommand{\eqq}[1]{Eq. \eqref{eq:#1}}

\newcommand{\uni}{\boldsymbol{\mathcal{D}}} 
\newcommand{\vol}{\boldsymbol{\mathcal{V}}} 
\newcommand{\mf}{\mathcal{L}} 
\newcommand{\f}{\mathcal{F}}

\newcommand{\vtheta}{\bm{\theta}}
\newcommand{\valpha}{\bm{\alpha}}

\newcommand{\vsigma}{\sigma}

\def\HC{\mathcal{H}}

\def\LC{\mathcal{L}}

\usepackage[makeroom]{cancel}
\usepackage[toc,page]{appendix}
\usepackage[colorlinks=true,citecolor=blue,linkcolor=magenta]{hyperref}

\usepackage{tikz}
\tikzset{every picture/.style=remember picture}

\newcommand{\poly}{\operatorname{poly}}

\renewcommand{\Ebb}{\mathbb{E}}

\newcommand{\AC}{\mathcal{A}}

\newcommand{\FC}{\mathcal{F}}

\newcommand{\OC}{\mathcal{O}}

\newcommand{\SC}{\mathcal{S}}

\renewcommand{\Var}{{\rm Var}}
\renewcommand{\Cov}{{\rm Cov}}

\renewcommand{\geq}{\geqslant}
\renewcommand{\leq}{\leqslant}

\renewcommand{\vec}[1]{\boldsymbol{#1}}  

\newcommand{\bs}{\textsf{BS}}

\newcommand{\thv}{\vec{\theta}}

\def\be{\begin{equation}}
		\def\ee{\end{equation}}
\def\bs{\begin{split}}
		\def\e{\end{split}}
\def\ba{\begin{eqnarray}}
		\def\bea{\begin{eqnarray}}
				
				\def\tea{\end{eqnarray}}
		\def\ea{\end{eqnarray}}
\def\eea{\end{eqnarray}}

\newtheorem{proposition}{Proposition}

\newcounter{precounter}
\newtheorem{propositionb}{Proposition}[precounter]

\newtheorem*{proposition*}{Proposition}

\usepackage{amssymb}
\usepackage{dsfont}

\usepackage[normalem]{ulem}
\usepackage{hyperref}

\def\be{\begin{equation}}
		\def\te{\end{equation}}
\def\ee{\end{equation}}
\def\ba{\begin{eqnarray}}
		\def\bea{\begin{eqnarray}}
				
				\def\tea{\end{eqnarray}}
		\def\ea{\end{eqnarray}}
\def\eea{\end{eqnarray}}

\newcommand{\revadd}[1]{{\color{black}{{#1}}}}

\usepackage{cancel}

\begin{document}

\preprint{APS/123-QED}

\title{{Variational quantum simulation: a case study for understanding warm starts} }

\author{Ricard Puig}
\thanks{The first two authors contributed equally to this work.}
\affiliation{Institute of Physics, Ecole Polytechnique F\'{e}d\'{e}rale de Lausanne (EPFL), CH-1015 Lausanne, Switzerland}

\author{Marc Drudis}
\thanks{The first two authors contributed equally to this work.}
\affiliation{Institute of Physics, Ecole Polytechnique F\'{e}d\'{e}rale de Lausanne (EPFL), CH-1015 Lausanne, Switzerland}
\affiliation{IBM Quantum, IBM Research – Zurich, 8803 R\"uschlikon, Switzerland}

\author{Supanut Thanasilp}
\affiliation{Institute of Physics, Ecole Polytechnique F\'{e}d\'{e}rale de Lausanne (EPFL), CH-1015 Lausanne, Switzerland}
\affiliation{Chula Intelligent and Complex Systems, Department of Physics, Faculty of Science, Chulalongkorn University, Bangkok, Thailand, 10330}

\author{Zo\"{e} Holmes}
\affiliation{Institute of Physics, Ecole Polytechnique F\'{e}d\'{e}rale de Lausanne (EPFL), CH-1015 Lausanne, Switzerland}

\date{\today}

\begin{abstract}
The barren plateau phenomenon, characterized by loss gradients that vanish exponentially with system size, poses a challenge to scaling variational quantum algorithms. Here we explore the potential of warm starts, whereby one initializes closer to a solution in the hope of enjoying larger loss variances. Focusing on an iterative variational method for learning shorter-depth circuits for quantum \textit{real time evolution} we conduct a case study to elucidate the potential and limitations of warm starts. We start by proving that the iterative variational algorithm will exhibit substantial (at worst vanishing polynomially in system size) gradients in a small region around the initializations at each time-step. Convexity guarantees for these regions are then established, suggesting trainability for polynomial size time-steps. However, our study highlights scenarios where a good minimum shifts outside the region with trainability guarantees. Our analysis leaves open the question whether such minima jumps necessitate optimization across barren plateau landscapes or whether there exist gradient flows, i.e., fertile valleys away from the plateau with substantial gradients, that allow for training. \revadd{While our main focus is on this case study of variational quantum simulation, we end by discussing how our results work in other iterative settings.}
\end{abstract}

\maketitle

\section{Introduction}

Variational quantum algorithms are a flexible family of quantum algorithms, whereby a problem-specific cost function is efficiently evaluated on a quantum computer, and a classical optimizer aims to minimize this cost by training a parametrized quantum circuit~\cite{cerezo2020variationalreview,bharti2021noisy,endo2021hybrid}. While a popular paradigm the potential of scaling these algorithms to interesting system sizes attracts much debate~\cite{anschuetz2022beyond, cerezo2023does}, in part due to the barren plateau phenomenon~\cite{marrero2020entanglement,sharma2020trainability,patti2020entanglement,wang2020noise,arrasmith2021equivalence,larocca2021diagnosing, holmes2021connecting, cerezo2020cost,rudolph2023trainability,kieferova2021quantum, tangpanitanon2020expressibility,thanaslip2021subtleties,holmes2021barren,martin2022barren,fontana2023theadjoint,ragone2023unified, thanasilp2022exponential, letcher2023tight,xiong2023fundamental, chang2024latent}. Barren plateaus are loss landscapes that concentrate exponentially in system size towards their mean value and thus, with high probability, exhibit exponentially small gradients~\cite{mcclean2018barren, arrasmith2021equivalence}. As quantum losses are computed via measurement shots, on a barren plateau landscape the resources required for training typically scale exponentially, quickly becoming prohibitive.

However, barren plateaus are fundamentally a statement about the landscape on average. They do not preclude the existence of regions of the landscape with significant gradients and indeed, the region immediately around a good minimum, must have such gradients. This has motivated the study of \textit{warm starts} whereby the algorithm is cleverly initialized closer to a minimum. Numerical studies indicate that these methods may be promising~\cite{grimsley2022adapt,dborin2022matrix,rudolph2022synergy,PhysRevA.106.L060401}. In parallel, analytic studies have proven that small angle initializations, whereby the parameterized quantum circuit is initialized in a small region typically around identity, can exhibit non-exponentially vanishing gradients~\cite{zhang2022escaping,park2023hamiltonian, wang2023trainability, park2024hardware, shi2024avoiding}. However, a good solution may be far from this region and thus these methods can (in full generality) only work on a vanishing fraction of problem instances~\cite{nietner2023unifying}.

Here we will consider a family of variational quantum algorithms that inherently use warm starts and take this as a case study to better understand their potential and limitations. In particular, we study an approach for learning shorter depth circuits for simulating quantum systems by iteratively compressing real time evolution circuits~\cite{otten2019noise, benedetti2020hardware, barison2021efficient, lin2021real, berthusen2022quantum, haug2021optimal, gentinetta2023overhead}. 
These approaches effectively use warm starts in virtue of their iterative constructions. At each iteration, the previous solution is used to initialize the parameters to learn a new compressed circuit to implement a slightly longer evolution. 
\revadd{While we focus on these iterative compression algorithms, we will end by arguing that our results can also be applied to any iterative variational quantum algorithm with a fidelity-like loss such as preparing a ground state via imaginary time evolution or learning an unknown target unitary via variational or quantum machine learning methods.}

\begin{figure*}[]
	\center
	\includegraphics[width = 1\textwidth]{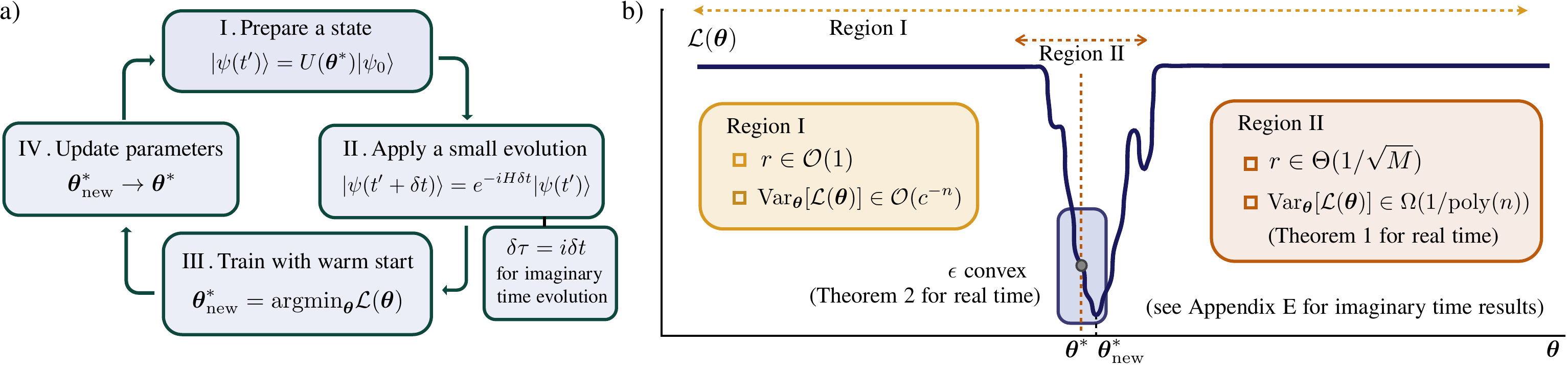}
 \caption{a) Each iteration of the variational compression scheme consists of four steps. Starting from the top: (i) apply the circuit with the last set of parameters $\vtheta^{*}$ to the initial state, (ii) apply $e^{-iH\delta t}$ for a small time-step $\delta t$, (iii) train the circuit initialising your parameters around the previous ones, (iv) update the parameters. b) We sketch a typical representation of a loss function $\mf(\thv)$ with a barren plateau across the full landscape (Region I). In Theorem~\ref{thm:variance-lower-bound} we prove that in a hypercube of width $2r$ with $r \in \Theta\left(\frac{1}{\sqrt{M}}\right)$ (sketched as Region II) the variance of the loss is only polynomially vanishing in system size $n$. In Theorem~\ref{thm:convex} we prove that in a smaller hypercube (highlighted as the blue region) the landscape is approximately convex. Similar results for \revadd{other iterative approaches with a fidelity-based loss} are discussed in Section~\ref{sec:extension-to-other-iterative} with technical details in \revadd{Appendix~\ref{app:extension}}. \revadd{Illustrative examples include preparing a ground state via imaginary time evolution as shown in Appendix~\ref{app:imaginary}, as well as learning an entire unknown unitary via a variational approach (as in Appendix~\ref{app:extension-variational-u}) or a machine learning approach (as in Appendix~\ref{app:extension-qml-u})}
 }
\label{fig:overview}
\end{figure*}  

This case study is an ideal playground for studying warm starts because its inbuilt structure allows one to analytically compute bounds on its trainability. We start by proving that as long as the training in one time-step is successful then the algorithm will exhibit substantial (at worst polynomially vanishing with problem size $n$) gradients in a small (a hypercube of radius $2r$ with $r \in \Omega( \frac{1}{\text{poly}(n)})$) region around the next initialization. We then establish guarantees on the approximate convexity of the gradients in this region and further argue that for polynomially large time-steps the new optimum will typically remain in this convex region. These results are summarised in Fig.~\ref{fig:overview}. 

However, these positive findings are partially tempered by the observation there is no guarantee that a good minimum remains in this region. Namely, there exist cases where a good minimum jumps from the region with trainability guarantees. Our analysis leaves open the question of whether such minima jumps necessitate optimization across barren plateau landscapes or whether there exist valleys away from the barren plateau that allow for training. Such fertile valleys, i.e. small regions away from the plateau with substantial gradients, that allow for successful training are theoretically possible but to what extent they arise in practise is an open question.

\section{Preliminaries}

\subsection{Iterative Variational \revadd{Time-evolution} Compression}\label{sec:alg}

We consider simulating the evolution of some initial state $\ket{\psi_0}$ under a Hamiltonian $H$ up to time $t$. That is, our aim is to implement a quantum circuit that approximates $e^{-i H t} \ket{\psi_0}$. One standard approach \revadd{is} to break the total evolution $t$ into a sequence of $N$ short $\delta t$ evolutions with $t = N \delta t$ such that
\begin{equation}\label{eq:real-time-evolution}
	\revadd{e^{-i H t} \ket{\psi_0} = \prod_{k=1}^{N } e^{-i H \delta t} \ket{\psi_0}}  \, .
\end{equation}
\revadd{For each small time-step evolution $e^{-iH\delta t}$, one could further decompose it by using either Trotter approximation for real-time dynamics~\cite{lloyd1996universal}. For example, for the real-time evolution case, if the Hamiltonian is expressed in Pauli basis as $H = \sum_{i=1}^{N_P} \alpha_i P_i$ where $\alpha_i$ and $P_i$ are associated coefficients and Pauli operators, the first-order Trotter approximation of the entire dynamics can be written as $e^{-i H t} \ket{\psi_0} \approx \prod_{k=1}^{N }\prod_{i=1}^{N_P} e^{-i \alpha_iP_i \delta t} \ket{\psi_0}$ with the approximate error $\mathcal{O}(N \delta t)$.} However, \revadd{all} these approaches are fundamentally limited by the linear growth of circuit depths with time simulated. This has prompted ongoing efforts to find alternative approaches that avoid this linear growth~\cite{trout2018simulating, endo2020variational,yao2020adaptive, otten2019noise, lin2021real, benedetti2020hardware, barison2021efficient, heya2019subspace, cirstoiu2020variational, gibbs2021long, gibbs2022dynamical, eassa2023high, bharti2020quantum,lau2021quantum,haug2020generalized, kokcu2022Fixed, steckmann2021simulating, jamet2021krylov}.

Here we focus on the proposal to use a variational quantum algorithm to compress the depth of the circuit at each iteration of the algorithm~\cite{otten2019noise, benedetti2020hardware, barison2021efficient, lin2021real, berthusen2022quantum, haug2021optimal, gentinetta2023overhead} as shown in Fig.~\ref{fig:overview}. More concretely, at any iteration of the algorithm, one variationally minimizes the following loss:
\begin{equation}\label{eq:loss}
\mf\left(\vtheta\right) = 1 - |\langle \psi_0| U(\vtheta)^\dagger e^{-i H \delta t} U(\vtheta^{*}) |\psi_0 \rangle |^2
\end{equation}
where $U(\thv)$ is a parameterized quantum circuit and $\thv^{*}$ denotes the optimized parameters found at the previous iteration step~\cite{otten2019noise, barison2021efficient, berthusen2022quantum}. At iteration $k$ the loss $\mf\left(\vtheta\right)$ is optimized using a hybrid quantum classical optimization loop to find the next set of optimal parameters $\thv^{*}_{\rm new}$. We note that while we focus on a fidelity loss here other cost functions are possible, e.g. in Ref.~\cite{benedetti2020hardware} they considered the real part of the state overlap and in Ref.~\cite{barison2021efficient} a local fidelity measure, and a similar analysis could be performed in those cases. It is also possible to use this approach to learn circuits to prepare approximate ground states and thermal states by replacing $it$ with $\tau$ in Eq.~\eqref{eq:loss} and with an appropriate choice in initial state $\ket{\psi_0}$~\cite{benedetti2020hardware}.

In practice, for each small evolution, one could further decompose it by using either Trotter approximation for real-time dynamics~\cite{lloyd1996universal}, or a similar procedure for imaginary-time dynamics~\cite{benedetti2020hardware}. For example, for the real-time evolution case, if the Hamiltonian is expressed in Pauli basis as $H = \sum \alpha_i P_i$ where $\alpha_i$ and $P_i$ are associated coefficients and Pauli operators, the first-order Trotter approximation of the entire dynamics can be written as $e^{-i H t} \ket{\psi_0} \approx \prod_{k=1}^{N } e^{-i \alpha_iP_i \delta t} \ket{\psi_0}$ with the approximate error $\mathcal{O}(N \delta t)$. \revadd{However our results are agnostic to the method used to implement the evolution step.}

The success of this protocol depends on a variety of factors including the choice of \textit{ansatz} for the parameterized circuit. Here we focus on a general ansatz of the form
\begin{align}\label{eq:circuit}
	U\left( \vtheta \right) = \prod_{i=1}^M V_i U_i(\theta_i)
\end{align}
where $\left\{ V_i \right\}_{i=1}^M$ are a set of fixed unitary matrices, $\left\{ U_i(\theta_i) = e^{-i\theta_i\vsigma_i} \right\}_{i=1}^M$ are parameter-dependent rotations, $M$ is the number of parameters in the circuit, and $\{ \vsigma_i \}_{i=1}^M$ is a set of generators on $n$ qubits such that $\vsigma_i = \vsigma_i^\dagger$ and $\vsigma_i^2 = \1$. In this work, we assume that all parameters $\theta_j$ are uncorrelated.

\subsection{Gradient magnitudes and barren plateaus}

In recent years there has been concerted effort to understand when quantum losses are trainable or untrainable. Several factors can lead to untrainable losses including the presence of sub-optimal local minima~\cite{bittel2021training,anschuetz2022beyond,anschuetz2021critical}, expressivity limitations~\cite{tikku2022circuit} and abrupt transitions in layerwise learning~\cite{campos2021abrupt}. However, much of this research has focused on loss gradients. \revadd{We also refer the readers to the recent review~\cite{larocca2024review} for a broader overview.}

To train a variational quantum algorithm successfully, the loss landscape must exhibit sufficiently large loss gradients (or more generally, loss differences). Chebyshev’s inequality bounds the probability that the
cost value deviates from its average as
\begin{equation}\label{eq:Cheb}
    {\rm Pr}_{\thv}( | \LC(\thv) -  \Ebb_{\thv}[\LC(\thv)] | \geq \delta ) \leq \frac{\text{Var}_{\thv}[\mathcal{L}(\thv)]}{\delta^2} \, ,
\end{equation}
for some $\delta>0$ and the variance of the loss defined as
\begin{equation}\label{eq:vardef}
    \Var_{\vec{\theta}}[\mathcal{L}(\thv)] = \Ebb_{\thv}\left[\LC^2(\thv)\right] -  \left(\Ebb_{\thv}\left[\LC(\thv)\right]\right)^2  \,, 
\end{equation}
where the expectation value is taken over the parameters. Hence if the variance is small then the probability of observing non-negligible loss differences for any randomly chosen parameter setting is negligible.

For a wide class of problems~\cite{marrero2020entanglement,sharma2020trainability,patti2020entanglement,wang2020noise,arrasmith2021equivalence,larocca2021diagnosing, holmes2021connecting, cerezo2020cost,rudolph2023trainability,kieferova2021quantum, tangpanitanon2020expressibility,thanaslip2021subtleties,holmes2021barren,martin2022barren,fontana2023theadjoint,ragone2023unified, thanasilp2022exponential, letcher2023tight,xiong2023fundamental, chang2024latent}, one can show that the loss variance vanishes exponentially with problem sizes, i.e. $  \Var_{\vec{\theta}}[\mathcal{L}(\thv)] \in \mathcal{O}(c^{-n})$ with $c > 1$. On such \textit{barren plateau} landscapes exponentially precise loss evaluations are required to navigate the towards the global minimum and hence the resources (shots) required for training also scales exponentially. This has prompted the search for architectures where loss variances vanish at worst polynomially with system size,  $  \Var_{\thv}[\mathcal{L}(\thv)] \in \Omega\left(\frac{1}{\text{poly}(n)}\right)$, such that resource requirements may scale polynomially.

\section{Main Results}

In this paper we analyse the trainability of the variational \revadd{time-evolution} compression scheme and use this to illustrate the complex interplay between the barren plateau phenomena, local minima and expressivity limitations. We start by presenting an overview of the factors we will consider, and the context provided by prior work, before proceeding to present our main analytic and numerical findings.

\subsection{Overview of analysis}

The variance of the loss, Eq.~\eqref{eq:vardef}, necessarily depends on the parameter region it is computed over. The majority of analyses of quantum loss landscapes have assumed the angles are initialized according to some distribution in the region $[0, 2\pi]$ and hence considered the variance over the entire loss landscape~\cite{marrero2020entanglement,sharma2020trainability,patti2020entanglement,wang2020noise,arrasmith2021equivalence,larocca2021diagnosing, holmes2021connecting, cerezo2020cost,rudolph2023trainability,kieferova2021quantum,thanaslip2021subtleties,holmes2021barren,martin2022barren,fontana2023theadjoint,ragone2023unified, thanasilp2022exponential, letcher2023tight}. However, in practise one is interested in the loss landscape in the region explored during the optimisation process (i.e., in the region around the initialization, the region around the sufficiently `good' minima and ideally the landscape that connects these regions). An analytic study of these different regions in the general case seems daunting. However, the structure provided by the variational \revadd{time-evolution} compression scheme allows us to take steps in this direction.

Prior work has established that small angle initializations, whereby the parameterized quantum circuit is initialized in a small region around identity, provide a means of provably avoiding barren plateaus~\cite{zhang2022escaping,park2023hamiltonian, wang2023trainability, park2024hardware, shi2024avoiding}. More concretely, let us define
\begin{equation}\label{eq:hypercube}
	\vol(\vec{\phi}, r) := \{ \vec{\theta} \} \; \; \text{such that} \; \; \theta_i \in [\phi_i -r , \phi_i + r ]  \, \, \forall \, \, i ,
\end{equation}
as the hypercube of parameter space centered around the point $\vec{\phi}$, and define $\uni(\vec{\phi}, r)$ as a uniform distribution over the hypercube $\vol(\vec{\phi}, r)$. It was shown in Ref.~\cite{wang2023trainability}, that if the parameters are uniformly sampled in a small hypercube with $r \in \mathcal{O}\left(\frac{1}{\sqrt{L}}\right)$ around $\vec{\phi} = \vec{0}$ for some hardware efficient architecture with $L$ being the number of layers, then the variance $\Var_{\thv\sim\uni(\vec{0}, r)}[\mathcal{L}(\thv)] \in \Omega\left(\frac{1}{\text{poly}(L)}\right)$ decays only polynomially with the depth of the circuit. Similar conclusions were reached for the Hamiltonian Variational Ansatz in Refs.~\cite{park2023hamiltonian, park2024hardware} and for Gaussian initializations in Ref.~\cite{zhang2022escaping,shi2024avoiding}. Typically in these cases the small angle initialization corresponds to initializing close to identity. 

These guarantees can broadly be used to argue that the first iteration of variational \revadd{time-evolution} compression scheme will exhibit non-vanishing variances for certain circuits. Moreover, assuming $\delta t$ is small such that $ e^{-i H \delta t}$ is close to identity, and assuming that the ansatz is sufficient expressive to be able to capture a good approximation of $ e^{-i H \delta t}$, it is reasonable to expect that the good approximate solution circuit is contained within the region $\vol\left(\vec{0},\frac{1}{\sqrt{M}}\right)$ that enjoys polynomial loss variances. However, at later time-steps, when $U(\thv)$ is far from identity, the guarantees provided for these small angle initializations are of debatable relevance. 

Here we address the task of providing guarantees for all iterations of the algorithm for a very general family of parameterized quantum circuits. 
Our guarantees are based on the observation that assuming the previous step was sufficiently well optimised, and the time-step $\delta t$ is small enough, then one can initialize close enough to the new global minimum (or, more modestly, a good new minimum) such that the landscape exhibits substantial gradients as sketched in Fig.~\ref{fig:overview}b). In Section~\ref{sec:variance} we use this observation to derive such analytical variance lower bounds. We note that Ref.~\cite{haug2021optimal} provides an approximate lower bound on the variance in the loss for an iterative compression scheme; however, to do so it makes a number of approximations and in effect assumes the convexity of the problem from the outset. We go beyond this by providing exact bounds without prior assumptions.

Non-vanishing gradients in the region around the initialisation are a necessary condition to have any hope of successfully training a variational quantum algorithm but they are far from sufficient. Of particular importance is the potential to become trapped in spurious local minima. One way to provide guarantees against this concern is to prove that this region is convex. We tackle this issue in Section~\ref{sec:convexity} by proving convexity guarantees in the region around the initialization provided by the previous iteration.

We then introduce the notion of the adiabatic minima as the minima that would be reached by increasing $\delta t$ infinitely slowly and minimizing $\mf(\vtheta)$ by gradient descent with a very small learning rate in Section~\ref{sec:adiabaticminimum}. We argue that as long as the time-step is $\delta t$ is not too large (i.e. decreases polynomially with the number of parameters $M$) we can ensure that an adiabatic minima is within the convex region with non-vanishing gradients, and thus it should be possible to train to the adiabatic minimum.  

Finally, in Section~\ref{sec:minimum-jump} we explore the limitations of our analytic bounds. Firstly, we highlight that our analysis can not provide convergence guarantees to a good minimum because \textit{minima jumps} are possible. Namely, from one time-step to the other, the adiabatic minimum can become a relatively poor local minimum and a superior minimum can emerge elsewhere in the landscape. We are then faced with the question of whether such minima jumps necessitate optimisation across barren plateau landscapes or whether there exist gradient flows between these minima. We provide numerical evidence for a 10-qubit contrived example of a gradient flow from an initialization to a seemingly jumped minima which suggests such gradient flows can exist. 

\subsection{Lower-bound on the variance}\label{sec:variance}

\begin{figure*}
    \centering
    \includegraphics[width=0.825\linewidth]{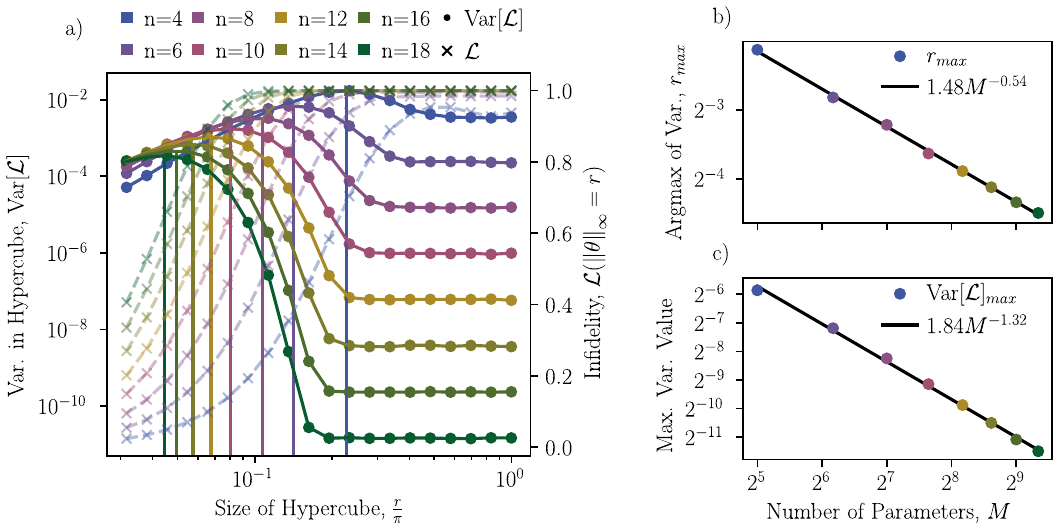}
    \caption{
    \textbf{Variance of landscape and width of narrow gorge.} 
    Here we study the landscape of $\mathcal{L}(\thv)$, for the first time-step of the variational \revadd{time-evolution} compression algorithm, for different system sizes $n$ as a function of the width of the hypercube $r$. We consider a hardware efficient 
    ansatz with $n$ layers and random initial parameters within the hypercube.
    a) We plot $\mathcal{L}(\thv)$ and its variance $\Var_{\thv\sim\uni(\vec{0}, r)}[\mathcal{L}(\thv)]$ as function of $r/\pi$. Since the shape of the landscape depends on the direction of the parameter update, to plot $\mathcal{L}(\thv)$ we have taken the average over 500 different directions. For $\Var_{\thv\sim\uni(\vec{0}, r)}[\mathcal{L}(\thv)]$, we
    keep track of its maximum value (marked with a vertical line) for each system size.
    b) The value $r_{\rm max}$ for which the variance peaks as function of the number of parameters in the ansatz.
    c) Maximum value of the variance for different system sizes. While the results shown here are for the first iteration of the variational compression scheme very similar results are observed at later iterations (in line with Theorem~\ref{thm:variance-lower-bound}).
    }
    \label{fig:variance}
\end{figure*}

The variance of the loss function at any iteration around the parameters $\vtheta^*$ obtained for the previous iteration will depend on the length of the time-step $\delta t$ as well as the volume of the region of the parameter space explored. Here we study the variance of the loss in a uniformly sampled hypercube  
of sides $2 r$ around $\vtheta^*$ as defined in Eq.~\eqref{eq:hypercube}. As proven in Appendix~\ref{app:variance}, we obtain the following bound.

\begin{theorem}[Lower-bound on the loss variance, Informal]\label{thm:variance-lower-bound} 
Consider the general ansatz in Eq.~\eqref{eq:circuit} and assume that in the first iteration the system is prepared in \revadd{an} initial state $\rho_0$ and let us choose $\vsigma_1$ such that $\Tr[\rho_0 \vsigma_1 \rho_0 \vsigma_1] = 0$. Given that the time-step is bounded as
\begin{align}
    \delta t \in \OC\left(\frac{1}{\lambda_{\rm max}}\right)
\end{align}
where $\lambda_{\rm max}$ is the largest eigenvalue of $H$ 
and we consider uniformly sampling parameters in a hypercube of width $2r$ around the solution from the previous iteration $\thv^*$, i.e. $\vol(\thv^*,r)$, such that 
\begin{align}
 r \in \Theta\left(\frac{1}{\sqrt{M}} \right) \;. 
\end{align}
Then the variance at any iteration of the algorithm is lower bounded as 
\begin{align}
    \Var_{\vtheta \sim\uni(\vtheta^*, r)} \left[ \LC (\vtheta)\right] \in \Omega\left( \frac{1}{M}\right) \;.
\end{align}
Thus, for $M \in \OC(\poly(n))$ we have 
\begin{align}
    \Var_{\vtheta \sim\uni(\vtheta^*, r)} \left[ \LC (\vtheta)\right] \in \Omega\left( \frac{1}{\poly(n)}\right) \;.
\end{align}
\end{theorem}
\revadd{We discuss later in Section~\ref{sec:extension-to-other-iterative} how to abstractify the theorem for other iterative settings with the fidelity-type loss e.g., for imaginary time evolution or unitary learning problems.}

\medskip

Theorem~\ref{thm:variance-lower-bound} establishes that within a small, but non-exponentially vanishing ($r \propto 1/\sqrt{M}$), region around the previous optimal solution, the loss landscape will exhibit non-exponentially vanishing gradients so long as $\delta t \in \OC(1/\lambda_{\rm max})$. The constraint on the time-step is to ensure that a state corresponding to a previous solution has a large overlap with a new target state. If the \revadd{time-}step is too large, then the initialization no longer contains enough information about the target state and it is equivalent to initializing on the barren plateau region. 

\revadd{We now discuss three key subtleties regarding the constraint on $\delta t$. (i).~The $\delta t \in \OC(1/\lambda_{\rm max})$ scaling comes from a loose bound on the overlap between the old optimized state and the new target state (i.e., it is a sufficient condition but may not be necessary) and thus a larger $\delta t$ is likely viable in practise. (ii).~This constraint on $\delta t$ does not affect the size of the attraction region, but rather ensures that step is initialized within the region of attraction. That is, as long as $\delta t \in \OC(1/\lambda_{\rm max})$, the region of attraction on that iteration is theoretically guaranteed to scale as $\Theta(1/\sqrt{M})$. (iii) As a consequence of the previous point, the overhead due to this $\delta t$ constraint only results in increasing the total number of steps i.e., $N \sim \lambda_{\rm max} t$. Particularly, when $\lambda_{\rm max}$ scales linearly with $n$, the overhead is linear with the system's size.}

On another related topic, our bound is presented here for simulating the evolution of \revadd{an initial state that satisfies the orthogonality condition with the first gate i.e., $\Tr[\rho_0 \vsigma_1 \rho_0 \vsigma_1] = 0$}. This assumption is made for ease of presentation. However, as highlighted in the appendices, this assumption is not strictly necessary. Rather one just needs to ensure that a gate in the first layer has a non-trivial effect on the loss. 

It is important to stress that Theorem~\ref{thm:variance-lower-bound} provides a sufficient, not a necessary, condition for observing polynomially vanishing gradients. For a necessary condition one would need to derive an upper bound as a function of $r$ and $\delta t$. In general, this seems challenging and is likely to be highly ansatz dependent~\cite{braccia2024computing}. 
Instead we address this question numerically. 

In Fig.~\ref{fig:variance} we study the landscape of $\mathcal{L}(\thv)$ for an initial time-step as a function of the width of the hypercube, $2r$. For concreteness, we consider a hardware efficient ansatz with $n$ layers and uniformly sampled parameters within the hypercube. The variance over the full landscape ($r = \pi $) vanishes exponentially in $n$. However, as $r$ is decreased the variance increases with $r$ and ceases to decay exponentially in $n$. When $r$ is very small the variance again begins to decrease. This is because for sufficiently small $r$ we are computing the variance over a small region of the loss landscape at the base of the narrow gorge. This account is confirmed by the average behaviour of $\mathcal{L}(\thv)$ also shown in Fig.~\ref{fig:variance}a). In particular, when the variance peaks, we have an infidelity of approximately 0.7 for each system size, which indicates that the peak of the variance is a good measure of the width of our gorge.

In  Fig.~\ref{fig:variance} b) and c) we plot the $r$ value for which the variance peaks and the maximum value of the variance as function of the number of parameters in the ansatz $M$. Both quantities decay polynomially in $M$. In particular, we find that $r_{\rm max}$ scales as $\frac{1}{\sqrt{M}}$ implying that the width of the gorge decreases with a $\frac{1}{\sqrt{M}}$ scaling. This is consistent with our theoretical lower bound, Theorem~\ref{thm:variance-lower-bound}, which also suggests that to ensure at worst poly vanishing gradients its necessary to consider a region of width $\frac{1}{\sqrt{M}}$ around the minimum. \revadd{Finally, in Appendix~\ref{appx:numeric-var}, we further plot the variance as a function of the time-step $\delta t$ and find good agreement with our analytic claim that $\delta t \propto 1/\lambda$ suffices to ensure at worst polynomially vanishing gradients.}

\subsection{Convexity region around the starting point}\label{sec:convexity}
Substantial gradients are a necessary condition but not sufficient condition for trainability. If the substantial gradients are attributable to poor local minima then finding a good solution is likely to be highly challenging. However, if as well as having substantial gradients, we can prove that the landscape is convex, or approximately convex, then training to a minimum looks promising.  
In this section, we present a theorem which shows that the region around the starting parameters is approximately convex. As expected, our condition depends both on the width of the hypercube region considered and the time-step $\delta t$ taken. 

A function is convex over a parameter range if its second order partial derivatives are all non-negative in that parameter range. In practise, a more convenient means of diagnosing convexity is to study the Hessian, $\nabla^2_{\thv}[\LC(\thv)]$, of a function. If the Hessian is positive semi definite, i.e. $\nabla^2_{\thv}[\LC(\thv)] \geq 0$, in a given parameter region then the function is convex in that region. We will introduce a notion of approximate convexity by relaxing this constraint and saying that a landscape is $\epsilon$ convex if the following condition holds:
\begin{definition}[$\epsilon$-convexity]\label{def:epsilon-convex}
A loss is $\epsilon$-convex in the region $\thv \in \vol(\thv^*, r_c)$ if 
 \begin{equation}
     \left[\nabla^2_{\thv} \mathcal{L}(\thv)\right]_{\rm min} \geq -|\epsilon| \, 
 \end{equation}
for all $\thv \in \vol(\thv^*, r_c)$. Here $\nabla^2_{\thv} \mathcal{L}(\thv)$ denotes the Hessian of $\mathcal{L}(\thv)$ and we denote $\left[ A \right]_{\rm min}$ as the smallest eigenvalue of the matrix $A$.
\end{definition}
If a loss is $\epsilon$-convex the loss is convex up to $\epsilon$ small deviations, as sketched in Fig.~\ref{fig:overview}b), and argued more formally in Appendix~\ref{app:epsilonconvexity}. This notion is particularly important in a quantum context where the loss is only ever measured with a finite number of shots making it hard to tell apart $\epsilon$ negative curvatures from $\epsilon$ positive ones. Thus in practise the relevant $\epsilon$ will be determined by the shot noise floor. 

Equipped with this definition we now show that a polynomially sized region around the starting point of the previous iteration is approximately convex. 
\begin{theorem}[Approximate convexity of the landscape, Informal] \label{thm:convex} For a time-step of size  
\begin{align}
    \delta t \in \OC\left( \frac{\mu_{\rm min} + 2 |\epsilon| }{M \lambda_{\rm max}}\right) \;,
\end{align}
the loss landscape is $\epsilon$-convex in a hypercube of width $2r_c$ around a previous optimum $\vec{\theta^*}$ i.e., $\vol(\thv^*, r_c)$ such that
\begin{equation}\label{eq:convexitytheorem}
    r_c \in \Omega\left(\frac{\mu_{\rm min}+2|\epsilon|}{16 M^2} - \frac{\lambda_{\rm max} \delta t}{M}\right) \;,
\end{equation}
where $\mu_{\min}$ is the minimal eigenvalue of the Fisher information matrix associated with the loss at $\thv^*$.
\end{theorem}
In Appendix~\ref{app:imaginary} we show that an analogous convexity guarantee can be proven for imaginary time evolution. \revadd{We additionally discuss the extension to the unitary learning tasks in Appendix~\ref{app:extension-variational-u} and Appendix~\ref{app:extension-qml-u}}  \\

Theorem~\ref{th:pRTE} tells us that it is always possible to pick a polynomially scaling $\delta t, r_c$ such that the landscape of $\mf(\thv)$ with respect to the parameters is approximately convex. The constraints on $\delta t$ and $r_c$ for convexity are pretty stringent in practice. However, convexity is also a lot to demand of a loss landscape. Nonetheless, it is nice to see that approximate convexity can be ensured at `only' a polynomially scaling cost. 

\begin{figure*}
    \centering
    \includegraphics[width=\linewidth]{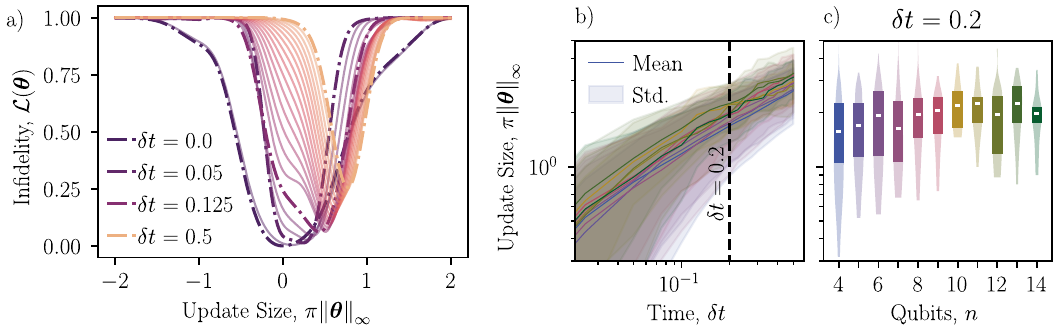}
    \caption{\textbf{Routine evolution of the adiabatic minimum}.
    Here we study the landscape of $\mf(\thv)$ as we increase the time-step $\delta t$. We study a 10 qubit Hamiltonian with nearest-neighbour interactions on a 1D lattice with $H=\sum X_i Z_{i+1} - 0.95 \sum Y_i$ where $X_i$, $Y_i$ and $Z_i$ are X-Pauli, Y-Pauli and Z-Pauli operators on the qubit $i$. We use a 2-layered Hamiltonian Variational Ansatz with random initial parameters.
    a) We plot our landscape for different $\delta t$. The cuts in our high dimensional $\mf(\vtheta)$ space
    contain both the initial parameters $\vtheta=\vec{0}$ and the adiabatic minimum  
    $\vtheta_{A}(\delta t)$ at $\delta t$.
    \revadd{
    b) We plot the size of our parameter update $\norm{\vtheta}_\infty$, i.e. the distance along the cuts between the old minimum and the new adiabatic minimum, as a function of the time-step for different system sizes (from $n=4$ to $14$ for at least $20$ different instances for each qubit).
    We repeat the experiment with different random initial parameters and plot their mean and standard deviation.
    c) We show a violin and box plot--with the median and quartiles-- of the distributions we obtain for $\delta t=0.2$ as we increase the number of qubits. Note that the color assigned to each number of qubits matches that of the curve in b). 
    }
    }
    \label{fig:UpdatingLandscape}
\end{figure*}

\subsection{Adiabatic minimum}\label{sec:adiabaticminimum}

So far we have identified two constraints on our parameters that push in the direction of trainability guarantees. Specifically, we have established a region in our landscape with substantial gradients and approximate convexity. The final condition required for convergence guarantees is to ensure our target circuit, i.e., a good minimum, lies within this region.

To address this point, let us start by introducing the notion of the \textit{adiabatic minima}. 
Intuitively, these are the minima that would be reached by increasing $\delta t$ infinitely slowly and minimizing $\mf(\vtheta)$ by gradient descent with a very small learning rate. By analogy, one can imagine dropping a marble in the initial minima and then slowly modifying the landscape by infinitesimally increasing $\delta t$. The position of the marble would correspond to our adiabatic minima and, in practice, it is where we expect our algorithm to converge for sufficiently small $\delta t$. 
In Fig.~\ref{fig:UpdatingLandscape} a) we plot a cut through the cost landscape around the old minimum $\vec{\theta}^*$ as a function of $\delta t$. We can see that the minimum smoothly moves rightwards and increases with increasing $\delta t$. More formally, we define the adiabatic minima as follows.

\begin{definition}[Adiabatic Minimum]\label{def:adiabatic-minimum}
For any time $\delta t$ in the range $[0,T]$, a function\footnote{We note that it is in fact possible for a single initial minima to have multiple corresponding adiabatic minima functions if there are multiple directions with zero gradients.} corresponding to the evolution of the adiabatic minima for some initial minimum $\vtheta^*$, is a continuous function $\vtheta_A(\delta t) \in C^{\infty}(\mathbb{R},\mathbb{R}^{m})$ such that $\vtheta_A(0) = \vtheta^*$ and
 $\nabla_{\thv} \mf(\vtheta_A(\delta t), \delta t)=\vec{0}$. The adiabatic minimum at time $\delta t$ is $\vtheta_A(\delta t)$.
 \label{def:adiabaticminimum}
\end{definition}

One can ensure that the iterative variational compression scheme will converge to some minimum if the time-step is small enough 
that an adiabatic minimum is inside of the convex region with non-vanishing gradients. One can question how good this minimum will be but we will set aside this question for the moment. Thus our next step will be to assess how small the time-step needs to be picked in order to guarantee this. 
In the following theorem we formalize this concept by bounding the $\delta t$ required to ensure an adiabatic minimum is in the substantial gradient region and in the convex region.

\begin{theorem}[Adiabatic minimum is within provably `nice' training region, Informal]\label{thm:adiabaticminimum}
If the time-step $\delta t$ is chosen such that
\begin{align}
    \delta t \in \OC\left(\frac{\beta_A}{M \lambda_{\rm max}}\right) \;,
\end{align}
then the adiabatic minimum $\thv_A(\delta t)$ is guaranteed to be within the non-exponentially-vanishing gradient region (as per Theorem~\ref{thm:variance-lower-bound}), and additionally, if $\delta t$ is chosen such that
\begin{align}
    \delta t \in \OC\left( \frac{\beta_A(\mu_{\rm min} + 2|\epsilon|) }{ M^{5/2} \lambda_{\rm max}}\right) \;,
\end{align}
then the adiabatic minimum $\thv_A(\delta t)$ is guaranteed to be within the $\epsilon$-convex region (as per Theorem~\ref{thm:convex}) where 
\begin{equation}
    \beta_A := \frac{\Dot{\vtheta}_A^T(\delta  t) \left( \nabla^2_{\vtheta} \LC(\vtheta) \big|_{\vtheta = \vtheta_A(\delta  t)}\right) \Dot{\vtheta}_A(\delta t) }{\| \Dot{\thv}_{A}(\delta t) \|_2^2} 
\end{equation}
corresponds to the second derivative of the loss in the direction in which the adiabatic minimum moves.
\end{theorem}

Theorem~\ref{thm:adiabaticminimum} tells us that it suffices to consider a time-step that scales as $\delta t \in \Omega\left(\frac{1}{M}\right)$ to ensure that the adiabatic minimum falls in the region with substantial gradients, or more stringently to take $\delta t \in \Omega\left(\frac{1}{\text{poly}(M)}\right)$ to ensure that the adiabatic minimum falls within the $|\epsilon|$-convex region. As in general $M \sim \text{poly}(n)$ it follows that if $\delta t$ is decreased polynomially with problem size, and the learning rate is chosen appropriately, it should be possible to train to the new adiabatic minimum.

We stress that this interpretation of Theorem~\ref{thm:adiabaticminimum} is only possible assuming that $\beta_A$ is not exponentially vanishing. This is a reasonable assumption as $\beta_A \rightarrow 0$ corresponds to the curvature of the loss at the minimum being flat in the direction in which the adiabatic minimum moves. While this is conceivably possible it is unlikely in practise (as is supported by our numerics in Fig.~\ref{fig:UpdatingLandscape} and Fig.~\ref{fig:minimumjump}).
Moreover, the $\beta_A$ dependence of Theorem~\ref{thm:adiabaticminimum} is a genuine feature that affects trainability, rather than a relic of our proof techniques. Namely, if the landscape is very flat in the direction of the new minimum then indeed the adiabatic minimum can move significant distances at short times. More poetically, one might visualise this case as a \textit{barren gorge}. That is, a sub-region of the landscape within the substantial gradient region that nonetheless has vanishing gradients. Such features are possible but perhaps unlikely unless the ansatz is highly degenerate.  

\revadd{To partially support our claim that the distance traveled by the minimum does not blow up with respect to the system's size we provide Fig.~\ref{fig:UpdatingLandscape}. Specifically, in the panel b) we see that while the distribution of distances traveled by the adiabatic minima is broad, even in the worse case scenario it seemingly scales polynomially with $\delta t$. By taking a cut in this distribution at an arbitrary time $\delta t = 0.2$--as seen in Fig.~\ref{fig:UpdatingLandscape} c)-- we do not perceive an significant scaling of the average distance with system size. }

Another caveat is that Theorem~\ref{thm:adiabaticminimum} only holds in the case that there exists a well-defined adiabatic minimum function $\thv_A(\delta t)$ in the time interval of interest. This is not always guaranteed to be the case because a minimum can vanish by evolving into a slope as $\delta t$ increases. If this occurs then the continuity condition in our definition of the adiabatic minima function fails. Nonetheless, this is not a situation that necessarily causes trainability problems (if the minimum turns into a slope then training is possible down that slope), rather it is a situation that makes finding analytic trainability guarantees more challenging. For a more detailed discussion of this caveat and a proof of Theorem~\ref{thm:adiabaticminimum} see Appendix~\ref{app:moving-min}.

Relatedly, it is worth mentioning that while Theorem~\ref{thm:adiabaticminimum} allows for polynomially shrinking step sizes in practise these step sizes are rather small. In particular, for the small problem sizes studied already in the literature practitioners have typically used larger step sizes than those that we have managed to derive guarantees for here. In parallel, we can see from our numerical implementations in Fig.~\ref{fig:UpdatingLandscape} that training would seem viable for larger $\delta t$ than allowed by our bounds. It is arguably an open question to what extent this can be attributed to looseness of our bounds or the small problem sizes that can be simulated classically. One thing to note in this regard is that optimization is often much more successful in practise than can be analytically guaranteed or even explained. As such, in practise, larger $\delta t$ may well be viable. This is specially relevant if one considers using adaptive approaches where $\delta t$ is modified at each step until a given precision threshold is reached. While heuristic, this method in the worst case enjoys the mathematical guarantees proven here, while in the best case allows for larger time-steps (and so reduces the average number of time-steps required in total). 

\begin{figure}
    \includegraphics[width=0.99\linewidth]{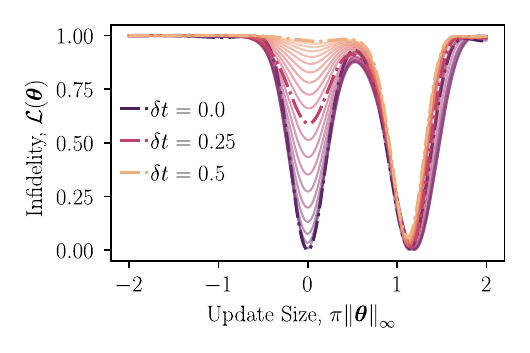}
    \caption{\textbf{Minimum jump.}
    Here we show a 1D-cut of the landscape $\mathcal{L}(\thv)$ as we increase the time-step $\delta t$. The cut includes the initial parameters-with update $\delta t = 0$ and $|| \vec{\theta} ||_\infty = 0$. We choose a 10 qubit Ising Hamiltonian $H=\sum X_iX_{i+1} - 0.95 \sum Y_i$ on a 1D-lattice. We use a 2-layered Hamiltonian Variational Ansatz.} 
    \label{fig:minimumjump}
\end{figure}

\medskip 

\subsection{Minima jumps and fertile valleys}
\label{sec:minimum-jump}

A final limitation of our analysis is that the adiabatic minimum (or indeed any minimum within the region with gradient guarantees) need not be a good minimum. The adiabatic minimum is the minimum that evolves away from the old minimum after the application of a time-step $\delta t$. However, it is possible that a different better minimum emerges in a different region of the landscape~\cite{campos2021abrupt}. That is, it is possible for the best minimum (or, more modestly, simply a significantly better minimum) to \textit{jump} from the initialization region to another region of the parameter landscape. As we only have lower bounds on the variance of the loss and convexity guarantees in the region around the initialization if the minimum jumps then we have no trainability guarantees to these superior minima. Moreover, if the full landscape has a barren plateau, which will be the case for most deep ans\"{a}tze~\cite{ragone2023unified, fontana2023theadjoint}, it may be very hard to train to this new minimum.

In Fig.~\ref{fig:minimumjump} we suggest that such apparent \textit{minimum jumps} can indeed occur. In particular we show a 1D cut of the landscape $\mathcal{L}(\thv)$ for different time-steps $\delta t$. The 1D cut includes both the `old minimum' at time $\delta t = 0$ and a new minimum that emerges for larger $\delta t$. Even after a short time-step $\delta t = 0.04$ the best minimum has jumped by a distance $|| \vec{\theta} ||_\infty \approx 0.8$. At this short time, the new minimum is only very slightly superior to the adiabatic minimum. However, at longer times the new minimum becomes substantially better. 

\begin{figure}
    \includegraphics[width=1.0\linewidth]{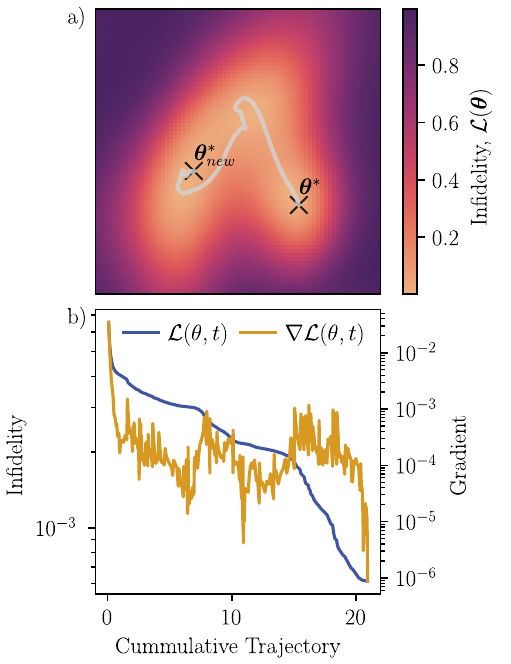}
    \caption{\textbf{Fertile valley.}
    a) Here we show a 2D plot of the loss landscape at $\delta t = 0.04$ for a 10 qubit Ising Hamiltonian $H=\sum X_iX_{i+1} - 0.95 \sum Y_i$ on a 1D-lattice and use a 2-layered Hamiltonian Variational Ansatz. $\vec{\theta_0}$ is the initial starting point and $\vec{\theta^*}$ is the true global minimum. The axes are chosen using principle component analysis to project the multi-dimensional space into a 2D-plane using ORQVIZ~\cite{rudolph2021orqviz}  and the white line is the projection of the optimization trajectory onto this 2D-plane. b) We plot the loss and directional loss gradient along the trajectory from the old to new minimum.} 
    \label{fig:fertile-valley-b}
\end{figure}

When a minimum jumps our theoretic guarantees developed in this manuscript lose most of their value. However, this does not mean that it is not possible to train (even in the case where the overall landscape exhibits a barren plateau). For training the `only thing' we need is a gradient flow, i.e., a path with substantial gradients, from the initialization to the new minimum. Such fertile valleys with nice gradients can theoretically exist on a barren plateau landscape but to what extent they occur in practise is currently unknown. 

In Fig.~\ref{fig:fertile-valley-b} we provide numerical evidence for a toy example of such a case. Specifically, we show a 2D cross-section of the landscape containing both an initialisation minimum and an apparently jumped minimum (marked by black crosses). We managed to successfully train from this initial minimum to the new minimum using the BFGS algorithm. This algorithm is a non-stochastic algorithm and so this indicates that there is indeed a trajectory between the two minima. As shown here in b) the gradients along this trajectory are of the order $10^{-3}$, in contrast to of the order of $10^{-6}$ on average over the landscape for a 10 qubit problem as shown in Fig.~\ref{fig:variance}. Thus, while a significant shot budget ($\sim 10^6$ shots) is likely needed for training it would seem that it is possible to train between these two minima without crossing into the most barren parts of the landscape.

The discussion in this section is necessarily heuristic. In our numerical investigations we found some minima jumps that we could train between (indicating a fertile valley) and other minima jumps where we could not. In the latter case there may or may not be fertile valleys. In both cases this is evidence for toy problems and at a small problem sizes (10 qubits). To what extent these phenomena occur at larger problem sizes, for more interesting problems and for relevant time-step sizes, remains entirely open.

\revadd{
\subsection{Outlook on other iterative approaches}\label{sec:extension-to-other-iterative}
So far, we have presented our results in the context of iteratively learning a circuit that implements $e^{-iHt}|\psi_0\rangle$. 
However, the same proof steps to guarantee substantial gradients can be re-interpreted and abstractified for other iterative approaches with a fidelity-type loss.

To intuitively see this, we note that core to Theorem~\ref{thm:variance-lower-bound} is the condition that an initial state of the current iteration (which is a solution of the previous iteration) $|\psi(\thv^*)\rangle = U(\thv^*) |\psi_0\rangle$ has a large overlap with the target state $e^{-iH\delta t}|\psi_0\rangle$. Indeed, precisely at this step, one could abstractify the setting beyond the variational quantum simulation. In particular, a target state for a given iteration need not represent a time-evolved state i.e., $e^{-iH\delta t} |\psi_0\rangle \rightarrow |\psi_{\rm target}\rangle$, and as long as the overlap between between the target state and the state initialized around the solution of the previous iteration is large, the region of attraction is theoretically guaranteed.

More precisely, consider any iterative method with a loss function expressed as
\begin{align}\label{eq:loss-extension-fidelity-type}
    \LC(\thv) = 1 - |\langle \psi_0 |U^\dagger(\thv) |\psi_{\rm target} \rangle|^2 \;,
\end{align}
where $|\psi_{\rm target}\rangle$ is a target state for the current iteration and $\thv^*$ is a set of parameters obtained from the previous iteration. Then, we have the following theorem regarding its attraction region. 
\begin{theorem}[A substantial gradient region, Informal]\label{th:substantial_gradient}
Consider any iterative method with the fidelity-type loss in Eq.~\eqref{eq:loss-extension-fidelity-type} and the general ansatz in Eq.~\eqref{eq:circuit} and 
the same setting for $\rho_0$ and $U(\thv)$ as in Theorem~\ref{thm:variance-lower-bound}. 
Given that the overlap between the target and the state around intialization with $\thv^*$ follows
\begin{align}
     \left| \langle \psi(\thv^*) | \psi_{\rm target}\rangle \right|^2 \geq \frac{1}{2} \;,
\end{align}
and we consider uniformly sampling parameters in a hypercube of width $2r$ around the solution from the previous iteration $\thv^*$, i.e. $\vol(\thv^*,r)$, such that 
\begin{align}
 r \in \Theta\left(\frac{1}{\sqrt{M}} \right) \;. 
\end{align}
Then the variance at any iteration of the algorithm is lower bounded as 
\begin{align}
    \Var_{\vtheta \sim\uni(\vtheta^*, r)} \left[ \LC (\vtheta)\right] \in \Omega\left( \frac{1}{M}\right) \;.
\end{align}
Thus, for $M \in \OC(\poly(n))$ we have 
\begin{align}
    \Var_{\vtheta \sim\uni(\vtheta^*, r)} \left[ \LC (\vtheta)\right] \in \Omega\left( \frac{1}{\poly(n)}\right) \;.
\end{align}
\end{theorem}

Theorem~\ref{th:substantial_gradient} can be applied to a wide range of variational tasks including: i. simulating \textit{imaginary time evolution} for ground or thermal state preparation (see Appendix~\ref{app:imaginary}) and ii. \textit{learning an unknown target unitary}, where one learns the unitary itself rather than its effect on a fixed initial state ( Appendix~\ref{app:extension}). 
Similarly to the constraint on $\delta t$ in Theorem~\ref{thm:variance-lower-bound}, the large fidelity condition in Theorem~\ref{th:substantial_gradient} is a sufficient condition to have such a guarantee, but it may not be necessary. This large fidelity condition is satisfied for dynamics learning by making $\delta t$ small. How to satisfy it in other iterative methods will vary on a case-by-case basis. 
We note that while we expect the theoretical results for convexity and adiabatic minimum to also hold generally, this is challenging to rigorously show without further considering specific details of the iterative approach that we are interested in.

Lastly, while we have focused on fidelity-type losses, we believe our approach can be extended to other losses composed of the expectation of a generic observable $O$. This is relevant for several other tasks, including finding a ground state of some Hamiltonian via a perturbative approach. For example, one could consider a perturbative Hamiltonian of the form $H(\alpha) = (1-\alpha) H_0 + \alpha H_{\rm target}$ where $\alpha \in [0,1]$ is a perturbation with $H_0$ being a Hamiltonian that has an easier-to-prepare ground state and $H_{\rm target}$ being the Hamiltonian one wants to learn the ground state of. Then, one could imagine starting with a circuit for preparing the ground state of $H_0$ and iteratively increasing the perturbation $\alpha$ bit by bit and applying the variational quantum eigensolver between each iteration~\cite{PRXQuantum.2.020329,10.1145/3479197}. 
If the perturbations do not pass through a phase transition then such an iterative scheme is plausible and could potentially be characterised in a similar manner to as we have done here. In particular, one would need a significant overlap between an initial state and a ground state of $H(\alpha)$ to ensure substantial gradients in the initialization region at each iteration. Carrying out these calculations is beyond scope of this work and we leave them for future investigation.}

\section{Discussion}

Thanks to significant progress in recent years, the barren plateau phenomenon, defined as an average statement for an entire loss landscape, is by now technically well understood~\cite{ragone2023unified, fontana2023theadjoint}. However, prior analyses are consistent with different accounts of the behavior of the loss landscape in the subregions most important for optimization. In this work we have taken steps to address these open questions by investigating a popular iterative variational circuit compression scheme~\cite{otten2019noise, benedetti2020hardware, barison2021efficient, lin2021real, berthusen2022quantum, haug2021optimal, gentinetta2023overhead}. The iterative nature of this algorithm ensures that variational problem is repeatedly warm started at each iteration of the variational scheme.

Theorem~\ref{thm:variance-lower-bound} establishes that for short enough time-steps the loss variance is guaranteed to decrease at worst polynomially in the number of parameters $M$ in a hypercube with a width that scales as $1/\sqrt{M}$ around the new initialization. Theorem~\ref{thm:convex} strengthens this result by arguing that in a region $\sim 1/M^2$ around the initialization the landscape will be approximately convex. Finally, we sew together these results with a bound on the distance the adiabatic minimum (Definition~\ref{def:adiabaticminimum}) can move after applying a \revadd{time-}step of length $\delta t$. Thus in Theorem~\ref{thm:adiabaticminimum} we establish that as long as the time-step is decreased polynomially with the number of trainable parameters in the ansatz the adiabatic minimum remains in the approximately convex region with substantial gradients. Hence we show that by decreasing the time-step appropriately one should be able to train to a new minimum. 

Our analysis leaves room for further research opportunities. For one, the analytic bounds provided here are lower bounds. We do not here provide upper bounds. Thus our analysis leaves open the question of whether the region exhibiting polynomial gradients strictly decreases as $1/\sqrt{M}$ or whether potentially a larger region exhibits substantial gradients. Our numerical implementations (Fig.~\ref{fig:variance}) suggest that for the problems we have looked at this $1/\sqrt{M}$ is reasonable. However, analytic upper bounds to verify this would be more satisfying. 

Moreover, whether these bounds are to be viewed positively or negatively remains open. While in `complexity-theory-land' polynomial guarantees are typically satisfactory, in practise polynomially vanishing gradients, in polynomially shrinking regions, with polynomially shrinking step sizes may not be that appealing. In particular, the $\delta t$ values that enjoy guarantees via Theorem~\ref{thm:adiabaticminimum} are typically smaller than those used currently by practitioners for the small problem sizes accessible currently. To what extent these bounds can be tightened versus to what extent they indicate a fundamental limitation remains to be seen. Indeed, there is always the possibility that heuristically the optimization turns out to be more effective than analytic guarantees would suggest (as is typically the case for optimizing classical machine learning models). 

Here we have pushed our analysis beyond a conventional average case analyses for the full loss landscape. However, our analysis is still fundamentally an average case analysis within a hypercube around an initialization. The limitations of this are highlighted by our inability to analytically describe the minimum jumps and fertile valleys that we numerically observe in Fig.~\ref{fig:minimumjump} and Fig.~\ref{fig:fertile-valley-b}. To analytically study such phenomena new theoretical tools will need to be developed to analyse quantum landscapes. 

We remark that recent work has highlighted a strong link between provable absence of barren plateaus and the classical simulability and surrogatability of the hybrid optimisation loop of a variational quantum algorithm~\cite{cerezo2023does, bermejo2024quantum, angrisani2024classically,lerch2024efficient}. The lower bounds obtained here are consistent with these claims~\cite{lerch2024efficient}. In particular, for classically simulable initial states one could perform early iterations fully classically and then later iterations by collecting data from quantum computer and then training a classical surrogate of the landscape. We leave a discussion of the relative merits of this approach to future work.

\revadd{Finally, despite largely discussed in the context of an iterative variational scheme for quantum dynamics with a fixed initial state, most of our results are believed to be applicable more generally to other iterative methods. Intuitively, the substantial gradient region is a consequence of initializing in a state that has a large overlap with a target state. While we illustrate this conjecture by discussing the result for any iterative setting with a fidelity-like loss, such as learning the entire unknown unitary, there remain many open questions. 
One of which, for example, is to extend this to an iterative method with other loss types such as for learning the ground state energy of some Hamiltonian. More generally, it would also be interesting to explore guarantees beyond iterative approaches for other warm-start intialization strategies.}

\section{Acknowledgments}
We thank Stefan Woerner, Christa Zoufal and Matthis Lehmkuehler for insightful conversations.  We further thank Manuel Rudolph for helpful comments on a draft and his ORQVIZ wizardry. R.P. and M.D. acknowledge the support of the SNF Quantum Flagship Replacement Scheme (grant No. 215933). S.T. and Z.H. acknowledge support from the Sandoz Family Foundation-Monique de Meuron program for Academic Promotion. 
\section{Code Availability}
All plots in this paper were simulated using Qiskit \cite{qiskit}  and the code is available in \url{https://github.com/MarcDrudis/WarmStartCaseStudy}

\bibliography{quantum, bibliography}

\onecolumngrid

\newpage
\appendix

\section{Preliminaries}

In this section, we briefly review some analytical tools and concepts that will be used through out the other sections. 
\subsection{Relation between Hessian of the loss function and quantum Fisher information}

Given the optimal parameters obtained from the previous iteration $\thv^*$ and a time-step $\delta t$ of the Hamiltonian $H$, the loss function at the current iteration is of the form
\begin{align}
    \LC (\thv) & =  1 - \left| \langle \psi_0 | U^\dagger(\vtheta) e^{-iH \delta t} U(\thv^*) |\psi_0  \rangle\right|^2 \\
    & = 1 - F\left( U(\thv) |\psi_0\rangle, e^{-iH \delta t} U(\thv^*) |\psi_0  \rangle \right)
\end{align}
where $F\left(|\psi\rangle, |\phi\rangle\right)$ is a fidelity between two pure states $|\psi\rangle$ and $|\phi\rangle$. We remark that although $\delta t$ is often taken as fixed and not optimised during the training process, the loss function also implicitly depends on $\delta t$. 

The warm-start strategy is to initialise the training of the current iteration around $\thv^*$. To analyse the trainability of this strategy, we often consider the expansion of the loss around $\thv^*$ and $\delta t = 0$. In this context, it is convenient to write $\vec{x} = (\thv - \thv^*, \delta t)$ and $F(\vec{x}) := F\left( U(\thv) |\psi_0\rangle, e^{-iH \delta t} U(\thv^*) |\psi_0  \rangle \right)$. Upon expanding the loss around $\vec{x} = \vec{0}$, the connection between the Hessian of the loss function and the quantum fisher information is    
\begin{align}
    \nabla^2_{\vec{x}} \LC(\vec{x})\; \big|_{\vec{x}= \vec{0}}= - \nabla^2 F(\vec{x})\; \big|_{\vec{x}= \vec{0}} =  \frac{1}{2} \FC(\vec{0}) \;,
\end{align}
where $\FC(\vec{0})$ is the quantum fisher information evaluated at $\vec{x} = \vec{0}$ and measures how the quantum state $U(\thv^*) |\psi_0  \rangle$ is sensitive to local perturbations around $\thv^*$ and $\delta t =0$ ~\cite{appconvex_fisherproperties, larocca2021theory, Toth_2014}.

\subsection{Taylor remainder theorem}\label{app:taylor}

We present the Taylor remainder theorem which expresses a multivariate differentiable function as a series expansion. We refer the reader to Ref.~\cite{taylorbook} for further details.    

\begin{theorem}[Taylor reminder theorem] \label{thm:taylor}
Consider a multivariate differentiable function $f(\vec{x})$ such that $f: \mathbb{R}^N \rightarrow \mathbb{R}$ and some positive integer $K$. The function $f(\vec{x})$ can be expanded around some fixed point $\vec{a}$ as
\begin{align}
    f(\vec{x}) = \sum_{k = 0}^{K} \sum_{i_1, i_2, ...,i_k}^N \frac{1}{k!} \left( \frac{\partial^k  f(\vec{x})}{\partial x_{i_1} \partial x_{i_2} ... \partial x_{i_k}} \right)\bigg|_{\vec{x} = \vec{a}} (x_{i_1} - a_{i_1})(x_{i_2} - a_{i_2}) ... (x_{i_k} - a_{i_k}) + R_{K,\vec{a}}(\vec{x}) \;,
\end{align}
where the remainder is of the form
\begin{align}
     R_{K,\vec{a}}(\vec{x}) = \sum_{i_1, i_2, ...,i_{K+1}}^N \frac{1}{(K+1)!} \left( \frac{\partial^{K+1}  f(\vec{x})}{\partial x_{i_1} \partial x_{i_2} ... \partial x_{i_{K+1}}} \right)\bigg|_{\vec{x} = \vec{\nu}} (x_{i_1} - a_{i_1})(x_{i_2} - a_{i_2}) ... (x_{i_{K+1}} - a_{i_{K+1}}) \;,
\end{align}
with $\vec{\nu} = c \vec{x} + (1-c) \vec{a}$ for some $c \in [0,1]$. 
\end{theorem}

\medskip

As an example, we apply the Taylor remainder theorem to prove the following statement.
\begin{lemma}\label{lem:bound_fidelity}
The fidelity between two pure states $\rho$ and $ e^{-i H t  }\rho e^{i H  t }$ (with the Hamiltonian $H$) can be upper bounded as
\begin{equation}
    F \left(\rho, e^{-i  t H }\rho e^{i t H}\right) \geq  1 - 2 \lambda_{\rm max}^2 t^2 
\end{equation}

where $\lambda_{\rm max}$ is the largest eigenvalue of $H$. 
\end{lemma}
\begin{proof} 
First, we denote $F(t) := F \left(\rho, e^{-i t H }\rho e^{i  t H}\right)$. By using Theorem~\ref{thm:taylor} (expanding around $t = 0$ up to the second order), the fidelity is of the form 
\begin{align}
    F(t) = 1 + \frac{t^2}{2} \left(\frac{d^2 F(t)}{dt^2}\right) \bigg|_{t = \tau}  \;,
\end{align}
where the zero order term is $1$, the first order term is zero by a direct computation and the second order term is evaluated at some $\tau \in [0, t]$. We can then bound the second derivative as the following
\begin{align}
    \left(\frac{d^2 F(t)}{dt^2}\right) \bigg|_{t = \tau}& = \Tr\left(\rho  e^{- i H \tau} i \left[ i \left[ \rho, H \right],H \right]e^{iH\tau} \right) \\
    & \leq \| \rho  \|_1 \| e^{- i H \tau} i \left[ i \left[ \rho, H \right],H \right]e^{iH\tau} \|_\infty \\
    & \leq 4 \lambda_{\rm max}^2 \;,
\end{align}
where the first inequality is due to H\"{o}lder's inequality. In the second inequality, we rely on the following identities: (i) $\| \rho \|_1 = 1$ for a pure state, (ii) the unitary invariance of the Schatten p-norm i.e., $\| U A\|_p = \| A\|_p$ for any unitary $U$, (iii) $\| i [A,B] \|_p \leq 2 \|A\|_p \|B\|_p$ and lastly (iv) $\|AB\|_p \leq \|A\|_p \|B\|_p$. 
Thus, the fidelity can be lower bounded as
\begin{align}
    F(t) \geq 1 - 2 \lambda_{\rm max}^2 t^2 \;.
\end{align}
This completes the proof.
\end{proof}

\subsection{Approximate convexity}\label{app:epsilonconvexity}
In this section, we provide a formal explanation of our definition of an $\epsilon$-convex function. We start by defining what convexity is (see for example Ref.~\cite{convexity}) and we relate it to our notion of $\epsilon$-convexity. 

\begin{definition}[Convexity]\label{def:convex}
A differentiable function of several variables $f: \mathbb{R}^N\to \mathbb{R}$ is convex in a region $ \mathcal{R}$,  if and only if for all $\vec{x},\vec{y}\in\mathcal{R}$ the function fulfils
 \begin{equation}
      f(\vec{x}) \geq f(\vec{y}) + \nabla f(\vec{y})\cdot(\vec{x}-\vec{y}) \ .
 \end{equation}
or equivalently, $\nabla^2 f(\vec{x})$ is positive semi-definite. Here $\nabla^2 {f}(\vec{x})$ denotes the Hessian of $f(\vec{x})$. Notice that we use $\nabla f(\vec{y}) = \nabla_{\Tilde{\vec{y}}} f(\Tilde{\vec{y}})|_{\Tilde{\vec{y}} = {\vec{y}}}\ , \nabla^2f(\vec{y}) =\nabla^2_{\Tilde{\vec{y}}} f(\Tilde{\vec{y}})|_{\Tilde{\vec{y}} = {\vec{y}}}$
\end{definition}

Informally this means that all the tangent planes to $f$ are below $f$ in the region $\mathcal{R}$. This is shown for one variable in Fig. \ref{fig:eps-convex}~a). 

\setcounter{definition}{0}
\begin{definition}[$\epsilon$-approximate convexity.]\label{def:epsilon-convex-app}
 A differentiable function of several variables $f: \mathbb{R}^N\to \mathbb{R}$ is $\epsilon$-convex in a region  $ \mathcal{R}$ if 
 \begin{equation}
     \left[\nabla^2 f(\vec{x})\right]_{\rm min} \geq -|\epsilon| \, 
 \end{equation}
for all $\vec{x} \in \mathcal{R}$. Here $\nabla^2 {f}(\vec{x})$ denotes the Hessian of $f(\vec{x})$ and we denote $\left[ A \right]_{\rm min}$ as the smallest eigenvalue of the matrix $A$.
\end{definition}

If a function is $\epsilon$-convex in a finite region $\mathcal{R}$, then we can show an equivalent intuition to the one for convexity. Indeed, if a function $f$ is $\epsilon$-convex in a finite region $\mathcal{R}$, then we can say that an ``$\epsilon$-displacement'' in every tangent line is enough to make it below $f$ in $\mathcal{R}$ as shown in Fig. \ref{fig:eps-convex}~b).

\begin{proposition}
  If a differentiable function of several variables $f: \mathbb{R}^n\to \mathbb{R}$ is $\epsilon$-convex in finite a region $\mathcal{R}$, then for all $ \vec{x},\vec{y}\in\mathcal{R}$ the function fulfils 
   \begin{equation}
         f(\vec{x}) \geq f(\vec{y}) + \nabla f(\vec{y})\cdot(\vec{x}-\vec{y}) - |\epsilon| \alpha  \ ,
   \end{equation}
    where $\alpha = \frac{1}{2}\max_{\vec{a}, \vec{b}\in \mathcal{R}} ||\vec{a}-\vec{b}||_2^2 $.
\end{proposition}

\begin{proof}
    First we recall the Taylor reminder theorem in \ref{app:taylor}. Then we can expand the function $f$ around the point $\vec{y}$ as
    \begin{align}
        f(\vec{x}) =& f(\vec{y}) +\sum_i^n\frac{\partial f(\vec{q})}{\partial q_i }\bigg|_{\vec{q} = \vec{y}} (x_i-y_i) + \frac{1}{2}\sum_{i,j}^n\frac{\partial^2 f(\vec{q})}{\partial q_i\partial q_j}\bigg|_{\vec{q} = \vec{z}} (x_j-y_j)(x_i-y_i)\\
        =&f(\vec{y}) + \nabla f(\vec{y}) \cdot (\vec{x}-\vec{y}) + \frac{1}{2}(\vec{x}-\vec{y})^T \nabla^2 f(\vec{z})(\vec{x}-\vec{y}) \ ,
    \end{align}
    for some $\vec{z} = c\vec{x} + (1-c)\vec{y}$ with some $c \in [0,1]$. 
    In the last equality we wrote the expression in its vector form for simplicity.

     Now if we apply the notion of $\epsilon$-convexity we can lower-bound the right-hand side of the previous equality to find
     \begin{equation}
        f(\vec{x}) \geq f(\vec{y}) + \nabla f(\vec{y}) \cdot (\vec{x}-\vec{y}) - \frac{1}{2}|\epsilon|(\vec{x}-\vec{y})^T (\vec{x}-\vec{y}) \ ,
    \end{equation}
    which can be further bounded by using that $(\vec{x}-\vec{y})^T (\vec{x}-\vec{y}) \leq \alpha$. Recall that $\alpha = \frac{1}{2}\max_{\vec{a}, \vec{b}\in \mathcal{R}} ||\vec{a}-\vec{b}||_2^2 $. With this we find
    \begin{equation}
         f(\vec{x}) \geq f(\vec{y}) + \nabla f(\vec{y})\cdot(\vec{x}-\vec{y}) - |\epsilon| \alpha \ ,
   \end{equation}
\end{proof}

Notice that in general $\alpha$ can increase with the dimension of the input $\vec{x}$ (for the loss function in Eq \eqref{eq:loss} this is equivalent to the number of parameters $M$) as well as with the size of the region that we demand $\epsilon$-convexity over. That is, for a region of size $r$, $\alpha$ approximately scales as $Mr^2$. In the regime that is relevant to our work (particularly, Theorem \ref{th:pRTE}), the region in which we require $\epsilon$-convexity is of order $r\in\Omega\left(1/M^2\right)$ (for appropriately chosen $\delta t$). Hence, $\alpha$ decays with the number of parameters as $1/M^2$. 

\begin{figure}
    \centering
    \includegraphics[width=.75\columnwidth,trim={1cm 0 15cm 0},clip]{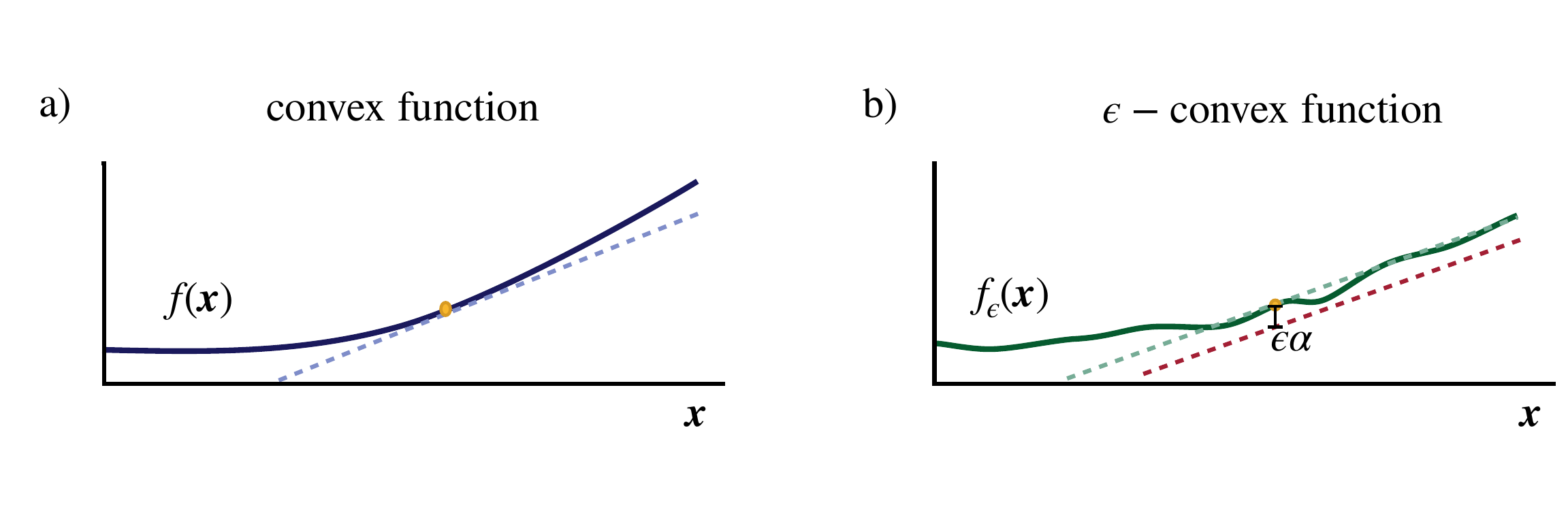}
    \caption{\textbf{Convexity and $\bm{\epsilon}$-convexity.} Here we show two schematics for one variable functions. a) Represents a convex function and the tangent (dashed line). The tangent is always below $f$. b) Represents a $\epsilon$-convex function and the tangent (dashed line). The light-green dashed line represents the tangent to the function. The red line represents the tangent displaced $\alpha\epsilon$ with respect to the green one. A ``displacement of $\alpha\epsilon$'' on the tangent makes it such that is always below the function. }
    \label{fig:eps-convex}
\end{figure}

\subsection{Upper bound on the eigenvalues}

We present a simplified version of Gershgorin's circle theorem~\cite{circle} which can be used to upper bound the eigenvalues of a real squared matrix. 

\begin{proposition}\label{prop:upper-eigen}
Consider a real $M\times M$ matrix $A$ and denote $\lambda_{\rm max}$ as the largest eigenvalue of $A$. Given that the sum of elements in any row is upper bounded by some value $T_0$, that is $\sum_{j} |A_{ij}| \leq T_0$ for any $i \in \{1,...,M\}$, the largest eigenvalue of $A$ can be bounded as 
\begin{align}
    \lambda_{\rm max} \leq T_0 \;\;.
\end{align}
\end{proposition}
\begin{proof} Let $\vec{v}$ be an eigenvector corresponding the largest eigenvalue of $A$. There exists the largest component in the eigenvector denoted as $v_i$. Then, we consider this eigenvector component in the eigenvalue equation
\begin{align}
    \sum_{j} A_{ij} v_j = \lambda_{\rm max} v_i \;.
\end{align}
The bound of the largest eigenvalue follows as
\begin{align}
    \lambda_{\rm max} & \leq \sum_{j}  \frac{|A_{i,j}|\cdot|v_j|}{|v_i|}  \\
    & \leq \sum_{j} |A_{i,j}| \\
    & \leq T_0 \;,
\end{align}
where the first inequality is due to triangle inequality, the second inequality is from $\frac{|v_j|}{|v_i|} \leq 1$ and in the last inequality we use the assumption that $\sum_{j} |A_{i,j}| \leq T_0 $.
\end{proof}

\section{Lower bound of the variance of the loss function}\label{app:variance}
In this section, we provide the exact formula for the lower bound of the loss function for the variational \revadd{time-evolution} compression algorithm.

\subsection{Exact formula for the lower bound} 

\begin{proposition}\label{prop:variance-lower-bound}
Consider the loss function $\mf(\thv)$ as defined in Eq.~\eqref{eq:loss} and with an ansatz of the general form defined in Eq.~\eqref{eq:circuit} with $M$ parameters. The variance of $\mf(\thv)$ over the hypercube parameter region $\vol(\thv^*, r)$ around an optimal solution of the previous iteration $\thv^*$ can be bounded as 
	\begin{equation} \label{eq:prop1-exact-var}
		\Var_{\thv\sim\uni(\thv^*, r) }[\mf(\thv)]\geq  \ (c_+ - k_+^2) \min_{\Tilde{\xi}\in [-1,1]} \left( k_+^{M-1} \Delta_{\thv^*}  + (1-k_+^{M-1})\Tilde{\xi} \right)^2 ,
	\end{equation}
	where we have
	\begin{align}
		&c_+     :=  \mathbb{E}_{\alpha \sim\uni(0,r)} [ \cos^4{\alpha}]\;, \\           
		&k_+   := \mathbb{E}_{\alpha \sim\uni(0,r)} [ \cos^2{\alpha}] \; ,                                                               \\
		&\Delta_{\thv^*} := \Tr[(\rho_0 -\vsigma_1 \rho_0 \vsigma_1) U^{\dagger}\left(\vec{\theta^*}\right)\rho_{(\vtheta^*,\delta t)} U\left(\vec{\theta^*}\right)] \; .\label{eq:var-delta-key}
	\end{align}
	Here $\vsigma_1$ is the Pauli string associated with the first gate in the circuit $U(\thv)$ as defined in Eq.~\eqref{eq:circuit}, $\rho_0 = |\psi_0 \rangle\langle\psi_0|$ is an initial state before the time evolution and $\rho_{(\vtheta^*,\delta t)}= e^{-iH\delta t} U(\vec{\theta^*})\rho_0 U^\dagger(\vec{\theta^*})e^{iH\delta t}$ with $H$ being the underlying Hamiltonian of the quantum dynamics. 
\end{proposition}

\begin{proof}
First, we recall that the loss function at each iteration (as defined in Eq.~\eqref{eq:loss}) is of the form
\begin{align}
    \mf (\vtheta) & =  1 - \left| \langle \psi_0 | U^\dagger(\vtheta) |\psi(\vtheta^* , \delta t)  \rangle\right|^2 \\
    & = 1 - \langle \psi_0 | U^\dagger(\vtheta) \rho_{(\vtheta^*,\delta t)} U(\vtheta) |\psi_0 \rangle \;
\end{align}
for some initial state $\ket{\psi_0}$ and
\begin{align}
    \rho_{(\vtheta^*,\delta t)} : = |\psi(\vtheta^* , \delta t) \rangle  \langle\psi(\vtheta^* , \delta t)  | = e^{- i H \delta t} U(\vtheta^*) | \psi_0 \rangle\langle\psi_0 | U^\dagger(\vtheta^*)  e^{i H \delta t}  \;,
\end{align}
where $\thv^*$ is an optimal solution of the previous iteration. The parameterised quantum circuit $U(\vtheta)$ with $M$ parameters takes the following general form
\begin{equation}
    U\left( \vtheta \right) = \prod_{i=1}^M V_i U_i(\theta_i) \;,
\end{equation}
where $\{ V_i \}_{i=1}^M$ are some fixed unitaries and $\{ U_i(\theta_i) = e^{- i \theta_i \vsigma_i }\}_{i=1}^M$ are a set of parameterised rotation gates with $\vsigma_i$ being a Pauli string associated with the $i^{\rm th}$ gate. Crucially, the rotation gates can be re-expressed as perturbations $\vec{\alpha}$ around the previous optimal solution i.e., $\theta_i = \theta_i^* + \alpha_i$ for all $i$
\begin{align}
     U\left( \vtheta \right) 
     & = \prod_{i=1}^M V_i U_i(\theta^*_i) U_i (\alpha_i) \\
     & = \prod_{i=1}^M \widetilde{V}_i (\theta^*_i) U_i (\alpha_i) \;,
\end{align}
where the first equality holds due to $e^{-i \theta_i \vsigma_i} = e^{- i \theta^*_i \vsigma_i} e^{- i \alpha_i \vsigma_i}$ and in the second equality we denote $\widetilde{V}_i := \widetilde{V}_i (\theta_i^*) = V_i e^{- i \theta_i^* \vsigma_i }$. 

\medskip

We consider the region of parameters around the previous optimum which can also be expressed in terms of $\Vec{\alpha}$
\begin{align}
    \vol(\vtheta^*, r) = \{\vtheta = \vtheta^* + \valpha \;|\;\; \alpha_i \in  [- r, r] \} \;,
\end{align}
where $r$ is a characteristic length of the region. Now, we are interested in the variance of the loss function over $\vol(\vtheta^*, r)$ such that each parameter is uniformly sampled
\begin{align}
    \Var_{\vtheta \sim\uni(\vtheta^*, r)} \left[ \LC (\vtheta)\right] &= \Var_{\vec{\alpha} \sim\uni(\vec{0}, r)} \left[ \LC (\vtheta = \vtheta^* + \vec{\alpha})\right] \\
    & = \Var_{\vec{\alpha} \sim\uni(\vec{0}, r)} \left[ 1 - \LC (\vtheta)\right] \\
    & = \Ebb_{\valpha \sim\uni(\vec{0},r)}\left[ (1 - \LC(\vtheta))^2\right] - \left( \Ebb_{\valpha \sim\uni(\vec{0}, r)} [1 - \LC(\vtheta)]\right)^2 \; , \label{eq:thm1-proof-var}
\end{align}
where the second equality is due to $\Var_{\valpha} [b_1X(\valpha) + b_2] = (b_1)^2\Var_{\valpha} X(\valpha)$ for some constants $b_1$ and $b_2$.

Importantly, since all parameters are assumed to be uncorrelated, this allows us to compute the variance over each individual parameter one-by-one from the outermost parameter $\alpha_M$ towards the first parameter $\alpha_1$. That is, each term in Eq.~\eqref{eq:thm1-proof-var} can be expressed as 
\begin{align}
    \Ebb_{\valpha \sim\uni(\vec{0},r)}\left[ (1 - \LC(\vtheta))^2\right] &= \Ebb_{\alpha_1, \alpha_2, ..., \alpha_{M}}\left[ (1 - \LC(\vtheta))^2\right]  \\
    &= \Ebb_{\alpha_1, \alpha_2, ..., \alpha_{M-1}} \Ebb_{\alpha_{M}}\left[ (1 - \LC(\vtheta))^2\right] \\
    &:= \Ebb_{\overline{\alpha_M}}\Ebb_{\alpha_{M}}\left[ (1 - \LC(\vtheta))^2\right] \;, \label{eq:thm1-proof-avg-square-loss}
\end{align}
with $\overline{\alpha_M} := \alpha_1, \alpha_2,...,\alpha_{M-1}$ and, similarly,
\begin{align}
    \Ebb_{\valpha \sim\uni(\vec{0},r)}\left[ (1 - \LC(\vtheta))\right] =&\;  \Ebb_{\alpha_1, \alpha_2, ..., \alpha_{M-1}} \Ebb_{\alpha_{M}}\left[ (1 - \LC(\vtheta))\right] \\ 
    := &\; \Ebb_{\overline{\alpha_M}}\Ebb_{\alpha_{M}}\left[ (1 - \LC(\vtheta))\right]\;. \label{eq:thm1-proof-avg-loss}
\end{align}
Before delving into computing these terms, we first stress the  $\alpha_M$ dependence of the loss by writing 
\begin{align}
    1 - \LC(\vtheta) =& \; \left| \langle \psi_0 | U^\dagger(\vtheta) |\psi(\vtheta^* , \delta t)  \rangle\right|^2 \\
    = &\; \langle \psi_{M-1} | U^\dagger_M(\alpha_M) \rho_M U_M(\alpha_M) |\psi_{M-1} \rangle \, ,
\end{align}
where we have defined
\begin{align}
    |\psi_{M-1} \rangle & :=  \prod_{i=1}^{M-1} \widetilde{V}_{i}(\theta^*_i) U_i(\alpha_i) |\psi_0 \rangle \;, \\
    \rho_M & :=  \widetilde{V}^\dagger_M (\theta^*_M) \rho_{(\vtheta^*,\delta t)} \widetilde{V}_M(\theta_M^*) \;.
\end{align}
Remark that $\ket{\psi_{M-1}}$ depends on the other parameters $\{ \alpha_i \}_{i=1}^{M-1}$ while $\rho_M$ is independent of $\valpha$. Next we use the identity 
\begin{align}
    U_i(\alpha_i)  = \cos (\alpha_i) \1 - i \sin (\alpha_i) \vsigma_i \;,
\end{align}
to rewrite the loss as
\begin{align}
    1 - \LC(\vtheta) 
    =& \;\cos^2 (\alpha_M)  \langle \psi_{M-1} |\rho_M|\psi_{M-1} \rangle + \sin^2 (\alpha_M)  \langle \psi_{M-1} |\vsigma_M \rho_M \vsigma_M |\psi_{M-1} \rangle \nonumber \\
    &\; - \cos (\alpha_M) \sin (\alpha_M)  \langle \psi_{M-1} | i [\rho_M, \vsigma_M]|\psi_{M-1} \rangle \\
    = &\; \cos^2 (\alpha_M) \langle \rho_M \rangle_{ \psi_{M-1}} +  \sin^2 (\alpha_M) \langle \vsigma_M \rho_M \vsigma_M  \rangle_{\psi_{M-1}} - \cos (\alpha_M) \sin (\alpha_M) \langle i [\rho_M, \vsigma_M]  \rangle_{\psi_{M-1}} \, ,
\end{align}
where in the final line we use the shorthand
\begin{align}
    \langle O \rangle_{\psi} & := \langle \psi | O |\psi\rangle \;,
\end{align}
for some observable $O$ and some state $|\psi \rangle$. 

\medskip

We are now ready to proceed with the averaging over $\alpha_M$ in Eq.~\eqref{eq:thm1-proof-avg-square-loss} which results in
\begin{align}
    \Ebb_{\alpha_{M}}\left[ (1 - \LC(\vtheta))^2\right] 
    = &\;  c_+ \langle  \rho_M \rangle_{\psi_{M-1}}^2 + c_-  \langle\vsigma_M \rho_M \vsigma_M \rangle^2_{\psi_{M-1}} + c_0 \langle  i [\rho_M, \vsigma_M] \rangle_{\psi_{M-1}}^2 \label{eq:thm1-proof00} \\
    &\; + 2c_0 \langle \rho_M \rangle_{\psi_{M-1}} \langle \vsigma_{M} \rho_M \vsigma_M \rangle_{\psi_{M-1}} \nonumber \\
    \geq &\; c_+  \langle  \rho_M \rangle_{\psi_{M-1}}^2 + c_-  \langle\vsigma_M \rho_M \vsigma_M \rangle^2_{\psi_{M-1}} + 2c_0 \langle \rho_M \rangle_{\psi_{M-1}} \langle \vsigma_{M} \rho_M \vsigma_M \rangle_{\psi_{M-1}} \;, \label{eq:thm1-proof1}
\end{align}
where we have
\begin{align}
    c_+ & = \frac{1}{2r} \int_{-r}^{r} d\alpha_M \cos^4(\alpha_M) \;, \\
    c_- & = \frac{1}{2r} \int_{-r}^{r} d\alpha_M \sin^4(\alpha_M) \;, \\
    c_0 & =  \frac{1}{2r} \int_{-r}^{r} d\alpha_M \cos^2(\alpha_M) \sin^2(\alpha_M) \;, \\
    0 & = \frac{1}{2r} \int_{-r}^{r} d\alpha_M \cos^3(\alpha_M) \sin(\alpha_M) =  \frac{1}{2r} \int_{-r}^{r} d\alpha_M \cos(\alpha_M) \sin^3(\alpha_M) \;.
\end{align}

\noindent Similarly, by considering Eq.~\eqref{eq:thm1-proof-avg-loss}, we have
\begin{align}
    \Ebb_{\alpha_{M}}\left[ 1 - \LC(\vtheta)\right] =& k_+  \langle \rho_M \rangle_{\psi_{M-1} } + k_-  \langle \vsigma_M \rho_M \vsigma_M \rangle_{\psi_{M-1} } \;,
\end{align}
with
\begin{align}
    k_+ & =  \frac{1}{2r} \int_{-r}^{r} d\alpha_M \cos^2(\alpha_M) \;, \label{eq:var-proof-k-plus}\\
    k_- & = \frac{1}{2r} \int_{-r}^{r} d\alpha_M \sin^2(\alpha_M) \;, \\
    0 & = \frac{1}{2r} \int_{-r}^{r} d\alpha_M \cos(\alpha_M)\sin(\alpha_M) \;. 
\end{align}
From the above expressions, we can see that $\Var_{\alpha}[\cos^2(\alpha)] = c_+ - k^2_+$, $\Var_{\alpha}[\sin^2(\alpha)] = c_- - k^2_-$ and $\Cov_{\alpha}[\cos^2(\alpha), \sin^2(\alpha)] = c_0 - k_+ k_- $. In addition, it can be verified by a direct computation that
\begin{align} \label{eq:thm1-proof-cos-sin-rel}
     c_+ - k^2_+ =  c_- - k^2_- = - (c_0 - k_+k_-) = \frac{-1 + 4r^2 + \cos(4r) +r \sin(4r)}{32r^2} \;.
\end{align}

\medskip

Together, the variance in Eq.~\eqref{eq:thm1-proof-var} can be bounded as
\begin{align}
     \Var_{\vtheta} \left[ \LC (\vtheta)\right] =&\; \Ebb_{\alpha_1,...,\alpha_{M-1}}  \Ebb_{\alpha_M}\left[ (1-\LC(\vtheta))^2\right] - \left(\Ebb_{\alpha_1,...,\alpha_{M-1}}  \Ebb_{\alpha_M}\left[ 1-\LC(\vtheta)\right]  \right)^2 \\
     \geq &\; \Ebb_{\overline{\alpha_M}}  \left[ c_+  \langle  \rho_M \rangle_{\psi_{M-1}}^2 + c_-  \langle\vsigma_M \rho_M \vsigma_M \rangle^2_{\psi_{M-1}} + 2c_0 \langle \rho_M \rangle_{\psi_{M-1}} \langle \vsigma_{M} \rho_M \vsigma_M \rangle_{\psi_{M-1}} \right] \nonumber \\
     &\; - \left(\Ebb_{\overline{\alpha_M}}  \left[ k_+  \langle \rho_M \rangle_{\psi_{M-1} } + k_-  \langle \vsigma_M \rho_M \vsigma_M \rangle_{\psi_{M-1} } \right]  \right)^2 \label{eq:thm1-proof1-2}\\
     = &\; c_+ \Ebb_{\overline{\alpha_M}} \left[\langle  \rho_M \rangle_{\psi_{M-1}}^2\right] - k^2_+ \left( \Ebb_{\overline{\alpha_M}} \left[\langle  \rho_M \rangle_{\psi_{M-1}}\right] \right)^2 \nonumber \\
     &\; + c_-  \Ebb_{\overline{\alpha_M}} \left[\langle \vsigma_M \rho_M \vsigma_M \rangle_{\psi_{M-1}}^2\right] - k^2_- \left( \Ebb_{\overline{\alpha_M}} \left[\langle \vsigma_M  \rho_M \vsigma_M \rangle_{\psi_{M-1}} \right]\right)^2 \nonumber \\
     &\; + 2c_0 \Ebb_{\overline{\alpha_M}}\left[ \langle \rho_M \rangle_{\psi_{M-1}} \langle \vsigma_M \rho_M \vsigma_M \rangle_{\psi_{M-1}} \right] - 2k_+ k_-  \Ebb_{\overline{\alpha_M}} \left[ \langle \rho_M \rangle_{\psi_{M-1}} \right]  \Ebb_{\overline{\alpha_M}}\left[ \langle \vsigma_M \rho_M \vsigma_M \rangle_{\psi_{M-1}}\right] \\
     = &\; (c_+ - k^2_+) \Ebb_{\overline{\alpha_M}} \left[\langle \rho_M \rangle_{\psi_{M-1}}^2\right] + k_+^2 \Var_{\overline{\alpha_M}} \left[\langle \rho_M \rangle_{\psi_{M-1}}\right] \nonumber \\
     &\; + (c_- - k^2_-) \Ebb_{\overline{\alpha_M}} \left[\langle \vsigma_M \rho_M \vsigma_M \rangle_{\psi_{M-1}}^2 \right] + k_-^2 \Var_{\overline{\alpha_M}} \left[\langle \vsigma_M \rho_M \vsigma_M \rangle_{\psi_{M-1}}\right] \nonumber \\
     &\; + 2(c_0 - k_+ k_-) \Ebb_{\overline{\alpha_M}}\left[ \langle \rho_M \rangle_{\psi_{M-1}} \langle \vsigma_M \rho_M \vsigma_M \rangle_{\psi_{M-1}} \right] + 2k_+k_- \Cov_{\overline{\alpha_M}}\left[  \langle \rho_M \rangle_{\psi_{M-1}}, \langle \vsigma_M \rho_M \vsigma_M \rangle_{\psi_{M-1}} \right]\label{eq:thm1-proof2}  \\ 
     = &\; (c_+ - k^2_+) \Ebb_{\overline{\alpha_M}} \left[ \langle \rho_M \rangle_{\psi_{M-1}}^2 + \langle \vsigma_M \rho_M \vsigma_M \rangle_{\psi_{M-1}}^2 - 2 \langle \rho_M \rangle_{\psi_{M-1}}\langle \vsigma_M \rho_M \vsigma_M \rangle_{\psi_{M-1}} \right] \nonumber \\
     &\; + k_+^2  \Var_{\overline{\alpha_M}} \left[\langle \rho_M \rangle_{\psi_{M-1}}\right] +  k_-^2 \Var_{\overline{\alpha_M}} \left[\langle \vsigma_M \rho_M \vsigma_M \rangle_{\psi_{M-1}}\right]  + 2k_+k_- \Cov_{\overline{\alpha_M}}\left[  \langle \rho_M \rangle_{\psi_{M-1}}, \langle \vsigma_M \rho_M \vsigma_M \rangle_{\psi_{M-1}} \right] \label{eq:thm1-proof3} \\ 
     = &\; (c_+ - k^2_+) \Ebb_{\overline{\alpha_M}} \left[ \langle \rho_M  \rangle_{\psi_{M-1}} - \langle \vsigma_M \rho_M \vsigma_M \rangle_{\psi_{M-1}} \right]^2 + \Var_{\overline{\alpha_M}}\left[ k_+\langle \rho_M  \rangle_{\psi_{M-1}} + k_- \langle \vsigma_M \rho_M \vsigma_M \rangle_{\psi_{M-1}}\right] \label{eq:thm1-proof4}\\
     \geq &\; \Var_{\overline{\alpha_M}}\left[ k_+\langle \rho_M  \rangle_{\psi_{M-1}} + k_- \langle \vsigma_M \rho_M \vsigma_M \rangle_{\psi_{M-1}}\right] \label{eq:thm1-proof4-2} 
\end{align}
where the first inequality is due to Eq.~\eqref{eq:thm1-proof1}, we then reach Eq.~\eqref{eq:thm1-proof2} by using the fact that $\Var_{\valpha}[X(\valpha)] = \Ebb_{\valpha}[X^2(\valpha)] - (\Ebb_{\valpha}[X(\valpha)])^2$ and $\Cov_{\valpha}[X(\valpha),Y(\valpha)] = \Ebb_{\valpha}[X(\valpha)Y[\valpha]]- \Ebb_{\valpha}[X(\valpha)]\Ebb_{\valpha}[Y(\valpha)]$, Eq.~\eqref{eq:thm1-proof3} is from the relation presented in Eq.~\eqref{eq:thm1-proof-cos-sin-rel}. 
Next, in Eq.~\eqref{eq:thm1-proof4} we use the identity $\Var_{\valpha}[X(\valpha)+ Y(\valpha)] = \Var_{\valpha}[X(\valpha)] + \Var_{\valpha}[Y(\valpha)] + 2\Cov_{\valpha}[X(\valpha),Y(\valpha)]$ and to reach the next inequality we throw away the first positive term in the sum. 

\medskip

Notably, the variance of the term in Eq.~\eqref{eq:thm1-proof4-2} is no longer taken over $\alpha_M$ (i.e., the contribution to the variance from $\alpha_M$ is already taken into account). In addition, by denoting $|\psi_{M-2} \rangle =  \prod_{i=1}^{M-2} \widetilde{V}_{i}(\theta^*_i) U_i(\alpha_i) |\psi_0 \rangle$ as well as
\begin{align}
     \widetilde{\rho}_{M-1} = k_+ \widetilde{V}^\dagger_{M-1}(\theta^*_{M-1}) \rho_M \widetilde{V}_{M-1}(\theta_{M-1}) + k_- \widetilde{V}^\dagger_{M-1}(\theta^*_{M-1}) \vsigma_M \rho_M \vsigma_M\widetilde{V}_{M-1}(\theta^*_{M-1}) \;, \label{eq:thm1-proof6}
\end{align}
the lower bound in Eq.~\eqref{eq:thm1-proof4-2} can be expressed as
\begin{align}
    \Var_{\vtheta} \left[ \LC (\vtheta)\right] & \geq \Var_{\overline{\alpha_M}}\left[ k_+\langle \rho_M  \rangle_{\psi_{M-1}} + k_- \langle \vsigma_M \rho_M \vsigma_M \rangle_{\psi_{M-1}}\right] \\
    & =  \Var_{\alpha_1,\alpha_2,...,\alpha_{M-1}}\left[ \langle\psi_{M-2}|U^\dagger_{M-1}(\alpha_{M-1})  \widetilde{\rho}_{M-1} U_{M-1}(\alpha_{M-1})| \psi_{M-2} \rangle \right]\label{eq:thm1-proof5}
\end{align}

Crucially, the derivation steps from Eq.~\eqref{eq:thm1-proof00} to Eq.~\eqref{eq:thm1-proof5}, which are used to get rid of $\alpha_M$ dependence, can be repeated to recursively integrate over other parameters. To be more precise, let us first define 
\begin{align}
    |\psi_{M-l-1} \rangle =  \prod_{i=1}^{M-l-1} \widetilde{V}_{i}(\theta^*_i) U_i(\alpha_i) |\psi_0 \rangle \;,
\end{align}
as well as a general recursive form of Eq.~\eqref{eq:thm1-proof6}
\begin{align}\label{eq:thm1-proof6-2}
    \widetilde{\rho}_{M-l} = k_+ \widetilde{V}^\dagger_{M-l}(\theta^*_{M-l})\widetilde{\rho}_{M-l+1} \widetilde{V}_{M-l}(\theta^*_{M-l}) + k_-  \widetilde{V}^\dagger_{M-l}(\theta^*_{M-l})\vsigma_{M-l+1} \widetilde{\rho}_{M-l+1}  \vsigma_{M-l+1} \widetilde{V}_{M-l}(\theta^*_{M-l}) \;,
\end{align}
where $l \in \{1,2,...,M-1\}$ and we have $\widetilde{\rho}_M = \rho_M$ which gives back Eq.~\eqref{eq:thm1-proof6} for $l = 1$. We note that $\widetilde{\rho}_{M-l}$ can be seen as a mixed state between $\widetilde{V}^\dagger_{M-l}(\theta^*_{M-l})\widetilde{\rho}_{M-l+1} \widetilde{V}_{M-l}(\theta^*_{M-l})$ and $\widetilde{V}^\dagger_{M-l}(\theta^*_{M-l})\vsigma_{M-l+1} \widetilde{\rho}_{M-l+1}  \vsigma_{M-l+1} \widetilde{V}_{M-l}(\theta^*_{M-l})$ for all $l$. This is since $k_+ + k_- = 1$ and $\rho_M$ is a valid quantum state. 

The variance then can be recursively lower bounded, leading to
\begin{align}
     \Var_{\vtheta} \left[ \LC (\vtheta)\right] 
     \geq &\; \Var_{\alpha_1,\alpha_2,...,\alpha_{M-1}}\left[ \langle\psi_{M-2}|U^\dagger_{M-1}(\alpha_{M-1})  \widetilde{\rho}_{M-1} U_{M-1}(\alpha_{M-1})| \psi_{M-2} \rangle \right] \label{eq:thm1-proof7} \\
     \geq &\; \Var_{\alpha_1,\alpha_2,...,\alpha_{M-l}}\left[ \langle\psi_{M-l-1}|U^\dagger_{M-l}(\alpha_{M-l})  \widetilde{\rho}_{M-l} U_{M-l}(\alpha_{M-l})| \psi_{M-l-1} \rangle \right]  \\
     \geq &\; \Var_{\alpha_1}\left[ \langle \psi_1 |U^\dagger_1(\theta_1) \widetilde{\rho}_1 U_1(\theta_1) |\psi_1 \rangle \right] \label{eq:thm1-proof8a} \;,
\end{align}
where in the second inequality we have recursively integrated out parameters $\alpha_{M-l+1}, ..., \alpha_{M}$ and in the last equality we have integrated out all the parameters except $\alpha_1$.  

All that remains is to explicitly bound the variance with respect to $\alpha_1$ 
\begin{align}
     \Var_{\vtheta} \left[ \LC (\vtheta)\right] 
     \geq &\; \Var_{\alpha_1}\left[ \langle \psi_1 |U^\dagger_1(\theta_1) \widetilde{\rho}_1 U_1(\theta_1) |\psi_1 \rangle \right] \label{eq:thm1-proof8b}\; \\
     \geq &\; \left( c_+ \langle \widetilde{\rho}_1 \rangle_{\psi_0}^2 + c_- \langle \vsigma_1 \widetilde{\rho}_1 \vsigma_1 \rangle_{\psi_0}^2 + 2c_0  \langle \widetilde{\rho}_1 \rangle_{\psi_0} \langle \vsigma_1 \widetilde{\rho}_1 \vsigma_1 \rangle_{\psi_0} \right)  -  \left( k_+  \langle \widetilde{\rho}_1 \rangle_{\psi_0}  + k_-  \langle \vsigma_1 \widetilde{\rho}_1 \vsigma_1 \rangle_{\psi_0} \right)^2 \label{eq:thm1-proof9} \\
     = &\; (c_+ - k_+^2) \left( \langle \widetilde{\rho}_1 \rangle_{\psi_0} -  \langle \vsigma_1 \widetilde{\rho}_1 \vsigma_1 \rangle_{\psi_0} \right)^2 \\
     = &\; (c_+ - k_+^2) \left(\Tr\left[ \left( |\psi_0 \rangle\langle \psi_0 | - \vsigma_1 |\psi_0 \rangle \langle \psi_0 | \vsigma_1 \right) \widetilde{\rho}_1\right] \right)^2  \label{eq:thm1-proof10}
\end{align}
where Eq.~\eqref{eq:thm1-proof9} to Eq.~\eqref{eq:thm1-proof10} follows in the same manner as Eq.~\eqref{eq:thm1-proof1-2} to Eq.~\eqref{eq:thm1-proof4-2}. From recursively expanding $\widetilde{\rho}_1$ (according to Eq.~\eqref{eq:thm1-proof6-2}), we can write:
\begin{align}
    \widetilde{\rho}_1 &=  k_+^{M-1} \left( \prod_{i=1}^M \widetilde{V}_i(\theta_i^*) \right)^{\dagger} \rho_{(\vtheta^*,\delta t)} \left( \prod_{i=1}^M \widetilde{V}_i(\theta_i^*) \right) + (1 - k_+^{M-1}) \xi \\ 
    &=  k_+^{M-1}  U^{\dagger}(\vtheta^*)   \rho_{(\vtheta^*,\delta t)}  U(\vtheta^*) + (1 - k_+^{M-1}) \xi 
\end{align}
where $\xi$ is some complicated mixed state and, for clarification, we note that $k_+^{M-1} = \left(k_+\right)^{M-1}$ with $k_+$ defined in Eq~\ref{eq:var-proof-k-plus}~\footnote{This clarification on $k_+^M$ is included at the request of one of the authors. }. Thus we can write 
\begin{align}
    \Var_{\vtheta} \left[ \LC (\vtheta)\right] \geq&\; (c_+ - k_+^2) \left(\Tr\left[ \left( |\psi_0 \rangle\langle \psi_0 | - \vsigma_1 |\psi_0 \rangle \langle \psi_0 | \vsigma_1 \right) \left( k_+^{M-1}  U^{\dagger}(\vtheta^*)   \rho_{(\vtheta^*,\delta t)}  U(\vtheta^*) + (1 - k_+^{M-1}) \xi \right)\right] \right)^2 \label{eq:thm1-proof11} \\
     \geq &\; (c_+ - k_+^2) \min_{\Tilde{\xi} \in [-1,1]}\left( k_+^{M-1} \Tr\left[\left( |\psi_0 \rangle\langle \psi_0 | - \vsigma_1 |\psi_0 \rangle \langle \psi_0 | \vsigma_1 \right)  U^{\dagger}(\vtheta^*)   \rho_{(\vtheta^*,\delta t)}  U(\vtheta^*) \right] + (1-k_+^{M-1}) \Tilde{\xi}\right)^2
\end{align}
where in the final line we minimize over the free parameter $\Tilde{\xi} = \Tr\left[ \left( |\psi_0 \rangle\langle \psi_0 | - \vsigma_1 |\psi_0 \rangle \langle \psi_0 | \vsigma_1\right) \xi \right] \in [-1,1]$ by noting that $ \Tr\left[ \left( |\psi_0 \rangle\langle \psi_0 | - \vsigma_1 |\psi_0 \rangle \langle \psi_0 | \vsigma_1 \right) \xi\right]$ is bounded between $-1$ and $1$. This completes the proof of the proposition.
\end{proof}

\subsection{Proof of Theorem~\ref{thm:variance-lower-bound}} 

In this subsection, we analytically show that the lower bound of the variance scales polynomially with the number of parameters $M$ when the perturbation is within $1/\sqrt{M}$ region.

\setcounter{theorem}{0}
\begin{theorem}[Lower-bound on the loss variance, Formal]\label{oldsimplifiedR}
Assume \revadd{an} initial state $\rho_0$ and let us choose $\vsigma_1$ such that $\Tr[\rho_0 \vsigma_1 \rho_0\vsigma_1] = 0$. Given that the  time-step $\delta t$ respects 
\begin{align}\label{eq:var-time-limit}
    \frac{1}{2 \lambda_{\rm max}} \geq \delta t \;,
\end{align}
where $\lambda_{\rm max}$ is the largest eigenvalue of $H$, 
and the perturbation $r$ obeys
\begin{align}\label{eq:perturb-condition1}
    \frac{3 r_0^2\left(1 - 4 \lambda_{\rm max}^2 \delta t^2\right)}{2(M-1)\left( 1 - 2\lambda_{\rm max}^2 \delta t^2\right)} \geq r^2  \;,
\end{align}
with some $r_0$ such that $0< r_0 <1$, then the variance of the loss function within the region $\vol(\thv^*, r)$ is lower bounded as
\begin{align}
    \Var_{\vtheta \sim\uni(\vtheta^*, r)} \left[ \LC (\vtheta)\right] 
    & \geq \frac{4 r^4}{45}\left(1 - \frac{4r^2}{7} \right) \left[(1 - r_0)(1 - 4 \lambda_{\rm max}^2 \delta t ^2 )\right]^2 \;.
\end{align}
In addition, by choosing $r$ such that $r\in \Theta\left( \frac{1}{\sqrt{M}}\right)$, we have
\begin{align}
     \Var_{\vtheta \sim\uni(\vtheta^*, r)} \left[ \LC (\vtheta)\right] \in \Omega\left( \frac{1}{M^2}\right) \;.
\end{align}
\end{theorem}

\begin{proof}
From Proposition~\ref{prop:variance-lower-bound}, we first recall the variance bound in Eq.~\eqref{eq:prop1-exact-var} is of the form
\begin{equation}\label{eq:proof-new1}
    \Var_{\thv\sim\uni(\thv^*, r) }[\mf(\thv)]\geq  \ (c_+ - k_+^2) \min_{\Tilde{\xi}\in [-1,1]} \left( k_+^{M-1} \Delta_{\thv^*}  + (1-k_+^{M-1})\Tilde{\xi} \right)^2 ,
\end{equation}
where we have
\begin{align}
    &c_+     :=  \mathbb{E}_{\alpha \sim\uni(0,r)} [ \cos^4{\alpha}]\;, \\
    &k_+   := \mathbb{E}_{\alpha \sim\uni(0,r)} [ \cos^2{\alpha}] \; ,                                                               \\
    &\Delta_{\thv^*} := \Tr[(\rho_0 -\vsigma_1 \rho_0 \vsigma_1) U^{\dagger}\left(\vec{\theta^*}\right)\rho_{(\vtheta^*,\delta t)} U\left(\vec{\theta^*}\right)] \; .
\end{align}
Here $\vsigma_1$ is the Pauli string associated with the first gate in the circuit $U(\thv)$ as defined in Eq.~\eqref{eq:circuit}, $\rho_0 = |\psi_0 \rangle\langle\psi_0|$ is an initial state before the time evolution and $\rho_{(\vtheta^*,\delta t)}= e^{-iH\delta t} U(\vec{\theta^*})\rho_0 U^\dagger(\vec{\theta^*})e^{iH\delta t}$ with $H$ being the underlying Hamiltonian of the quantum dynamics. 

We now notice that if the perturbation $r$ is chosen such that the following condition is satisfied
\begin{align}\label{eq:condition-perturbation}
     k_+^{M-1} \Delta_{\thv^*} \geq 1-k_+^{M-1} \;,
\end{align}
then $\Tilde{\xi} = -1$ minimises the lower bound which leads to
\begin{align}
    \Var_{\vtheta \sim\uni(\vtheta^*, r)} \left[ \LC (\vtheta)\right] & \geq (c_+ - k_+^2) \left( k_+^{M-1} \Delta_{\thv^*}  - (1-k_+^{M-1}) \right)^2 \\
    & = (c_+ - k_+^2) \left( k_+^{M-1} \Tr\left[\left( \rho_0 - \vsigma_1 \rho_0 \vsigma_1 \right)U^{\dagger}(\vtheta^*)   \rho_{(\vtheta^*,\delta t)}  U(\vtheta^*) \right] - (1-k_+^{M-1}) \right)^2 \\
    & = (c_+ - k_+^2) \left( k_+^{M-1}\left( F\left(\rho_{(\vtheta^*,0)}, \rho_{(\vtheta^*, \delta t)}\right) - \Tr\left[U(\vtheta^*) \vsigma_1 \rho_0 \vsigma_1 U^\dagger(\vtheta^*) \rho_{(\vtheta^*, \delta t)} \right] + 1\right) - 1\right)^2 \;, \label{eq:proof-coro1-1}
\end{align}
where $F(\rho, \rho') = \Tr[\rho \rho']$ is the fidelity between two pure states $\rho$ and $\rho'$, and
\begin{align}
    \rho_{(\vtheta^*, \delta t)} = |\psi(\vtheta^* , \delta t) \rangle  \langle\psi(\vtheta^* , \delta t)  |= e^{-iH\delta t} U(\vtheta^*) \rho_0 U^\dagger(\vtheta^*) e^{iH\delta t} \;.
\end{align}

\medskip

We note that the condition in Eq.~\eqref{eq:condition-perturbation} can be equivalently expressed as
\begin{align}\label{eq:condition-perturbation2}
    k_+^{M-1} & \geq \frac{1}{1 + \Delta_{\thv^*}}
     = \frac{1}{1 + \Tr\left[\left( \rho_0 - \vsigma_1 \rho_0 \vsigma_1 \right) U^{\dagger}(\vtheta^*)   \rho_{(\vtheta^*,\delta t)}  U(\vtheta^*) \ \right]} \;,
\end{align}
where we explicitly expand $\Delta_{\thv^*}$.

\medskip

Crucially, for the majority of the rest of the proof, we aim to show that the condition in Eq.~\eqref{eq:condition-perturbation2} is satisfied if the perturbation is chosen such that
\begin{align}\label{eq:proof-coro1-perturb}
    \frac{3 r_0^2\left(1 - 4 \lambda_{\rm max}^2 \delta t^2\right)}{2(M-1)\left( 1 - 2\lambda_{\rm max}^2 \delta t^2\right)} \geq r^2 \;,
\end{align}
where $r_0$ is some constant within the range $0<r_0<1$.
In order to prove this, we first note the following bound of $ k_+^{M-1}$ which follows as
\begin{align}
    k_+^{M-1} & = \left(\frac{1}{2r} \int_{-r}^{r} d\alpha \cos^2(\alpha) \right)^{M-1} \\
    & = \left( \frac{1}{2} + \frac{\sin{(2r)}}{4r}\right)^{M-1} \\
    & \geq  \left( 1 - \frac{r^2}{3}\right)^{M-1} \\
    & \geq 1 - \frac{(M-1)r^2}{3}  \label{eq:proof-bound-kplus} \\
    & > 1 - \frac{(M-1)r^2}{3r_0^2} \;, \label{eq:proof-coro1-0}
\end{align}
where the first inequality is by directly expanding the base and keeping only the second order term, the second inequality is due to Bernoulli's inequality and finally the last inequality holds because $0 <r_0 < 1$. 
We will come back to this inequality soon. 

\medskip

Now, the term $\Tr\left[\left( \rho_0 - \vsigma_1 \rho_0 \vsigma_1 \right) U^{\dagger}(\vtheta^*)   \rho_{(\vtheta^*,\delta t)}  U(\vtheta^*) \ \right]$ can be bounded as follows. \revadd{Assume we can choose $\vsigma_1$} such that $\Tr[\rho_0 \vsigma_1 \rho_0 \vsigma_1] = 0$~\footnote{For example, consider the all-zero basis state as an initial state $\rho_0 = |00...0\rangle\langle00...0|$. We can pick the first generator as $\vsigma_1 = X_1$.}, which leads to
\begin{align}\label{eq:product_state}
    \Tr\left[ \rho_{(\thv^*, 0)} U(\thv^*)\vsigma_1 \rho_0 \vsigma_1 U^\dagger(\thv^*) \right]   = \Tr\left[\rho_0 \vsigma_1 \rho_0 \vsigma_1 \right] = 0 \;.
\end{align}
Then, we construct an orthonormal basis $\{ |\phi_i \rangle\langle \phi_i | \}_{i=1}^{2^n}$ such that 
\begin{align}\label{eq:extend1}
    |\phi_1 \rangle\langle \phi_1| & = \rho_{(\thv^*,0)} \; , \\
    |\phi_2 \rangle\langle \phi_2| & = U(\thv^*)\vsigma_1 \rho_0 \vsigma_1 U^\dagger(\thv^*) \;, \label{eq:basis-phi2}
\end{align}

and the rest are some other orthornormal states necessary to complete the basis. With this basis, we have the following bound
\begin{align}
 \Tr\left[\left( \rho_0 - \vsigma_1 \rho_0 \vsigma_1 \right) U^{\dagger}(\vtheta^*)   \rho_{(\vtheta^*,\delta t)}  U(\vtheta^*) \ \right] 
    & = 
    F\left(\rho_{(\vtheta^*,0)}, \rho_{(\vtheta^*, \delta t)}\right) - \Tr\left[|\phi_2 \rangle\langle \phi_2| \rho_{(\vtheta^*, \delta t)} \right] \\
    &\geq F\left(\rho_{(\vtheta^*,0)}, \rho_{(\vtheta^*, \delta t)}\right) - \sum_{i=2}^{2^n}  \Tr\left[  |\phi_i \rangle \langle \phi_i |\rho(\vtheta^*,\delta t)\right] \\
    & = 2F\left(\rho_{(\vtheta^*,0)}, \rho_{(\vtheta^*, \delta t)}\right) - 1 \\
    & \geq 1 - 4 \lambda_{\rm max}^2 \delta t ^2 \;, \label{eq:proof-coro1-2}
\end{align}
where the first equality is by writing the first term in the fidelity form and writing the second term in $|\phi_2 \rangle\langle \phi_2|$ in Eq.~\eqref{eq:basis-phi2}, in the first inequality we include terms corresponding to other basis (which holds since $\Tr[\rho |\phi_i \rangle\langle \phi_i|] \geq 0$ for any $\rho$ and $|\phi_i \rangle\langle \phi_i|$). Next, the second equality is from the completeness of the basis $\sum_{i=1}^{2^n} |\phi_i \rangle\langle \phi_i | = \1$, the last inequality is due to Lemma~\ref{lem:bound_fidelity} with $\lambda_{\rm max}$ being the largest eigenvalue of $H$. 

We note that in order for the lower bound in Eq.~\eqref{eq:proof-coro1-2} to be informative it is required that $1 \geq 4 \lambda_{\rm max}^2 \delta t ^2$. Up on rearranging, this leads to the constraint on the time-step as
\begin{align}
    \frac{1}{2\lambda_{\rm max}} \geq \delta t \;,
\end{align}
which is the condition specified in Eq.~\eqref{eq:var-time-limit}. By assuming that the time-step satisfying the aforementioned constrain, we now proceed from Eq.~\eqref{eq:proof-coro1-2} by adding $1$ to both sides and rearranging the terms which leads to
\begin{align} \label{eq:proof-coro1-3}
    \frac{1}{2 - 4 \lambda_{\rm max}^2 \delta t ^2}\geq \frac{1}{1 + \Tr\left[\left( \rho_0 - \vsigma_1 \rho_0 \vsigma_1 \right) U^{\dagger}(\vtheta^*)   \rho_{(\vtheta^*,\delta t)}  U(\vtheta^*) \ \right]}  \;.
\end{align}
We remark that the right-hand side of Eq.~\eqref{eq:proof-coro1-3} appears in the condition in Eq.~\eqref{eq:condition-perturbation2}. 

\medskip
We are now ready to put everything together. Importantly, the condition in Eq.~\eqref{eq:condition-perturbation2} is satisfied if we enforce the left-hand side of Eq.~\eqref{eq:proof-coro1-0} to be larger than the right-hand side of Eq.~\eqref{eq:proof-coro1-3}. That is, we have $ k_+^{M-1} \geq  \frac{1}{1 + \Tr\left[\left( \rho_0 - \vsigma_1 \rho_0 \vsigma_1 \right) U^{\dagger}(\vtheta^*)   \rho_{(\vtheta^*,\delta t)}  U(\vtheta^*) \ \right]}$ to be true if the following holds
\begin{align} \label{eq:proof-coro1-perturb-rearrange}
    1 - \frac{(M-1)r^2}{3r^2_0} \geq \frac{1}{2 -   4 \lambda_{\rm max}^2 \delta t ^2} \;.
\end{align}
By rearranging the inequality in Eq.~\eqref{eq:proof-coro1-perturb-rearrange}, we have the perturbation regime of $r$ to be Eq.~\eqref{eq:proof-coro1-perturb} as previously stated.

\medskip

The last step is to bound the variance when $r$ satisfies Eq.~\eqref{eq:proof-coro1-perturb}. the variance of the loss in Eq.~\eqref{eq:proof-coro1-1} can be bounded as
\begin{align}
    \Var_{\vtheta \sim\uni(\vtheta^*, r)} \left[ \LC (\vtheta)\right] & \geq (c_+ - k_+^2) \left[ k_+^{M-1}\left( F\left(\rho_{(\vtheta^*,0)}, \rho_{(\vtheta^*, \delta t)}\right) - \Tr\left[U(\vtheta^*) \vsigma_1 \rho_0 \vsigma_1 U^\dagger(\vtheta^*) \rho_{(\vtheta^*, \delta t)} \right] + 1\right) - 1\right]^2 \label{eq:final-variance-forricard-1} \\
    & \geq  (c_+ - k_+^2) \left[ k_+^{M-1}\left(2 - 4 \lambda_{\rm max}^2 \delta t ^2 \right) - 1\right]^2  \\
    & \geq (c_+ - k_+^2) \left[ \left(1 - \frac{(M-1)r^2}{3}\right)\left(2 - 4 \lambda_{\rm max}^2 \delta t ^2 \right) - 1\right]^2  \label{eq:final-variance-forricard-2}\\
    & \geq (c_+ - k_+^2) \left[(1 - r^2_0)(1 - 4 \lambda_{\rm max}^2 \delta t ^2 )\right]^2 \\
    & \geq \frac{4 r^4}{45} \left(1 - \frac{4r^2}{7} \right)\left[(1 - r^2_0)(1 - 4 \lambda_{\rm max}^2 \delta t ^2 )\right]^2 \label{eq:final-variance-forricard-3}\;,
\end{align}
where the second inequality is due to Eq.~\eqref{eq:proof-coro1-2}, the third inequality is by bounding $k_+^{M-1}$ with Eq.~\eqref{eq:proof-bound-kplus} and in the next inequality we explicitly use the perturbation regime of $r$ in Eq.~\eqref{eq:proof-coro1-perturb}. To reach the last inequality, we directly bound $c_+ - k_+^2  = \frac{1}{2r} \int_{-r}^{r} d\alpha \cos^4(\alpha) - \left(\frac{1}{2r} \int_{-r}^{r} d\alpha \cos^2(\alpha) \right)^2 \geq \frac{4 r^4}{45} - \frac{16r^6}{315}$ by expanding it in the series and keeping the terms which result in the lower bound.

\end{proof}

We now comment on the assumption that an initial state $\rho_0$ is a product state and discuss a possible extension to an arbitrary initial state. In essence, the product state assumption is used in the proof above to ensure that the term $\Delta_{\thv^*}$ in Eq.~\eqref{eq:var-delta-key} is non-vanishing (see Eq.~\eqref{eq:product_state} to Eq.~\eqref{eq:proof-coro1-2}). However, we argue here that our results should hold more generally for arbitrary initial states as long as the first gate interacts non-trivially with the loss. In particular, this happens for a small enough time-step  $\delta t$ as long as the first gate does not rotate $\rho_0$ into a subspace that is fully parallel to itself.

To illustrate this, we can expand $e^{-iH\delta t}$ and keep only the leading order in $\delta t$ with an arbitrary non-product initial state. Since $\rho_0$ is no longer limited to be a product state, the orthornormal basis construction where $ U(\thv^*)\vsigma_1 \rho_0 \vsigma_1 U^\dagger(\thv^*) $ is chosen to be orthonormal to $\rho_{(\thv^*, 0)}$ (see Eq.~\eqref{eq:product_state}) is no longer guaranteed. However, we can modify the steps slightly and 
decompose $U(\thv^*)\vsigma_1 \rho_0 \vsigma_1 U^\dagger(\thv^*) $ into a parellel and a perpendicular component i.e.,
\begin{align}
   U(\vtheta^*)\vsigma_1 \rho_0 \vsigma_1  U^{\dagger}(\vtheta^*) = \left(a\ket{\phi_1}+ b\ket{\phi_2} \right)\left( \bra{\phi_1}a^* + \bra{\phi_2}b^* \right) \ ,
\end{align}
where $\ketbra{\phi_1} = \rho_{(\vtheta^*,0)}$, $\ketbra{\phi_2}$ is orthonormal to $\rho_{(\vtheta^*,0)}$, $a$ and $b$ are coefficients in the parallel and orthogonal directions such that $|a|, |b| \leq 1$. Then we can use Taylor's series to expand
\begin{align}
     \Tr\left[U(\vtheta^*) \vsigma_1 \rho_0 \vsigma_1 U^{\dagger}(\vtheta^*)   \rho_{(\vtheta^*,\delta t)} \right] =& \Tr\left[\left(a\ket{\phi_1}+ b\ket{\phi_2} \right)\left( \bra{\phi_1}a^* + \bra{\phi_2}b^* \right)\rho_{(\vtheta^*,\delta t)} \right] \\
     =& \Tr\left[\left(a\ket{\phi_1}+ b\ket{\phi_2} \right)\left( \bra{\phi_1}a^* + \bra{\phi_2}b^* \right) e^{-i H \delta t} \ketbra{\phi_1}  e^{i H \delta t} 
 \right] \\
    = & |a|^2+ \delta t\left( i a b^*\bra{\phi_1}H\ket{\phi_2} -ia^* b\bra{\phi_2}H\ket{\phi_1} \right) + \order{\delta t ^2} \;,
\end{align}
where, for the purpose of demonstration, we are only interested in the leading order in $\delta t$. 
Therefore with this we have that the term $\Delta_{\thv^*}$ is 
\begin{equation}
    \Delta_{\thv^*} = (1-|a|^2) + \delta t \left( i a b^*\bra{\phi_1}H\ket{\phi_2} -ia^* b\bra{\phi_2}H\ket{\phi_1} \right) + \order{\delta t^2} \ .
\end{equation}
Hence, for $|a| \in \Omega(1/\poly(n))$ (which is expected to hold when the first gate does not commute with $\rho_0$) and small time-step $\delta t \ll 1$, one can follow the same proof steps which then results in the polynomial scaling of the loss variance in the hypercube with $r$ scaling polynomially.

\revadd{\subsection{Additional numerics}\label{appx:numeric-var}

\begin{figure}[h]
    \centering
    \includegraphics[width=1\linewidth]{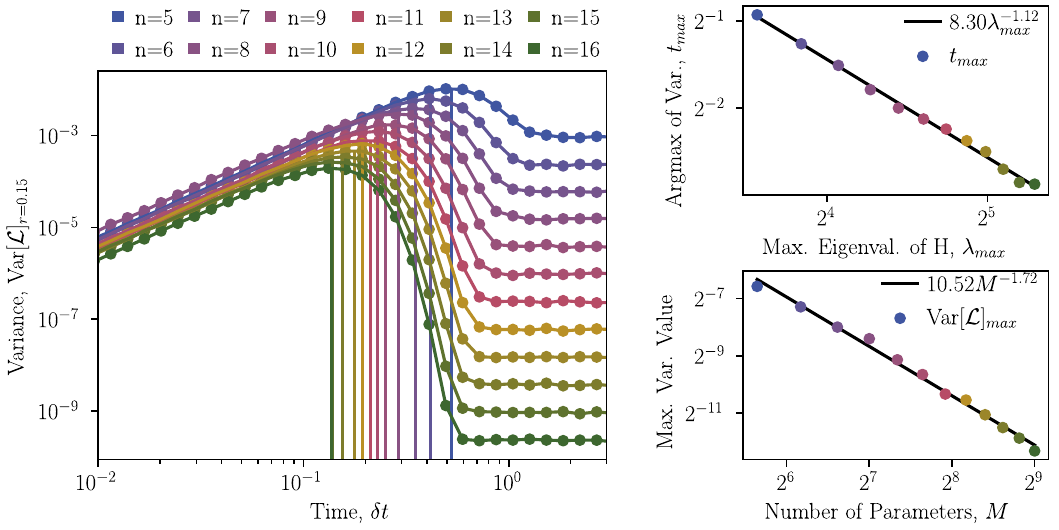}
    \caption{\revadd{\textbf{Variance of landscape and width of narrow gorge.} Here we study the landscape of $\mathcal{L}(\thv)$ at a given distance $r=0.15$ as we increase $\delta t$ and for different system sizes $n$. We consider a hardware efficient 
    ansatz with $n$ layers and random initial parameters within the hypercube.
    a) We plot $\Var_{\thv\sim\uni(\vec{0}, r=0.15)}[\mathcal{L}(\thv,\delta t)]$ as function of $\delta t$. We keep track of its maximum value (marked with a vertical line) for each system size.
    b) The value $r_{\rm max}$ for which the variance peaks as function of the number of parameters in the ansatz.
    c) Maximum value of the variance for different system's size is plotted as a function of parameters. 
    }}
    \label{fig:fixedr_variance}
\end{figure}

In Figure \ref{fig:fixedr_variance} we study the dependency of 
$
\Var_{\thv\sim\uni(\thv^*, r) }[\mf(\thv)]
$
on different time-step sizes for a fixed region $\uni(\theta,r=0.15)$. It is observed that for small $\delta t$ the variance scales favorably while large $\delta t$ leads to the exponential vanishing of the variance. We keep track of the maximum value of the variance as a function of $\delta t$ and we find that it is reached for $\delta t \propto \lambda_{max}^{-1.12}$ while its value scales as $ \Var_{\thv\sim\uni(\thv^*, r) }[\mf(\thv)]\propto M^{-1.72}$. While the scaling for $\delta t$ is slightly larger than that suggested in Theorem~\ref{thm:variance-lower-bound}, i.e. $\delta t \propto \lambda_{\rm max}^{-1}$, the variance also decays slower than the suggested scaling of $
\Var_{\thv\sim\uni(\thv^*, r) }[\mf(\thv)]
\propto M^{-2}
$.
Note that following the maximum value of the variance is an arbitrary choice. In practice, having the polynomial large variance rather than maximizing it is more important. From what numerically observed, it seems we could achieve this with larger time-step's sizes.

}

\section{Proof convexity}\label{app:convexity}
\begin{theorem}[Approximate convexity of the landscape, Formal]\label{th:pRTE} For a time-step of size 
\begin{align}
    \delta t \leq \frac{\mu_{\rm min} + 2 |\epsilon| }{16 M \lambda_{\rm max}} \;,
\end{align}
the loss landscape is $\epsilon$-convex in a hypercube of width $2r_c$ around a previous optimum $\vec{\theta^*}$ i.e., $\vol(\thv^*, r_c)$ such that
\begin{equation}\label{eq:convexitytheorem-appx}
     r_c \geq \frac{1}{M}\left(\frac{\mu_{\rm min}+2|\epsilon|}{16 M} - \lambda_{\rm max} \delta t\right) \;,
\end{equation}
where $\mu_{\min}$ is the minimal eigenvalue of the Fisher information matrix associated with the loss. 
\end{theorem}
\begin{proof}
We first recall that the region of the loss function is $\epsilon$-convex (i.e., Definition~\ref{def:epsilon-convex}) if all eigenvalues of the Hessian matrix of the loss function i.e., $\LC(\thv) = 1 - F\left[ U(\thv) \rho_0 U^\dagger(\thv), \rho(\vtheta^*, \delta t) \right]$ within the region are larger than $-|\epsilon|$, which can be re-expressed in terms of the fidelity as
\begin{align}\label{eq:proof-def-epsilon-convex}
    \left[\nabla^2_{\thv} F\left( U(\thv) \rho_0 U^\dagger(\thv), \rho(\vtheta^*, \delta t) \right)\right]_{\rm max} \leq |\epsilon| \; ,
\end{align}
for all $\vtheta \in \vol(\theta^*, r)$ with $[A]_{\rm max}$ being the largest eigenvalue of the matrix $A$.   

By using Taylor's expansion around $\thv^*$ (see Theorem~\ref{thm:taylor}), the fidelity can be written in the form of
\begin{align}\label{eq:fidel-taylor}
    F(\vec{x}) = 1 - \sum_{i,j} \frac{x_i x_j}{4} \mathcal{F}_{ij}(\vec{0}) + \sum_{i,j,k} \frac{x_i x_j x_k}{6} \left(\frac{\partial^3 F(\vec{x})}{\partial x_i \partial x_j \partial x_k}\right) \bigg|_{\vec{x} = \vec{\nu}} \;,
\end{align}
where we introduce the shorthand notation of the fidelity around this region as $F(\vec{x})$ with $\vec{x} = (\thv - \thv^*, \delta t)$, $\FC_{ij} (\vec{0})$ are elements of the quantum fisher information at $\vec{x} = \vec{0}$, and the last term is the result of the Taylor's remainder theorem with $\vec{\nu} = c \vec{x}$ for some $c\in [ 0,1]$.

For convenience, we denote $\AC_{ijk}(\vec{x}) = \frac{\partial^3 F(\vec{x})}{\partial x_i \partial x_j \partial x_k}$. This third derivative can be expressed as a nested commutator of the form (for $k > j >i$) 
\begin{align}\label{eq:nested_comm}
    \AC_{ijk}(\vec{x})  := \frac{\partial^3 F(\vec{x})}{\partial x_i \partial x_j \partial x_k} & = \Tr\left[ U^{(M+1,k)} i\left[ U^{(k,j)} i \left[ U^{(j,i)} i \left[ U^{(i,0)} \rho_0 {U^{(i,0)}}^{\dagger}, \vsigma_i\right]{U^{(j,i)}}^\dagger, \vsigma_j  \right]{U^{(k,j)}}^\dagger, \vsigma_k \right] {U^{(M+1,k)}}^\dagger \rho_{(\thv^*, 0)}\right] \;,
\end{align}
where $U(\vec{x}) = U^{(M+1,k)}U^{(k,j)}U^{(j,i)}U^{(i,0)}$ with $U^{(a,b)} = \prod_{l = a+1}^b e^{-i x_l \vsigma_l} \widetilde{V}_{l}$ such that $\vsigma_{M+1} := H$ and $\widetilde{V}_{M+1} = \1$. For clarification we emphasise that the notation $\sigma_{M+1}:=H$ does \textit{not} imply that $H^2 = \1$, $H$ is still a general Hamiltonian, but rather this is just a way of simplifying the notation. That is, $U(\vec{x})$ is decomposed into 4 sections e.g., $U^{(i,0)}$ contains the part of $U(\thv)$ from the first gate to the $i^{\rm th}$ gate.

\medskip

Now, we consider an element of $\nabla^2_{\thv} F\left( \vec{x}\right)$ which can be obtained by explicitly differentiating $F(\vec{x})$ in Eq.~\eqref{eq:fidel-taylor} with respect to the variational parameters (i.e., $x_l$ and $x_m$ cannot be $\delta t$)
\begin{align}
    \frac{\partial^2 F(\vec{x})}{\partial x_l \partial x_m } = - \frac{1}{2} \FC_{lm}(\vec{0}) + \frac{1}{6} \widetilde{\AC}_{lm}(\vec{\nu}) \;,
\end{align}
with
\begin{align}
     \widetilde{\AC}_{lm}(\vec{\nu}) = \sum_{i=1}^{M+1} x_i \left( \AC_{lmi}(\vec{\nu}) + \AC_{lim}(\vec{\nu}) + \AC_{ilm}(\vec{\nu}) + \AC_{mli}(\vec{\nu}) + \AC_{mil}(\vec{\nu}) +  \AC_{iml}(\vec{\nu})\right) \;,
\end{align}
where we remark that here the sum includes the time component $\delta t$. 

Now, the largest eigenvalue of $\nabla^2_{\thv} F\left( \vec{x}\right)$ can be bounded as
\begin{align}
    \left[ \nabla^2_{\thv} F\left( \vec{x}\right)\right]_{\rm max} \leq -\frac{1}{2} \left[ \FC(\vec{0}) \right]_{\rm min} + \frac{1}{6} [\widetilde{\AC}(\vec{\nu}) ]_{\rm max} \;,
\end{align}
where we denote $[A]_{\rm min}$ as the smallest eigenvalue of the matrix $A$.

In order to bound $[\widetilde{\AC}(\vec{\nu}) ]_{\rm max}$, we first consider the bound on $\AC_{ilm}(\vec{\nu})$
\begin{align}\label{eq:bound}
    \AC_{ilm}(\vec{\nu}) & \leq \left| \AC_{ilm}(\vec{\nu})\right| \\
    & \leq \left\| U^{(M+1,m)} i\left[ U^{(m,l)} i \left[ U^{(l,i)} i \left[ U^{(i,0)} \rho_0 {U^{(i,0)}}^{\dagger}, \vsigma_i\right]{U^{(l,i)}}^\dagger, \vsigma_l  \right]{U^{(m,l)}}^\dagger, \vsigma_m \right] {U^{(M+1,m)}}^\dagger \right\|_\infty \left\| \rho_{(\thv^*, 0)} \right\|_1 \\
    & \leq 2^3 \| \vsigma_i \|_{\infty} \| \vsigma_l \|_{\infty} \| \vsigma_m \|_{\infty} \\
    & = 8  \| \vsigma_i \|_{\infty} \label{eq:bound2}\;.
\end{align}
Here the second inequality is due to H\"{o}lder's inequality. In the third inequality we use a few identities including (i) the one-norm of a pure state is $1$, (ii) $\| U A\|_p = \| A\|_p$ for any unitary $U$, (iii) $\| i [A,B] \|_p = 2 \|A\|_p \|B\|_p$ and lastly (iv) $\|AB\|_p \leq \|A\|_p \|B\|_p$. To reach the final equality, we recall that since $x_l$ and $x_m$ cannot be a time component $\delta t$, $\vsigma_l$ and $\vsigma_m$ are generators of the circuit which have $ \| \vsigma_l \|_{\infty} = \| \vsigma_m \|_{\infty} =1$.

We now bound the sum of the absolute of elements in a row of $\widetilde{\AC}(\vec{\nu})$ as
\begin{align}\label{eq:bound_tilde}
    \sum_{m=1}^M \left| \widetilde{\AC}_{lm}(\vec{\nu}) \right| 
    & \leq \sum_{m=1}^M \sum_{i = 1}^{M+1} |x_i| \left( |\AC_{lmi}(\vec{\nu})| + |\AC_{lim}(\vec{\nu})| + |\AC_{ilm}(\vec{\nu})| + |\AC_{mli}(\vec{\nu})| + |\AC_{mil}(\vec{\nu})| +  |\AC_{iml}(\vec{\nu})| \right)\\
    & \leq 48 \sum_{m=1}^{M} \sum_{i=1}^{M+1} |x_i| \| \vsigma_i \|_{\infty} \\
    & \leq 48 M  \left( \lambda_{\rm max} \delta t + M r \right)
\end{align}
By invoking Proposition~\ref{prop:upper-eigen}, the largest eigenvalue of the matrix can then be bounded as
\begin{align}
    [\widetilde{\AC}(\vec{\nu}) ]_{\rm max} 
    & \leq 48 M  \left( \lambda_{\rm max} \delta t + M r \right) \;.
\end{align}
Finally, we can guarantee the region of $\epsilon$-convexity (i.e., Eq.~\eqref{eq:proof-def-epsilon-convex}) by enforcing the following condition
\begin{align}
    -\frac{1}{2} \left[ \FC(\vec{0}) \right]_{\rm min} + 8 M  \left( \lambda_{\rm max} \delta t + M r \right)   \leq |\epsilon| \;.
\end{align}
Upon rearranging the terms, we have
\begin{align} \label{eq:proof-ricard-hello}
    r \leq \frac{1}{M}\left(\frac{\mu_{\rm min}+2|\epsilon|}{16 M} - \lambda_{\rm max} \delta t\right) \;.
\end{align}
Indeed, this implies that \textit{any} hypercube $\vol(\thv^*, r)$ such that $r$ satisfies Eq.~\eqref{eq:proof-ricard-hello}  is guaranteed to be approximately convex. Hence, we know that the total $\epsilon$-convex region has to be at least of size $\frac{1}{M}\left(\frac{\mu_{\rm min}+2|\epsilon|}{16 M} - \lambda_{\rm max} \delta t\right)$. More explicitly, by denoting $r_c$ to be the length of the total $\epsilon$-approximate convex region $\vol(\thv^*, r_c)$, we have
\begin{equation}
    r_c \geq \frac{1}{M}\left(\frac{\mu_{\rm min}+2|\epsilon|}{16 M} - \lambda_{\rm max} \delta t\right) \;.
\end{equation}
We note that the bound is only informative if the time-step respects
\begin{align}
    \delta t \leq \frac{\mu_{\rm min} + 2 |\epsilon| }{16 M \lambda_{\rm max}} \;.
\end{align}
This completes the proof of the theorem.
\end{proof}

\section{Adiabatic Moving Minima}\label{app:moving-min}
In this section, we provide further analysis on the adiabatic moving minimum, including the proof of Theorem~\ref{thm:adiabaticminimum} and some some technical subtleties. We first recall the definition of the adiabatic minimum and also introduce a definition of the adiabatic shift. 

\setcounter{definition}{1}
\begin{definition}[Adiabatic Minima]
For any time $\delta t$ in the range $[0,T]$, the function corresponding to the evolution of the adiabatic minima for some initial minimum $\vtheta^*$, is a continuous function $\vtheta_A(\delta t) \in C^{\infty}(\mathbb{R},\mathbb{R}^{m})$ such that $\vtheta_A(0) = \vtheta^*$ and
\begin{align}
    \nabla_{\thv} \mf(\vtheta_A(\delta t), \delta t)=\vec{0} \;.
\end{align}
The adiabatic minimum at time $\delta t$ is $\vtheta_A(\delta t)$
 \end{definition}

\setcounter{definition}{3}
\begin{definition}[Adiabatic shift of the previous minima]\label{def:adaibatic-shift}
The shift of the adiabatic minimum with respect to the previous optimal point is defined as
\begin{align}
    \valpha_{A}(\delta t) = \thv_A(\delta t) - \thv^* \;,
\end{align}
and also respects
\begin{align}
    \nabla_{\valpha} \LC(\valpha,\delta t) \big|_{\valpha = \valpha_A(\delta t)} = \vec{0} \;,
\end{align}
for any time $\delta t$. 
\end{definition}

Intuitively, the adiabatic function corresponds to the minima one would converge to by increasing $\delta t$ infinitely slowly and minimizing $\mf(\vtheta,\delta t)$ by gradient descent with a very small learning rate. By analogy, one can imagine dropping a marble in the initial minima and then slowly modifying the landscape by increasing $\delta t$. The position of the marble would correspond to our adiabatic minima and in practice it is where we expect our algorithm to converge.

Up to this point, there are two caveats that we would like to highlight. First, this adiabatic minimum is not necessarily the global minimum (as discussed in Section~\ref{sec:minimum-jump} - there could potentially be a jump in the global minimum). The other subtlety is that the existence of the adiabatic minimum is not always guaranteed for increasing $\delta t$. This is highlighted in Figure \ref{fig:disapearing_adiabatic}. While for a small time-step one intuitively expects to have the adiabatic minimum, it is not certain whether we have this for large time-steps. That is, the adiabatic function can cease to be continuous beyond $T$ (and in practice we do not in general know what $T$ is). Crucially, the discontinuity in the adiabatic minimum path implies that zero gradients now turn into some slopes. Hence, the lack of a continuous adiabatic minimum does not necessarily imply untrainability.

\begin{figure}
    \centering
    \includegraphics[width = 0.45\textwidth]{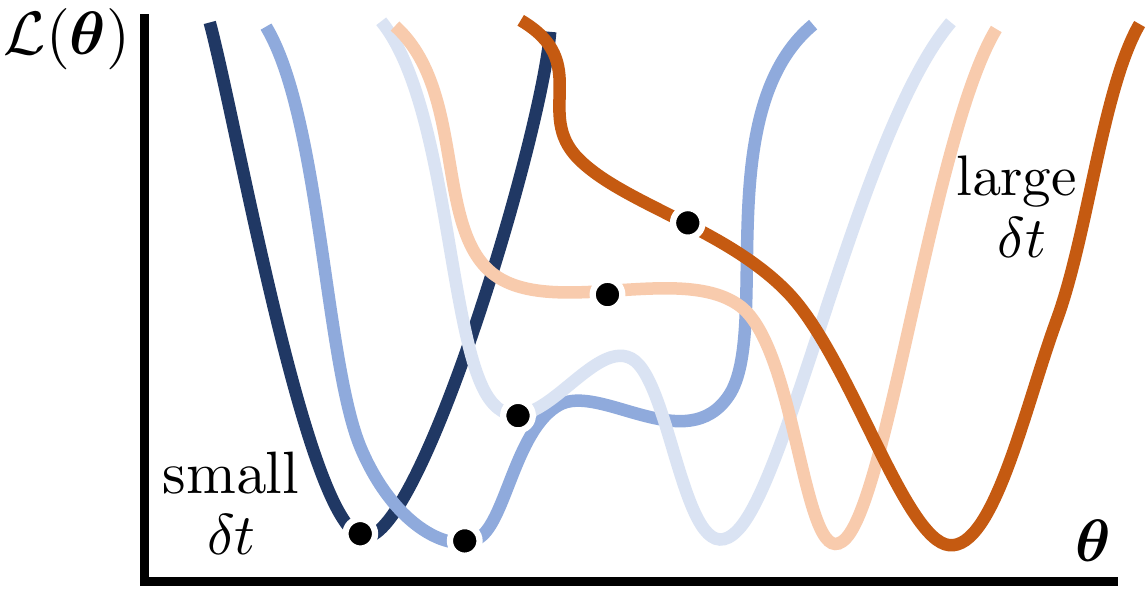}
    \caption{\textbf{Adiabatic minima.} Here we show the adiabatic minima (highlighted in the black dots) as a function of the time-step $\delta t$. From dark blue to orange, we highlight how the loss function evolves with increasing $\delta t$. Indeed, when the time-step increases the adiabatic minima stops being the global minima. Then it turns into a saddle point and finally disappears to become a slope. As mentioned in the main text, when the adiabatic minima disappears it turns into a slope.}
    \label{fig:disapearing_adiabatic}
\end{figure}

With these caveats in mind, we proceed under the assumption that the adiabatic minimum exists within the time-step of our interest. We first present Proposition~\ref{prop:adiabatic-min} which shows that the shift in the adiabatic minimum can be bounded with the time-step.

\begin{proposition}\label{prop:adiabatic-min}
Given a time-step of the current iteration $\delta t$ and assuming that the adiabatic minimum exists within this time frame, the shift of the adiabatic minimum $\valpha_{A}(\delta t)$ as defined in Definition~\ref{def:adaibatic-shift} can be bounded as 
\begin{align}\label{eq:adiabatic-shift-prop-appx}
    \| \valpha_A(\delta t) \|_2 \leq \frac{2\sqrt{M} \lambda_{\rm max} \delta t}{\beta_A} \;,
\end{align}
where $M$ is the number of parameters, $\lambda_{\rm max}$ is the largest eigenvalue of the dynamic Hamiltonian $H$ and $\beta_A = \frac{\Dot{\valpha}_A^T(\delta t) \left( \nabla^2_{\valpha} \LC(\valpha,\delta t) \big|_{\valpha = \valpha_A(\delta t)}\right) \Dot{\valpha}_A(\delta t) }{\| \Dot{\valpha}_{A}(\delta t) \|_2^2} \;\;$. 
\end{proposition}
\begin{proof}
We note that to improve readability of the proof it is more convenient here to use $t$ to refer as a time-step (instead of $\delta t$ as in other sections). 
We recall from Definition~\ref{def:adaibatic-shift} that the adiabatic shift can be expressed as
\begin{align}
    \nabla_{\valpha} \LC(\valpha,t) \big|_{\valpha = \valpha_A(t)} = \vec{0} \;,
\end{align}
which holds for \textit{any} $t$. For convenience, we denote $\nabla_{\valpha_A} \LC := \nabla_{\valpha} \LC(\valpha,t) \big|_{\valpha = \valpha_A(t)}$. By a direct differentiation with respect to $t$, this leads to
\begin{align}\label{eq:moving-min-2-proof1}
    \frac{d}{dt}\left(\nabla_{\valpha_A} \LC \right) = \partial_t \nabla_{\valpha_A} \LC + \left(\nabla_{\valpha_A}^2 \LC  \right)\Dot{\valpha}_A = \vec{0} \;,
\end{align}
where we denote $\partial_t = \partial/\partial t$ and $\Dot{\valpha}_A  = d \valpha_A(t)/dt$. We remark that $\left(\nabla_{\valpha_A}^2 \LC  \right)\Dot{\valpha}_A $ is a matrix vector multiplication with $\nabla_{\valpha_A}^2 \LC$ a Hessian matrix evaluated at $\valpha_{A}(t)$. By multiplying with $\frac{\Dot{\valpha}_A^T}{\|\Dot{\valpha}_A\|_2}$ from the left, we have
\begin{align}
    \frac{\Dot{\valpha}_A^T}{\|\Dot{\valpha}_A \|_2} \partial_t \nabla_{\valpha_A} \LC  + \frac{\Dot{\valpha}_A^T \left(\nabla_{\valpha_A}^2 \LC  \right)\Dot{\valpha}_A}{\|\Dot{\valpha}_A \|_2^2} \|\Dot{\valpha}_A \|_2 = 0 \;.
\end{align}
For convenience, we denote $\beta_A := \Dot{\valpha}_A^T \left(\nabla_{\valpha_A}^2 \LC  \right)\Dot{\valpha}_A / \| \Dot{\valpha}_A \|_2^2 $. 
By rearranging the terms, we can bound  the norm of $\Dot{\valpha}_A(t)$ as
\begin{align}
    \| \Dot{\valpha}_A(t) \|_2 & =  \left\| -  \frac{\Dot{\valpha}_A^T \partial_t \nabla_{\valpha_A} \LC}{\|\Dot{\valpha}_A \|_2 \beta_A} \right\|_2 \\
    & \leq \frac{\| \partial_t \nabla_{\valpha_A} \LC\|_2}{\beta_A} \\
    & \leq \frac{\sqrt{M} \left|\partial_t \partial_{\alpha^{(i)}_{A}} \LC \right|_{\rm max}}{\beta_{A}} \\
    & \leq \frac{2\sqrt{M} \lambda_{\rm max}}{\beta_A} \;,
\end{align}
where the first inequality is due to Cauchy-Schwarz inequality, in the second inequality we expand the 2-norm out explicitly and take the largest value in the sum i.e., $ \| \vec{a} \|_2 = \sqrt{\sum_{i=1}^M a_i^2} \leq \sqrt{M} |a_i|_{\rm max}$ with $\alpha_A^{(i)}$ being the $i^{\rm th}$ component of $\valpha_{A}$. To reach the last inequality, we recall that the loss function is of the form $\LC(\valpha,t) = 1 - \Tr\left( e^{-iHt} \rho_{\thv^*} e^{i Ht} \rho_{\vtheta}\right)$ where $H$ is the dynamic Hamiltonian, $\rho_{\thv^*}$ is the state corresponding to the solution of the previous iteration and $\rho_{\vtheta}$ is the parametrized state that depends on $\valpha$ and respects a parameter shift's rule. We can then bound the quantity of interest as
\begin{align}
    \partial_t \partial_{\alpha^{(i)}_{A}} \LC  & = \frac{\partial}{\partial t}\left( \frac{\partial}{\partial \alpha^{(i)}}  \LC(\valpha, t)\right)\bigg|_{\valpha = \valpha_A(t)}  \\
    & = \frac{1}{2} \left(\frac{\partial}{\partial t} \LC \left(\valpha_A + \frac{\pi}{2} \hat{\alpha}_i\right) - \frac{\partial}{\partial t} \LC \left(\valpha_A - \frac{\pi}{2} \hat{\alpha}_i\right)\right) \\
    & = \frac{1}{2} \left( \Tr\left( i [H, e^{-iHt}\rho_{\thv^*} e^{i Ht}] \rho_{\valpha_{A,+}}\right) - \Tr\left( i [H, e^{-iHt}\rho_{\thv^*} e^{i Ht}] \rho_{\valpha_{A,-}}\right) \right) \\
    & \leq \left\|  [H, e^{-iHt}\rho_{\thv^*} e^{i Ht}]\right\|_{\infty} \\
    &\leq 2 \lambda_{\rm max} \;,
\end{align}
where the second equality is due to a parameter shift's rule, in the third equality we perform the direct differentiation with $t$ and denote $\valpha_{A,+} = \valpha_A + \frac{\pi}{2} \hat{\alpha}_i$ as well as  $\valpha_{A,-} = \valpha_A - \frac{\pi}{2} \hat{\alpha}_i$. The first inequality is due to the triangle inequality followed by the H\"{o}lder's inequality with the fact that $\| \rho \|_1 =1 $ for any pure quantum state $\rho$. In the last inequality, we use $\| i [A,B] \|_p = 2 \|A\|_p \|B\|_p$ and the unitary invariance of the p-norm $\| U A\|_p = \|A\|_p$ as well as $\|\rho \|_{\infty} = 1$ for any pure quantum state $\rho$.

Lastly, we can bound the shift of the adiabatic minimum as
\begin{align}
    \left\| \valpha_A(t) \right\|_2 & = \left\| \int_{0}^t \Dot{\valpha}_A(\tau) d\tau\right\|_2 \\
    & \leq \int_{0}^t \| \Dot{\valpha}_A(\tau) \|_2 d\tau \\
    & \leq \frac{2\sqrt{M} \lambda_{\rm max} t}{\beta_A} \;,
\end{align}
which completes the proof.

\end{proof}

Note that for the bound in Eq.~\ref{eq:adiabatic-shift-prop-appx} to be informative about the asymptotic scaling we require $\beta_A$ to be at least polynomially small.

In order for the training to have convergence guarantees to the adiabatic minimum, we need to ensure that our time-step is small enough 
so that our \textit{adiabatic minima} is inside of the trainable and convex regions. This is formalized in Theorem~\ref{thm:adiabaticminimum} which is presented in the main text and proved here.

\begin{theorem}[Adiabatic minimum is within provably `nice' region, Formal]\label{th:min-in-convexregion}
If the time-step $\delta t$ is chosen such that
\begin{align}
    \delta t \leq \frac{\eta_0 \beta_A }{2 M \lambda_{\rm max}} \;,
\end{align}
with some small constant $\eta_0$, then the adiabatic minimum $\thv_A(\delta t)$ is guaranteed to be within the non-vanishing gradient region (as per Theorem~\ref{thm:variance-lower-bound}), and additionally, if $\delta t$ is chosen such that
\begin{align}
    \delta t \leq \frac{\beta_A(\mu_{\rm min} + 2|\epsilon|)}{32 \lambda_{\max}M^{5/2} \left(1 + \frac{\beta_A}{2M^{3/2}} \right)} \;.
\end{align}
then the adiabatic minimum $\thv_A(\delta t)$ is guaranteed to be within the $\epsilon$-convex region (as per Theorem~\ref{thm:convex}) where 
\begin{equation}
    \beta_A := \frac{\Dot{\vtheta}_A^T(\delta  t) \left( \nabla^2_{\vtheta} \LC(\vtheta) \big|_{\vtheta = \vtheta_A(\delta  t)}\right) \Dot{\vtheta}_A(\delta t) }{\| \Dot{\thv}_{A}(\delta t) \|_2^2} 
\end{equation}
corresponds to the second derivative of the loss in the direction in which the adiabatic minimum moves.
\end{theorem}

\begin{proof}
From Proposition~\ref{prop:adiabatic-min} and the norm inequality, the adiabatic minimum follows
\begin{align}
    \| \valpha_A(\delta t) \|_\infty \leq \| \valpha_A(\delta t) \|_2 \leq \frac{2\sqrt{M} \lambda_{\rm max} \delta t}{\beta_A} \;.
\end{align}

What we want now is to incorporate the conditions of the regions of interest. That is, by fine-tuning $\delta t$, we want a guarantee that the new minimum is within (i) the non-vanishing gradient region and (ii) the convex region. 

\medskip

\underline{For (i) the non-vanishing gradient region,} we recall from the formal version of Theorem~\ref{oldsimplifiedR} that given the time-step bounded as $\delta t \leq 1/2\lambda_{\rm max}$, the hypercube of width $2r$ has the substantial non-vanishing gradients when $r$ follows 
\begin{align}
    r  = \frac{\eta_0}{\sqrt{M}} \;,
\end{align}
with some constant $\eta_0$. Then, it is sufficient to have the guarantee that the adiabatic minimum is inside this region by imposing
\begin{align}
    \| \valpha_A(\delta t) \|_\infty  \leq \frac{2\sqrt{M} \lambda_{\rm max} \delta t}{\beta_A} \leq  \frac{\eta_0}{\sqrt{M}}\;,
\end{align}
which leads to
\begin{align}
    \delta t \leq \frac{\eta_0 \beta_A }{2 M \lambda_{\rm max}} \;.
\end{align}

\medskip

\underline{For (ii) the convex region,} from Theorem~\ref{th:pRTE}, recall that given $\delta t \leq \frac{\mu_{\rm min} + 2 |\epsilon| }{16 M \lambda_{\rm max}}$, we have an $\epsilon$-convex hypercube region of width $2r_c$ such that
\begin{align}
    r_c \geq \frac{1}{M}\left(\frac{\mu_{\rm min}+2|\epsilon|}{16 M} - \lambda_{\rm max} \delta t\right) \;.
\end{align}
Therefore, it is sufficient to guarantee that the adiabatic minimum is inside this convex region by imposing
\begin{align}
    \| \valpha_A(\delta t) \|_\infty \leq \frac{2\sqrt{M} \lambda_{\rm max} \delta t}{\beta_A} \leq \frac{1}{M}\left(\frac{\mu_{\rm min}+2|\epsilon|}{16 M} - \lambda_{\rm max} \delta t\right) \leq r_c \;.
\end{align}
Upon rearranging the terms, the time-step is bounded as
\begin{align} \label{eq:dt-convex-min}
    \delta t \leq \frac{\beta_A(\mu_{\rm min} + 2|\epsilon|)}{32 \lambda_{\max}M^{5/2} \left(1 + \frac{\beta_A}{2M^{3/2}} \right)} \;.
\end{align}
We remark that this bound in Eq.~\eqref{eq:dt-convex-min} is much tighter than the bound in Theorem~\ref{th:pRTE} (specified above). That is, to have such a guarantee, $\delta t$ is relatively shorter.

\end{proof}

\section{Imaginary Time Evolution}\label{app:imaginary}
\subsection{Framework}
The variational imaginary time evolution can be used to prepare ground states and thermal states~\cite{Yuan_2019,McArdle_2019,Garcia-Ripoll,Wofl2015}. In this section, we focus on a variational \revadd{time-evolution} compression version of imaginary time evolution~\cite{benedetti2020hardware}. Most steps of the algorithm are identical to the real-time version described in the main text except we substitute $\delta t\to i \delta \tau$. This leads to $U = e^{i\delta t H}\to \mathcal{U} = e^{-\delta \tau H}$. One key technicality is to add a constraint such that the state after the evolution is forced to be normalised. That is, we have 
\begin{equation}
    |\psi_{\delta \tau}\rangle = \frac{1}{\sqrt{Z}}e^{-\delta \tau H}\ket{\psi}
\end{equation}
where $Z =\bra{\psi}e^{-2\tau H}\ket{\psi}$. We remark that $|\psi_{\delta \tau}\rangle$ is now a valid \textit{pure} quantum state and thus there is a unitary that prepares it.
To have any chance of preparing the ground state via imaginary time evolution the initial state $\ket{\psi}$ must have a non-vanishing overlap with the ground state. If the initial state is instead a maximally entangled state, and imaginary time evolution is applied to only half the Bell state, then this approach can be used to prepare a thermal double field state (and thereby a thermal state). However, in what follows we will focus on ground state preparation.

In the variational \revadd{time-evolution} compression approach for imaginary time evolution, one aims to iteratively learn $|\psi_{\delta \tau}\rangle$ with a parametrized quantum circuit $U(\thv)$ (with parameters initialized around the optimal parameter values obtained from the previous iteration). More explicitly, the loss function at each iteration is of the form
\begin{equation}
    \mf_{\rm ITE}(\thv) = \Tr\left[U(\thv)\rho U^\dagger(\thv)\frac{1}{Z}e^{-\delta\tau}U(\vtheta^*)\rho U^\dagger(\vtheta^*)e^{-\delta\tau}  \right] \;,
\end{equation}
where $\thv^*$ are the optimal parameters from the previous iteration. For details on how to compute this loss in practise see Ref.~\cite{benedetti2020hardware}. 
We further suppose that the parametrised quantum circuit is of the same form used in the real-time case i.e.,
\begin{equation}
	U\left( \vtheta \right) = \prod_{i=1}^M V_i U_i(\theta_i)
\end{equation}
where $\left\{ V_i \right\}_{i=1}^M$ are a set of fixed unitary matrices, $\left\{ U_i(\theta_i) = e^{-i\theta_i\vsigma_i} \right\}_{i=1}^M$, are the parameter-dependant rotations, and $\{ \vsigma_i \}_{i=1}^M$ is a set of gate generators such that $\sigma_i^2 = \1$ e.g., Pauli strings on $n$ qubits. Crucially, it is natural to consider the parameter initialization around $\thv^*$ i.e., 
\begin{align}
    \thv \sim \uni(\thv^*, r) \;,
\end{align}
with $r$ being some small perturbation.

\subsection{Summary of analytical results for imaginary time evolution}

Here we summarize analytical results similar to those obtained for the real time evolution scenario. These include the existence of the non-vanishing gradient region with the warm-start initialization, the guarantee of the approximate convex region as well as the analysis on the adiabatic minimum. 
We note that the derivations of these results follow the same steps as in the case of the real-time evolution and are provided in Appendix~\ref{app:ITE-proof} for the completeness.

First, we show that the region around the optimal solution of the previous iteration exhibits substantial gradients. More precisely, the following theorem demonstrates the polynomial large variance of the loss function within in a small hypercube around the starting point. This theorem is similar to Theorem~\ref{thm:variance-lower-bound} (as discussed in Section~\ref{sec:variance} of the main text). 
\begin{theorem}[Lower-bound on the loss variance for imaginary time evolution, Informal]\label{thm:ITE-variance-lower-bound} 
Assume \revadd{an} initial state \revadd{$\rho_0$} and let us choose $\vsigma_1$ such that $\Tr[\rho_0 \vsigma_1 \rho_0 \vsigma_1] = 0$. Given that the imaginary time-step scales as
$\delta \tau \leq  \frac{1}{\sqrt{24}\lambda_{\rm max}}$ 
where $\lambda_{\rm max}$ is the largest eigenvalue of $H$ 
and we consider a hypercube of width $2r$ such that 
\begin{align}
 r = \Theta\left(\frac{1}{\sqrt{M}}\right) \;,
\end{align}
the variance at any iteration of the variational compression algorithm is lower bounded as 
\begin{align}
    \Var_{\vtheta \sim \uni(\vtheta^*, r)} \left[ \LC_{\rm ITE} (\vtheta)\right] \in \Omega\left( \frac{1}{M}\right) \;.
\end{align}
Thus, for $M \in \OC(\poly(n))$, then we have
\begin{align}
    \Var_{\vtheta \sim \uni(\vtheta^*, r)} \left[ \LC_{\rm ITE} (\vtheta)\right] \in \Omega\left( \frac{1}{\poly(n)}\right) \;.
\end{align}\end{theorem}

Next, we can ensure that for a sufficiently small time-step the loss landscape around the optimal solution of the previous iteration is $|\epsilon|$-convex. We refer the readers to Appendix~\ref{app:epsilonconvexity} for the definition of the $|\epsilon|$-convexity. This result is an imaginary time evolution version of Theorem~\ref{thm:convex} discussed in Section~\ref{sec:convexity}. 
\begin{theorem}[Approximate convexity of the landscape for imaginary time evolution, Informal.]\label{thm:ITE-convex-in}
For a time-step of size
\begin{equation}
\delta\tau\in\order{\frac{\mu_{\min}+2|\epsilon|}{M\lambda_{\max}}}
\end{equation}
the loss landscape is $\epsilon$-convex in a hypercube of width $2r_c$ around a previous optimum $\vtheta^*$ i.e., $\vol(\thv^*, r_c)$ such that
\begin{equation}
    r_c\in\Omega\left( \frac{\mu_{\min}+2|\epsilon|}{16M^2}-3\lambda_{\max} \delta\tau \right)
\end{equation}
where $\mu_{\min}$ is the minimal eigenvalue of the Fisher information matrix associated with the loss at $\vtheta^*$.

\end{theorem}

A similar result on the adiabatic minimum in Theorem~\ref{thm:adiabaticminimum} can also be analytically obtained. Assuming the adiabatic minimum exists within our time interval of interest, we present below the scaling of imaginary time-step such that the adiabatic minimum is in the region with substantial gradients and in the convex region. We refer the readers to Section~\ref{sec:adiabaticminimum} and Appendix~\ref{app:moving-min} for the refresher of the adiabatic minimum.
\begin{theorem}[Adiabatic minimum is within provably `nice' region for imaginary time evolution, Informal]\label{thm:ITE-adiabaticminimum}
If the imaginary time-step $\delta \tau$ is chosen such that
\begin{align}
    \delta \tau \in \OC\left(\frac{\beta_A}{M \lambda_{\rm max}}\right) \;,
\end{align}
then the adiabatic minimum $\thv_A(\delta \tau)$ is guaranteed to be within the non-vanishing gradient region (as per Theorem~\ref{thm:ITE-variance-lower-bound}), and additionally, if $\delta \tau$ is chosen such that
\begin{equation}
\delta\tau \in \OC\left(  \frac{\beta_A (\mu_{\min} + 2|\epsilon|)}{M^{5/2}\lambda_{\max}} \right)
\end{equation}
then the adiabatic minimum $\thv_A(\delta \tau)$ is guaranteed to be within the $\epsilon$-convex region (as per Theorem~\ref{thm:ITE-convex}) where 
\begin{equation}
    \beta_A := \frac{\Dot{\vtheta}_A^T(\delta  \tau) \left( \nabla^2_{\vtheta} \LC_{\rm ITE}(\vtheta) \big|_{\vtheta = \vtheta_A(\delta  \tau)}\right) \Dot{\vtheta}_A(\delta \tau) }{\| \Dot{\thv}_{A}(\delta \tau) \|_2^2} 
\end{equation}
corresponds to the second derivative of the loss in the direction in which the adiabatic minimum moves.
\end{theorem}

Finally, it is crucial to note that the discussion about the minimum jumps and the potential existence of the fertile valley in Section~\ref{sec:minimum-jump} is also applicable to imaginary time evolution.

\subsection{Proof of analytical results}\label{app:ITE-proof}
In this section, we analytically derive the analytical results presented in the previous sub-section. Again, these derivations are identical to the ones presented in Appendix~\ref{app:variance}, Appendix~\ref{app:convexity} and Appendix~\ref{app:moving-min}. We present them again here for completeness and the readers are also encouraged to look at those relevant appendices. 

\subsubsection{Bound on the variance of the landscape: Proof of Theorem~\ref{thm:ITE-variance-lower-bound}}

First, we introduce the equivalent version of Lemma \ref{lem:bound_fidelity} for imaginary time evolution. 

\begin{lemma}\label{lem:bound_fidelity_imaginary}
The fidelity between two pure states $\rho$ and $ \rho_\tau = \frac{1}{Z}e^{- H \tau  }\rho e^{-H  \tau }$ with $Z = \Tr \left( e^{- H \tau  }\rho e^{-H  \tau } \right)$ can be upper bounded as
\begin{equation}
    F \left[\rho, \rho_\tau\right] \geq  1 - 12 \lambda_{\rm max}^2 \tau^2 
\end{equation}

where $\lambda_{\rm max}$ is the largest eigenvalue of $H$.
\end{lemma}

\begin{proof}
First, the derivative of the loss function with respect to time can be written as 
\begin{equation}\label{eq:im_time_ev}
    \frac{d \rho_\tau}{dt} = - \{\rho_\tau, H\} + 2\Tr(H\rho_\tau)\rho_\tau \;,
\end{equation}
where $\{\cdot, \cdot\}$ is an anti-commutator. 

Now we can use a Taylor's expansion around $\tau = 0$, and then the fidelity is of the form
\begin{equation}
    F \left[\rho, \rho_\tau\right] = 1 + \frac{\tau^2}{2}\left( \frac{d^2 F \left[\rho, \rho_\tau\right]}{d\tau^2} \right)\bigg|_{\tau = \tau'}
    \label{eq:derivativeIm}
\end{equation}
where the zero order term is 1, the first order term is zero by a direct computation and the second order term is evaluated at $\tau'\in[0,\tau]$ by Taylor's remainder (see Theorem~\ref{thm:taylor}). Thus we can bound the second derivative as follows. 

\begin{align}
    \left( \frac{d^2 F \left[\rho, \rho_\tau\right]}{d\tau^2} \right)\bigg|_{\tau = \tau'}  =\;& 2 \text{Re}\left[\Tr \left( \rho\rho_\tau H^2 \right)\right] + 2\Tr(\rho H \rho_\tau H) - 8 \Tr(H\rho_\tau)\text{Re}[\Tr(\rho\rho_\tau H)]- 4 \Tr \left(H^2\rho_\tau\right)\Tr(\rho\rho_\tau) \\
    &+ 8 \Tr(\rho_\tau H)^2\Tr(\rho_\tau\rho)\nonumber\\
    \leq \;& 24\lambda_{\max}^2 \label{eq:ite-proof1}
\end{align}
where in the inequality is due to (i) Hölder's inequality, (ii) $\| \rho \|_1 = 1$ for any pure quantum state $\rho$, (iii) $\| A B \|_{p} \leq \|A\|_p \|B\|_p$. More precisely, we used the following bounds to reach Eq.~\eqref{eq:ite-proof1}
\begin{align}
    &\Tr \left( \rho\rho_\tau H^2\right)\leq ||\rho||_1 ||\rho_\tau H^2||_\infty\leq \lambda_{\max}^2\\
    &\Tr(\rho H \rho_\tau H)\leq  ||\rho||_1 ||\rho_\tau H^2||_\infty\leq\lambda_{\max}^2\\
    &\Tr(\rho\rho_\tau H)\leq  ||\rho||_1 ||\rho_\tau H||_\infty\leq \lambda_{\max}\\
    &\Tr(H\rho_\tau)\leq  \lambda_{\max}\\
    &\Tr \left(H^2\rho_\tau\right)\leq  \lambda_{\max}^2\\
    &\Tr(\rho_\tau\rho)\leq  1
\end{align}

Lastly, we can just substitute Eq.~\eqref{eq:ite-proof1} back to \eqq{derivativeIm} and obtain
\begin{equation}
    F \left[\rho, \rho_\tau\right] \geq 1 - 12\tau^2\lambda_{\max}^2 \;,
\end{equation}
which completes the proof.
\end{proof}

\medskip

We now present the formal version of Theorem~\ref{thm:ITE-variance-lower-bound} and provide a detailed proof.
\setcounter{theorem}{3}
\begin{theorem}[Lower-bound on the loss variance for imaginary time evolution, Formal]\label{simplifiedI}
Assume \revadd{an} initial state  and let us choose $\vsigma_1$ such that $\Tr[\rho_0 \vsigma_1 \rho_0\vsigma_1] = 0$. Given that the time-step $\delta \tau$ respects 
\begin{align}
    \frac{1}{\sqrt{24} \lambda_{\rm max}} \geq \delta \tau \;,
\end{align}
where $\lambda_{\rm max}$ is the largest eigenvalue of $H$ as well as the perturbation $r$ follows
\begin{equation}\label{eq:perturb-conditionIm}
        r^2\leq \frac{3r_0^2 (1-24\lambda^2_{\max} \delta\tau^2)}{2(M-1)(1-12\lambda^2_{\max} \delta\tau^2)} \ , 
\end{equation}
with some $r_0$ such that $0< r_0 <1$, then the variance of the loss function within the region $\vol(\thv^*, r)$ is lower bounded as
\begin{align}
    \Var_{\vtheta \sim \uni(\vtheta^*, r)} \left[ \LC_{\rm ITE} (\vtheta)\right] \geq \frac{4r^2}{45}\left[(1-r_0)(1-24\lambda_{\max}^2 \delta\tau^2) \right]^2 \;.
\end{align}
In addition, by choosing $r$ such that $r\in\Theta\left(\frac{1}{\sqrt{M}}\right)$, then we have
\begin{align}
     \Var_{\vtheta \sim \uni(\vtheta^*, r)} \left[ \LC_{\rm ITE} (\vtheta)\right] \in \Omega\left( \frac{1}{M}\right) \;.
\end{align}
\end{theorem}

\begin{proof}
First, we note that Proposition~\ref{thm:variance-lower-bound} in Appendix \ref{app:variance} also applies for the imaginary time evolution. This is since all the proof steps in Proposition~\ref{thm:variance-lower-bound} hold when replacing $i\delta t \rightarrow \delta \tau$ and $\LC(\thv) \rightarrow \LC_{\rm ITE}(\thv)$. Hence, this proof starts by recalling Proposition~\ref{thm:variance-lower-bound}
\begin{equation}\label{eq:lowerbound_pre_imaginary}
    \Var_{\thv\sim\uni(\thv^*, r) }[\mf_{\rm ITE}(\thv)]\geq  \ (c_+ - k_+^2) \min_{\Tilde{\xi}\in [-1,1]} \left( k_+^{M-1} \Delta_{\thv^*}  + (1-k_+^{M-1})\Tilde{\xi} \right)^2 ,
\end{equation}
where the quantities in the bound above are
\begin{align}
    &c_+     :=  \mathbb{E}_{\alpha \sim \uni(0,r)} [ \cos^4{\alpha}]\;, \\
    &k_+   := \mathbb{E}_{\alpha \sim \uni(0,r)} [ \cos^2{\alpha}] \; ,                                                               \\
    &\Delta_{\thv^*} := \Tr[(\rho_0 -\vsigma_1 \rho_0 \vsigma_1) U^{\dagger}\left(\vec{\theta^*}\right)\rho_{(\vtheta^*,\delta \tau)} U\left(\vec{\theta^*}\right)]  \; ,\\
    &\rho_{(\vtheta^*,\delta\tau)} = \frac{1}{Z}e^{- H \delta\tau  }U\left(\vtheta^*\right)\rho_0U^\dagger\left(\vtheta^*\right) e^{-H  \delta\tau }\;,\\
    &Z = \Tr\left[ e^{- H \delta\tau  }U\left(\vtheta^*\right)\rho_0U^\dagger\left(\vtheta^*\right) e^{-H  \delta\tau } \right]\;.
\end{align}
with $\vsigma_1$ being the first (non commuting) gate of the ansatz and $\rho_0$ is the initial state. 

\medskip

Importantly, we notice that if the perturbation $r$ is chosen such that the following condition is satisfied
\begin{equation}
    k_+^{M-1} \Delta_{\thv^*}  \geq(1-k_+^{M-1})\label{eq:necessary_conditionIm}
\end{equation}
then $\hat{\xi}=-1$ minimises the variance lower bound in Eq.~\eqref{eq:lowerbound_pre_imaginary} which leads to 
\begin{equation}\label{eq:lowerbound_pre_imaginary2}
    \Var_{\thv\sim\uni(\thv^*, r) }[\mf(\thv)]\geq  \ (c_+ - k_+^2) \left( k_+^{M-1} (\Delta_{\thv^*}+1)-1 \right)^2 ,
\end{equation}

Now, we focus now on $\Delta_{\thv^*}$ which can be expressed as
\begin{align}\label{eq:ite-delta-big}
    \Delta_{\thv^*} :=& \Tr[(\rho_0 -\vsigma_1 \rho_0 \vsigma_1) U^{\dagger}\left(\vec{\theta^*}\right)\rho_{(\vtheta^*,\delta \tau)} U\left(\vec{\theta^*}\right)] \\
    =& F(\rho_{(\vtheta^*,0)},\rho_{(\vtheta^*,\delta\tau)}) - \Tr\left[\vsigma_1 \rho_0 \vsigma_1U^{\dagger}\left(\vec{\theta^*}\right)\rho_{(\vtheta^*,\delta \tau)} U\left(\vec{\theta^*}\right)\right]
\end{align}
where we have the fidelity $F(\rho_{(\vtheta^*,0)},\rho_{(\vtheta^*,\delta\tau)}) = \Tr(\rho_{(\vtheta^*,0)}\rho_{(\vtheta^*,\delta\tau)})$. \revadd{Assume we can choose $\vsigma_1$} such that $\Tr[\rho_0 \vsigma_1 \rho_0 \vsigma_1] = 0$ which also implies
\begin{align}
    \Tr\left[ \rho_{(\thv^*, 0)} U(\thv^*)\vsigma_1 \rho_0 \vsigma_1 U^\dagger(\thv^*) \right]   = \Tr\left[\rho_0 \vsigma_1 \rho_0 \vsigma_1 \right] = 0 \;.
\end{align}
Then, we define an orthonormal $\left\{\phi_i := \ketbra{\phi_i}\right\}_{i=1}^{2^n}$ basis as 
\begin{align}
    \phi_1 &= \rho_{(\vtheta^*,0)}\ , \\
    \phi_2 &= U\left(\vtheta^*\right)\vsigma_1\rho_0\vsigma_1U^\dagger\left(\vtheta^*\right) \;,
\end{align}
and other $\{ \phi_i \}$ are necessary orthornormal states to complete the basis. We can upper bound the second term on the right hand side of Eq.~\eqref{eq:ite-delta-big} as
\begin{align}
    \Tr\left[\vsigma_1 \rho_0 \vsigma_1U^{\dagger}\left(\vec{\theta^*}\right)\rho_{(\vtheta^*,\delta \tau)} U\left(\vec{\theta^*}\right)\right] =& \Tr\left(\phi_2\rho_{(\vtheta^*,\delta \tau)} \right)\\
    \leq& \sum_{i=2}^{2^n}\Tr\left(\phi_i\rho_{(\vtheta^*,\delta \tau)} \right)\\
    =& \Tr\left[(\1-\phi_1)\rho_{(\vtheta^*,\delta \tau)} \right]\label{eq:1minusid}\\
    =& F(\rho_{(\vtheta^*,0)},\rho_{(\vtheta^*,\delta\tau)}) -1
\end{align}
where in the inequality we add terms corresponding to other basis which holds because the trace of positive matrices is positive and in \eqq{1minusid} we use the fact that $\sum_i\phi_i = \1$. 

With this we can lower-bound $\Delta_{\vtheta^*}$ as follows
\begin{align}
    \Delta_{\vtheta^*}&\geq 2F(\rho_{(\vtheta^*,0)},\rho_{(\vtheta^*,\delta\tau)}) -1 \\ 
    & \geq 1-24 \lambda_{\max}^2\delta\tau^2 \label{eq:ite-delta-bound}
\end{align}
where we have used Lemma \ref{lem:bound_fidelity_imaginary} in the last inequality. We remark that the bound in Eq.~\eqref{eq:ite-delta-bound} can be equivalently expressed as
\begin{align}\label{eq:ite-delta-bound2}
    \frac{1}{2-24 \lambda_{\max}^2\delta\tau^2} \geq \frac{1}{1+\Delta_{\vtheta^*}} \;.
\end{align}
Importantly, for this bound to be informative (i.e., non-negative), we require the constrain on $\delta \tau$ as
\begin{align}
    \frac{1}{\sqrt{24} \lambda_{\rm max}} \geq \delta \tau \;.
\end{align}

\medskip

In this next step, we show how the condition in \eqq{necessary_conditionIm} can be fulfilled. We note that the condition can be equivalently expressed as 
\begin{equation}\label{eq:ite-key-con}
    k_{+}^{M-1}\geq \frac{1}{1+\Delta_{\vtheta^*}}
\end{equation}

We consider the bound of $ k_+^{M-1}$ which follows as
\begin{align}
    k_+^{M-1} & = \left(\frac{1}{2r} \int_{-r}^{r} d\alpha \cos^2(\alpha) \right)^{M-1} \\
    & = \left( \frac{1}{2} + \frac{\sin{(2r)}}{4r}\right)^{M-1} \\
    & \geq  \left( 1 - \frac{r^2}{3}\right)^{M-1} \\
    & \geq 1 - \frac{(M-1)r^2}{3}  \label{eq:ite-proof-bound-kplus} \\
    & > 1 - \frac{(M-1)r^2}{3r_0^2} \;, \label{eq:ite-proof-coro1-0}
\end{align}
where the first inequality is by directly expanding the base and keeping only the second order term, the second inequality is due to Bernoulli's inequality and finally the last inequality holds because $0 <r_0 < 1$. 

With this we can see that the condition in Eq.~\eqref{eq:ite-key-con} holds if the right hand side of Eq.~\eqref{eq:ite-proof-coro1-0} is larger then the left hand side of Eq.~\eqref{eq:ite-delta-bound2} i.e.,
\begin{equation}
    1-\frac{(M-1)r^2}{3r_0^2}>\frac{1}{2-24 \lambda_{\max}^2\delta\tau^2} \;\;\Rightarrow\;\; k_{+}^{M-1}\geq \frac{1}{1+\Delta_{\vtheta^*}} \;.
    \label{eq:boundrim}
\end{equation}

By rearranging \eqq{boundrim}, we find the bound on the perturbation $r$ as \begin{equation}\label{eq:ite-r-cond}
 r^2\leq \frac{3r_0^2 (1-24\lambda^2_{\max} \delta\tau^2)}{2(M-1)(1-12\lambda^2_{\max} \delta\tau^2)} 
\end{equation}. 

By enforcing the condition of the perturbation in Eq.~\eqref{eq:ite-r-cond}, the variance can be lower-bounded further from Eq.~\eqref{eq:lowerbound_pre_imaginary2} as
\begin{align}
    \Var_{\thv\sim\uni(\thv^*, r) }[\mf(\thv)]
    \geq & \ (c_+ - k_+^2) \left[ k_+^{M-1} (\Delta_{\thv^*}+1)-1 \right]^2 \\
    \geq& (c_+ - k_+^2)  \left[\left( 1 - \frac{(M-1)r^2}{3}\right)(2 - 24\lambda_{\max}^2\delta\tau^2) - 1\right]^2 \\
    \geq & \ (c_+ - k_+^2) \left[(1-r_0^2)(1-24\lambda_{\max}^2\delta\tau^2)\right]^2\\
    \geq& \frac{4r^4}{45}\left[(1-r_0^2)(1-24\lambda_{\max}^2\delta\tau^2) \right]^2 \;,
\end{align}
where the second inequality is due to  Eq.~\eqref{eq:ite-delta-bound} and Eq.~\eqref{eq:ite-proof-bound-kplus}, the third inequality is by the condition on $r$ in Eq.~\eqref{eq:ite-r-cond}. To reach the last inequality, we directly bound $c_+ - k_+^2  = \frac{1}{2r} \int_{-r}^{r} d\alpha \cos^4(\alpha) - \left(\frac{1}{2r} \int_{-r}^{r} d\alpha \cos^2(\alpha) \right)^2 \geq \frac{4 r^4}{45} - \frac{16r^6}{315}$ by expanding it in the series and keeping the terms which result in the lower bound. This completes the proof.

\end{proof}

\subsubsection{Convexity guarantee: Proof of Theorem~\ref{thm:ITE-convex}}
We devote this subsection to prove Theorem~\ref{thm:ITE-convex} which shows the convexity of the loss landscape for imaginary time evolution.

\begin{theorem}[Approximate convexity of the landscape for imaginary time evolution, Formal]\label{thm:ITE-convex}
Given that the dynamic imaginary time follows
\begin{equation}
    \delta\tau\leq \frac{\mu_{\min}+2|\epsilon|}{48 M \lambda_{\max}} \;,
\end{equation}
the loss landscape is $\epsilon$-convex in a hypercube of width $2r_c$ around a previous optimum $\vec{\theta^*}$ i.e., $\vol(\thv^*, r_c)$ such that
\begin{equation}
    r_c\geq\frac{1}{M}\left(\frac{\mu_{\min}+2|\epsilon|}{16 M}-3\lambda_{\max} \delta\tau \right)
\end{equation}
where $\mu_{\min}$ is the minimal eigenvalue of the Fisher information matrix associated with the loss. 
\end{theorem}

\begin{proof}
    We first recall that the region of the loss function is $\epsilon$-convex (i.e., Definition~\ref{def:epsilon-convex}) if all eigenvalues of the Hessian matrix within the region are larger than $-|\epsilon|$, which can be re-expressed in terms of the fidelity as
\begin{align}\label{eq:proof-def-epsilon-convex-im}
   \left[ \nabla^2_{\thv} F\left( U(\thv) \rho_0 U^\dagger(\thv), \rho_{(\vtheta^*, \delta \tau)} \right)\right]_{\max} \leq |\epsilon| \; ,
\end{align}
for all $\vtheta \in \vol(\vtheta^*, r)$.

By using Taylor's expansion around $\thv^*$ and Taylor reminder theorem (explained in Appendix \ref{app:taylor}), we can write the fidelity like
\begin{align}\label{eq:fidel-taylor-im}
    F(\vec{x}) = 1 - \sum_{i,j} \frac{x_i x_j}{4} \mathcal{F}_{ij}(\vec{0}) + \sum_{i,j,k} \frac{x_i x_j x_k}{6} \left(\frac{\partial^3 F(\vec{x})}{\partial x_i \partial x_j \partial x_k}\right) \bigg|_{\vec{x} = \vec{\nu}} \;,
\end{align}
where we introduce the shorthand notation of the fidelity around this region as $F(\vec{x})$ with $\vec{x} = (\thv - \thv^*, \delta \tau)$, $\FC_{ij} (\vec{0})$ are elements of the Quantum Fisher Information matrix at $\vec{x} = \vec{0}$, and the last term is the result of the Taylor's remainder theorem with $\vec{\nu} = c \vec{x}$ such that $ c \in \left[0, 1  \right]$. We remark that by this notation of $\vec{x}$ the imaginary time component is the last component of $\vec{x}$ i.e., $x_{M+1} = \delta \tau$.

A third order derivative in the last term when taken only with respect to the parameters $\vtheta$ (i.e., no $\delta \tau$ component) can be expressed as (which is the same fashion as in \eqq{im_time_ev} for the real time dynamic case)
\begin{align}
    \frac{\partial^3 F(\vec{x})}{\partial \theta_i \partial \theta_j \partial \theta_k} & =\frac{1}{Z} \Tr\left\{ e^{-\delta\tau H}U^{(M,k)} i\left[ U^{(k,j)} i \left[ U^{(j,i)} i \left[ U^{(i,0)} \rho_0 {U^{(i,0)}}^{\dagger}, \vsigma_i\right]{U^{(j,i)}}^\dagger, \vsigma_j  \right]{U^{(k,j)}}^\dagger, \vsigma_k \right] {U^{(M,k)}}^\dagger e^{-\delta\tau H} \rho_{(\thv^*, 0)}\right\} \;,
\end{align} 
where $U(\thv) = U^{(M,k)}U^{(k,j)}U^{(j,i)}U^{(i,0)}$ with $U^{(a,b)} = \prod_{l = a+1}^b e^{-i x_l \vsigma_l} \widetilde{V}_{l}$ (for $b<M+1$).

When the third derivative is taken with respect to the imaginary time in of the components, by direct calculation we have  
\begin{align}
     \frac{\partial^3 F(\vec{x})}{\partial \theta_i \partial \theta_j \partial \tau} &=\Tr\left[ -\rho_{(\thv^*, 0)}\left\{ \mathcal{B}_{i,j}, H \right\} + 2\Tr[H\mathcal{B}_{i,j}]\rho_{(\vtheta^*, \delta \tau)}\rho_{(\thv^*, 0)} + 2\Tr[\rho_{(\vtheta^*, \delta \tau)} H]\mathcal{B}_{i,j}\rho_{(\thv^*, 0)}\right]
\end{align}
where we defined $\mathcal{B}_{i,j} = \frac{\partial^2 \rho_{(\vtheta^*, \delta \tau)}}{\partial \theta_i \partial \theta_j }$ which can be written as a nested commutator as
\begin{equation}
    \mathcal{B}_{i,j} = \frac{1}{Z} \Tr\left\{ e^{-\tau H}U^{(M,k)} i \left[ U^{(k,i)} i \left[ U^{(i,0)} \rho_0 {U^{(i,0)}}^{\dagger}, \vsigma_i\right]{U^{(j,i)}}^\dagger, \vsigma_j  \right] {U^{(M,j)}}^\dagger e^{-\tau H} \rho_{(\thv^*, 0)}\right\} \;.
\end{equation}
We note that the other third derivative terms with respect to the imaginary time in more than one components can be also be expressed in a similar way. However, since they are not important in our analysis, we do not write them explicitly. Indeed, if we now compute the Hessian matrix of $F(\vec{x})$ with respect to the variational parameters $\vtheta$, all the terms with higher derivatives in time will be $0$. 

We now focus on one element of this hessian matrix $\nabla^2_{\vtheta} F(\vec{x})$ (recall that $\vec{x} = (\vec{\theta}-\vec{\theta}^*, \tau)$). For convenience, we denote $\AC_{i,j,k}(\vec{x}) = \frac{\partial^3 F(\vec{x})}{\partial x_i \partial x_j \partial x_k}$. By direct compuation, we see that
\begin{equation}
    \frac{\partial^2 F(x)}{\partial \theta_i \partial \theta_j} = -\frac{1}{2}\f_{j,k}(\vec{0}) + \frac{1}{6}\widetilde{\AC}_{j,k}(\nu)
\end{equation}
with
\begin{align}
     \widetilde{\AC}_{j,k}(\vec{\nu}) = \sum_{i=1}^{M+1} x_i \left( \AC_{j,k,i}(\vec{\nu}) + \AC_{j,i,k}(\vec{\nu}) + \AC_{i,j,k}(\vec{\nu}) + \AC_{k,j,i}(\vec{\nu}) + \AC_{k,i,j}(\vec{\nu}) +  \AC_{i,k,j}(\vec{\nu})\right) \;,
\end{align}
where we remark that the sum up to $M+1$ is because $\delta\tau$ is included in this sum.

Thus, the largest eigenvalue of $\nabla^2_{\vtheta}F(\vec{x})$ can be bounded as follows
\begin{align}
    \left[ \nabla^2_{\thv} F\left( \vec{x}\right)\right]_{\rm max} \leq -\frac{1}{2} \left[ \FC(\vec{0}) \right]_{\rm min} + \frac{1}{6} [\widetilde{\AC}(\vec{\nu}) ]_{\rm max} \;,
\end{align}
where we define $[A]_{\rm max}$ as the largest eigenvalue of the matrix $A$ and similarly $[A]_{\rm min}$ is used for the smallest eigenvalue. 

Our strategy is to bound $[\widetilde{\AC}(\vec{\nu}) ]_{\rm max}$ with Proposition \ref{prop:upper-eigen}. To do this, we first consider a bound on $\AC_{i,j,k}(\vec{x})$. We consider two cases when (i) the index $k$ represents the parameter component or (ii) the index $k$ represents the time component. For the first case, we can do the same steps as in the real time dynamics presented in Eq.~\eqref{eq:bound} to Eq.~\eqref{eq:bound2}, which is repeated here for completeness.
\begin{align}
    \AC_{i,j,k}(\vec{\nu}) & \leq \left| \AC_{i,j,k}(\vec{\nu})\right| \\
    & \leq \left\| e^{-\delta \tau H} U^{(M,k)} i\left[ U^{(k,j)} i \left[ U^{(j,i)} i \left[ U^{(i,0)} \rho_0 {U^{(i,0)}}^{\dagger}, \vsigma_i\right]{U^{(j,i)}}^\dagger, \vsigma_j  \right]{U^{(k,l)}}^\dagger, \vsigma_k \right] {U^{(M,k)}}^\dagger  e^{-\delta \tau H}\right\|_\infty \left\| \rho_{(\thv^*, 0)} \right\|_1 \\
    & \leq 2^3 \| \vsigma_i \|_{\infty} \| \vsigma_j \|_{\infty} \| \vsigma_k \|_{\infty} \\
    & = 8  
\end{align}
Here the second inequality is due to H\"{o}lder's inequality. In the third inequality we use a few identities including (i) the one-norm of a pure state is $1$, (ii) $\| U A\|_p = \| A\|_p$ for any unitary $U$, (iii) $\| i [A,B] \|_p = 2 \|A\|_p \|B\|_p$, (iv) $\|AB\|_p \leq \|A\|_p \|B\|_p$ and lastly (v) $\| e^{-\delta \tau H}\|_{\infty} < 1$. To reach the final equality, we recall that since $x_l$ and $x_m$ cannot be a time component $\delta t$, $\vsigma_l$ and $\vsigma_m$ are generators of the circuit which have $ \| \vsigma_i \|_{\infty} = \| \vsigma_j \|_{\infty}= \| \vsigma_k \|_{\infty} =1$.

For the second case with the index $k$ representing th time component, we have the following
\begin{align}
   \AC_{i,j,M+1}(\vec{x}) =&\Tr\left[ -\rho_{(\thv^*, 0)}\left\{ \mathcal{B}_{i,j}, H \right\} + 2\Tr[H\mathcal{B}_{i,j}]\rho_{(\vtheta^*, \delta \tau)}\rho_{(\thv^*, 0)} + 2\Tr[\rho_{(\vtheta^*, \delta \tau)} H]\mathcal{B}_{i,j}\rho_{(\thv^*, 0)}\right]\\
    \leq& \left|\Tr\left[ -\rho_{(\thv^*, 0)}\left\{ \mathcal{B}_{i,j}, H \right\} + 2\Tr[H\mathcal{B}_{i,j}]\rho_{(\vtheta^*, \delta \tau)}\rho_{(\thv^*, 0)} + 2\Tr[\rho_{(\vtheta^*, \delta \tau)} H]\mathcal{B}_{i,j}\rho_{(\thv^*, 0)}\right]\right|\\
     \leq & \left|\Tr\left[ -\rho_{(\thv^*, 0)}\left\{ \mathcal{B}_{i,j}, H \right\}\right]\right|+\left|\Tr\left[ 2\Tr[H\mathcal{B}_{i,j}]\rho_{(\vtheta^*, \delta \tau)}\rho_{(\thv^*, 0)} \right]\right|+ \left|\Tr\left[ 2\Tr[\rho_{(\vtheta^*, \delta \tau)} H]\mathcal{B}_{i,j}\rho_{(\thv^*, 0)} \right]\right| \label{eq:ite-con-con}
\end{align}

Now we can bound each individual term in Eq.~\eqref{eq:ite-con-con} with
\begin{align}
    &\left|\Tr\left[ -\rho_{(\thv^*, 0)}\left\{ \mathcal{B}_{i,j}, H \right\}\right]\right| \leq 2 \left|\Tr\left[\rho_{(\thv^*, 0)} \mathcal{B}_{i,j}H\right]\right|\\
    &\left|\Tr\left[ 2\Tr[H\mathcal{B}_{i,j}]\rho_{(\vtheta^*, \delta \tau)}\rho_{(\thv^*, 0)} \right]\right|\leq 2\left| \Tr[H\mathcal{B}_{i,j}]\right| \label{eq:trace_states}\\
    &\left|\Tr\left[ 2\Tr[\rho_{(\vtheta^*, \delta \tau)} H]\mathcal{B}_{i,j}\rho_{(\thv^*, 0)} \right]\right| = 2||H||_{\infty}\left|\Tr\left[ \mathcal{B}_{i,j}\rho_{(\thv^*, 0)} \right]\right|\label{eq:trace_H}
\end{align}
where we are using $|\Tr[\rho_{(\vtheta^*, \delta \tau)}\rho_{(\thv^*, 0)}]|\leq 1$ in \eqq{trace_states} and  H\"{o}lder's inequality in \eqq{trace_H}. Now, all of the remaining terms can be bounded using \eqq{bound}. Thus we obtain
\begin{equation}
    \AC_{i,j,M+1}\leq 24 ||H||_{\infty} \ .
\end{equation}

Now, we can bound the sum of the absolute of elements in a row of $\widetilde{\AC}(\vec{\nu})$ as
\begin{align}\label{eq:bound_tilde-im}
    \sum_{j=1}^M \left| \widetilde{\AC}_{i j}(\vec{\nu}) \right| 
    & \leq \sum_{j=1}^M \sum_{k = 1}^{M+1} |x_k| \left( |\AC_{i,j,k}(\vec{\nu})| + |\AC_{i,k,j}(\vec{\nu})| + |\AC_{k,i,j}(\vec{\nu})| + |\AC_{j,i,k}(\vec{\nu})| + |\AC_{j,k,i}(\vec{\nu})| +  |\AC_{k,j,i}(\vec{\nu})| \right)\\
    & \leq 48 M (3\lambda_{\max}\delta\tau + Mr) \;. 
\end{align}

Finally we invoke Proposition~\ref{prop:upper-eigen}. Thus, the largest eigenvalue of the matrix $\widetilde{\AC}$, $[\widetilde{\AC}_{j,k}]_{\max}$ can be bounded by 
\begin{equation}
    [\widetilde{\AC}_{j,k}]_{\max}\leq 48 M (3\lambda_{\max}\delta\tau + Mr_c) \ . 
\end{equation}

With this result we can guarantee the legion of $\epsilon$-convexity (i.e. \eqq{proof-def-epsilon-convex-im}) by enforcing the following condition
\begin{align}
    -\frac{1}{2} \left[ \FC(\vec{0}) \right]_{\rm min} + 8 M  \left( 3\lambda_{\rm max} \delta \tau + M r \right)   \leq |\epsilon| \;.
\end{align}
Upon rearranging the terms, we find
\begin{equation}\label{eq:proof-ricard-hello-im}
    r\leq\frac{1}{M}\left( \frac{\mu_{\min}+2|\epsilon|}{16 M}-3\lambda_{\max}\delta\tau \right) \ . 
\end{equation}
Indeed, this implies that \textit{any} hypercube $\vol(\thv^*, r)$ such that $r$ satisfies Eq.~\eqref{eq:proof-ricard-hello-im}  is guaranteed to be approximately convex. Hence, we know that the total $\epsilon$-convex region has to be at least of size $\frac{1}{M}\left(\frac{\mu_{\rm min}+2|\epsilon|}{16 M} - 3\lambda_{\rm max} \delta \tau\right)$. More explicitly, by denoting $r_c$ to be the length of the total $\epsilon$-approximate convex region $\vol(\thv^*, r_c)$, we have
\begin{equation}
    r_c \geq \frac{1}{M}\left(\frac{\mu_{\rm min}+2|\epsilon|}{16 M} - \lambda_{\rm max} \delta \tau\right) \;.
\end{equation}
Finally, we note that the bound is only informative if the time-step respects
\begin{align}
    \delta \tau \leq \frac{\mu_{\rm min} + 2 |\epsilon| }{48 M \lambda_{\rm max}} \;.
\end{align}
This completes the proof of the theorem.

\end{proof}

\subsubsection{Adiabatic moving minima: Proof of Theorem~\ref{thm:ITE-adiabaticminimum}}
In this subsection, we analytically prove Theorem~\ref{thm:ITE-adiabaticminimum}. We first show an equivalent result to Proposition \ref{prop:adiabatic-min} for imaginary time evolution. We refer the readers to Appendix~\ref{app:moving-min} for definitions of adiabatic minimum (Definition \ref{def:adiabatic-minimum}) and adiabatic shift (Definition \ref{def:adaibatic-shift}). 

\begin{proposition}\label{prop:adiabatic-min-im}
Given a time-step of the current iteration $\delta \tau$ and assuming that the adiabatic minimum exists within this time frame, the shift of the adiabatic minimum $\valpha_{A}(\delta \tau)$ as defined in Definition~\ref{def:adaibatic-shift} can be bounded as 
\begin{align}
    \| \valpha_A(\delta \tau) \|_2 \leq \frac{4\sqrt{M} \lambda_{\rm max} \delta \tau}{\beta_A} \;,
\end{align}
where $M$ is the number of parameters, $\lambda_{\rm max}$ is the largest eigenvalue of the dynamic Hamiltonian $H$ and $\beta_A = \frac{\Dot{\valpha}_A^T(\delta \tau) \left( \nabla^2_{\valpha} \LC(\valpha,\delta \tau) \big|_{\valpha = \valpha_A(\delta \tau)}\right) \Dot{\valpha}_A(\delta \tau) }{\| \Dot{\valpha}_{A}(\delta \tau) \|_2^2} $ 
\end{proposition}

\begin{proof}
    The proof here is very similar to the equivalent version for Real Time Evolution. Through out this proof, it is more convenient to use $\tau$ as an imaginary time-step (instead of $\delta \tau$). We start by recalling that by Definition \ref{def:adaibatic-shift}, the adiabatic shift follows
\begin{align}
    \nabla_{\valpha} \LC(\valpha,\tau) \big|_{\valpha = \valpha_A(\tau)} = \vec{0} \;,
\end{align}
which holds for any $\tau$. Similarly to the previous case, we use this notation $\nabla_{\valpha_A} \LC := \nabla_{\valpha} \LC(\valpha,\tau) \big|_{\valpha = \valpha_A(\tau)}$. 

We can derivative with respect to $\tau$ to find
\begin{align}\label{eq:moving-min-2-proof1-im}
    \frac{d}{d\tau}\left(\nabla_{\valpha_A} \LC \right) = \partial_\tau \nabla_{\valpha_A} \LC + \left(\nabla_{\valpha_A}^2 \LC  \right)\Dot{\valpha}_A = \vec{0} \;,
\end{align}
where we denote $\partial_\tau = \partial/\partial\tau$ and $\Dot{\valpha}_A(\tau) = d {\valpha}_A(\tau) / d\tau$. 

Recall that $\nabla_{\valpha_A}^2 \LC$ is a matrix vector multiplication with the $\nabla_{\valpha_A}^2 \mf$ is the Hessian Matrix evaluated at the adiabatic minima $\valpha_A(\tau)$. We multiply from the left with $\Dot{\valpha}^T_A(\tau)/\|\Dot{\valpha}_A(\tau) \|_2$ to find
\begin{align}
    \frac{\Dot{\valpha}_A^T}{\|\Dot{\valpha}_A \|_2} \partial_\tau \nabla_{\valpha_A} \LC  + \frac{\Dot{\valpha}_A^T \left(\nabla_{\valpha_A}^2 \LC  \right)\Dot{\valpha}_A}{\|\Dot{\valpha}_A \|_2^2} \|\Dot{\valpha}_A \|_2 = 0 \;.
\end{align}

For convenience, we denote $\beta_A := \Dot{\valpha}_A^T \left(\nabla_{\valpha_A}^2 \LC  \right)\Dot{\valpha}_A / \| \Dot{\valpha}_A \|_2^2 $. We then rearrange the terms to find a bound on the norm of $\Dot{\valpha}_A(\tau)$:
\begin{align}
    \| \Dot{\valpha}_A(\tau) \|_2 & =  \left\| -  \frac{\Dot{\valpha}_A^T \partial_\tau \nabla_{\valpha_A} \LC}{\|\Dot{\valpha}_A \|_2 \beta_A} \right\|_2 \\
    & \leq \frac{\| \partial_\tau \nabla_{\valpha_A} \LC\|_2}{\beta_A} \\
    & \leq \frac{\sqrt{M} \left|\partial_\tau \partial_{\alpha^{(i)}_{A}} \LC \right|_{\rm max}}{\beta_{A}} \\
    & \leq \frac{4\sqrt{M} \lambda_{\rm max}}{\beta_A} \;,
\end{align}
where in the first inequality we use Cauchy-Schwartz. In the second we expand the 2-norm explicitly and take the largest value of the sum (i.e. $ \| \vec{a} \|_2 = \sqrt{\sum_{i=1}^M a_i^2} \leq \sqrt{M} |a_i|_{\rm max}$ with $\alpha_A^{(i)}$ being the $i^{\rm th}$ component of $\valpha_{A}$).

To reach the last inequality we use the explicit form of the loss function $\LC(\valpha,\tau) = 1 - \Tr\left( e^{-H\tau} \rho_{\thv^*} e^{H\tau} \rho_{\vtheta}\right)$ where $H$ is the dynamical Hamiltonian, $\rho_{\thv^*}$ is the state corresponding to the solution of the previous iteration and $\rho_{\vtheta}$ is the parameterised state that depends on $\valpha$ and respect a parameter shift rule. With this, we can bound the term $\left|\partial_\tau \partial_{\alpha^{(i)}_{A}} \LC \right|_{\rm max}$ as follows
\begin{align}
    \left|\partial_\tau \partial_{\alpha^{(i)}_{A}} \LC \right|_{\rm max}  & =\left| \frac{\partial}{\partial \tau}\left( \frac{\partial}{\partial \alpha^{(i)}}  \LC(\valpha, \tau)\right)\bigg|_{\valpha = \valpha_A(\tau)}\right|_{\rm max}  \\
    & = \left|\frac{1}{2} \left(\frac{\partial}{\partial \tau} \LC \left(\valpha_A + \frac{\pi}{2} \hat{\alpha}_i\right) - \frac{\partial}{\partial \tau} \LC \left(\valpha_A - \frac{\pi}{2} \hat{\alpha}_i\right)\right)\right|_{\rm max} \\
    & \leq \max_{i}  \left|\frac{\partial}{\partial \tau} \LC \left(\valpha_A \pm \frac{\pi}{2} \hat{\alpha}_i\right)\right|\\
    &=\left\| e^{-H\tau}\left(-\{\rho_{\thv^* },H\} +2\Tr[\rho_{\thv^* }H]\rho_{\thv^* } \right) e^{-\tau H} \right\|_{\infty} \\
    &\leq 4 \lambda_{\rm max} \;,
\end{align}
where we use the parameter shift rule in the second equality. In the first inequality we we maximise on all the possible terms of this parameter shift rule, and in the third equality we apply the derivative of imaginary time shown in \ref{eq:im_time_ev}, apply H\"{o}lder's inequality and use that $\|\rho\|_1 = 1$ for any pure quantum state $\rho$. In the last inequality we simply use the triangle inequality to bound the $\|A + B\|_p\leq \|A\|_p + \|B\|_p$. 

Lastly, we can bound shift in the adiabatic minima as follows
\begin{align}
    \left\| \valpha_A(\tau) \right\|_2 & = \left\| \int_{0}^\tau \Dot{\valpha}_A(\tau') d\tau'\right\|_2 \\
    & \leq \int_{0}^\tau \| \Dot{\valpha}_A(\tau') \|_2 d\tau' \\
    & \leq \frac{4\sqrt{M} \lambda_{\rm max} \tau}{\beta_A} \;,
\end{align}
which completes the proof.

\end{proof}

\medskip

We are now ready to prove Theorem~\ref{thm:ITE-adiabaticminimum} which is detailed in the following.

\begin{theorem}[Adiabatic minimum is within provably `nice' region for imaginary time evolution, Formal]
If the imaginary time-step $\delta \tau$ is chosen such that
\begin{align}
    \delta \tau \leq\frac{\eta_0 \beta_A}{4 M \lambda_{\rm max}} \;,
\end{align}
for some small constant $\eth_0$. then the adiabatic minimum $\thv_A(\delta \tau)$ is guaranteed to be within the non-vanishing gradient region (as per Theorem~\ref{thm:ITE-variance-lower-bound}), and additionally, if $\delta \tau$ is chosen such that
\begin{equation}
\delta\tau \leq \frac{ \beta_A(\mu_{\min} + 2|\epsilon|)}{64M^{5/2}\lambda_{\max}\left(1+\frac{3\beta_A}{4M^{3/2}}\right)} \;,
\end{equation}
then the adiabatic minimum $\thv_A(\delta \tau)$ is guaranteed to be within the $|\epsilon|$-convex region (as per Theorem~\ref{thm:ITE-convex}) where 
\begin{equation}
    \beta_A := \frac{\Dot{\vtheta}_A^T(\delta  \tau) \left( \nabla^2_{\vtheta} \LC_{\rm ITE}(\vtheta) \big|_{\vtheta = \vtheta_A(\delta  \tau)}\right) \Dot{\vtheta}_A(\delta \tau) }{\| \Dot{\thv}_{A}(\delta \tau) \|_2^2} 
\end{equation}
corresponds to the second derivative of the loss in the direction in which the adiabatic minimum moves.
\end{theorem}

\begin{proof}
    From Proposition~\ref{prop:adiabatic-min-im} and the norm inequality, the adiabatic minimum follows
\begin{align}
    \| \valpha_A(\delta \tau) \|_\infty \leq \| \valpha_A(\delta \tau) \|_2 \leq \frac{4\sqrt{M} \lambda_{\rm max} \delta \tau}{\beta_A} \;.
\end{align}

Now we want to incorporate the conditions of the region that we are interested. Indeed, by tuning $\delta\tau$, we want a guarantee that the adiabatic minimum is within (i) the region with substantial gradients and/or (ii) the convex region. 

\medskip

\underline{For (i) the region with substantial gradients,} we recall from Theorem~\ref{thm:ITE-variance-lower-bound} that for the imaginary time scaling as $\delta \tau \leq 1/12 \lambda_{\rm max}$, the hypercube of width $2r$ has polynomial large variance within the region where $r$ scales as
\begin{align}
    r = \frac{\eta_0}{\sqrt{M}} \;,
\end{align}
for some constant $\eta_0$. Hence, the sufficient condition to have the adiabatic minimum to be within this substantial gradient region is that
\begin{align}
     \| \valpha_A(\delta \tau) \|_\infty \leq \frac{4 \sqrt{M} \lambda_{\rm max} \delta \tau}{\beta_A} \leq \frac{\eta_0}{\sqrt{M}} \;,
\end{align}
which, upon rearranging leads to
\begin{align}
    \delta \tau \leq \frac{\eta_0 \beta_A}{4 M \lambda_{\rm max}} \;.
\end{align}

\medskip

\underline{For (ii) the convex region,} from Theorem~\ref{thm:ITE-convex}, if we have the dynamic time bounded as $\delta\tau\leq \frac{\mu_{\min}+2|\epsilon|}{48 M \lambda_{\max}}$, we have $\epsilon$-convexity in the hypercube of with $2r_c$ for 
\begin{equation}
    r_c\geq\frac{1}{M}\left(\frac{\mu_{\min}+2|\epsilon|}{16 M}-3\lambda_{\max} \delta\tau \right)
\end{equation}

Therefore it is sufficient to have the guarantee that the adiabatic minima is inside this convex region by imposing
\begin{equation}
     \| \valpha_A(t) \|_\infty \leq \frac{4\sqrt{M} \lambda_{\rm max} \delta\tau}{\beta_A}\leq\frac{1}{M}\left(\frac{\mu_{\min}+2|\epsilon|}{16 M}-3\lambda_{\max} \delta\tau \right)\leq r_c
\end{equation}
which after rearranging terms the imaginary time-step is bounded by
\begin{equation}\label{eq:dt-convex-min-im}
\delta\tau\leq \frac{ \beta_A(\mu_{\min} + 2|\epsilon|)}{64M^{5/2}\lambda_{\max}\left(1+\frac{3\beta_A}{4M^{3/2}}\right)} \;.
\end{equation}
This completes the proof. 
\end{proof}

\section{Outlook on other iterative approaches}\label{app:extension}

So far, we have presented our results in the context of iteratively learning a Hamiltonian dynamics from a fixed initial state. However, we can reinterpret the result for the substantial gradient region in an abstractified form to apply them to other iterative methods. 
In Section~\ref{sec:extension-to-other-iterative} of the main text, we have discussed about physical intuition together with an informal version of the theoretical result.  
In this section, we provide further information with the formal version of the theorem and the associated proof in the general fidelity-type loss setting, as well as some illustrative tasks of learning an unknown unitary via the variational and quantum machine learning approaches. Importantly, we remark that the imaginary time evolution result discussed in Appendix~\ref{app:imaginary} constitutes another prime example.

\medskip

The section is structured as follows:
\begin{itemize}
    \item In Appendix~\ref{app:other-iterative-summary}, 
    we provide the formal version of the theorem which gives a theoretical guarantee of the substantial gradient region in the general setting of the iterative approach with the fidelity-type loss.
    \item In Appendix~\ref{app:extension-variational-u}, we discuss a specific task of learning an unknown unitary with the variational approach and show how the extension explicitly works in this scenario.
    \item In Appendix~\ref{app:extension-qml-u}, we consider an alternative approach of iteratively learning the unknown unitary via quantum machine learning and provide discussion on how to extend our results.
    \item In Appendix~\ref{app:proof-extension}, we provide proofs of the theoretical results.
\end{itemize}

\subsection{Summary of key analytical results}
\subsubsection{General iterative methods with the fidelity-type loss}\label{app:other-iterative-summary}

In what follows, we state the formal version of Theorem \ref{th:substantial_gradient} in the main text. For any fidelity-type loss iterative methods, the theorem guarantees the existence of the substantial gradient region around the intialization of the solution from the previous iteration.

\setcounter{theorem}{3}

\begin{theorem}[A substantial gradient region, Formal]\label{th:substantial_gradient_app}
Consider any iterative method where the loss function of any iteration can be expressed in the fidelity form as
\begin{align}
    \LC(\thv) = 1 - |\langle \psi_0 |U^\dagger(\thv) |\psi_{\rm target} \rangle|^2 \;,
\end{align}
with $|\psi_{\rm target}\rangle$ being a target state for the current iteration and $\thv^*$ being a set of parameters obtained from the previous iteration. Further consider the general ansatz in Eq.~\eqref{eq:circuit} and assume that in the first iteration the system is prepared in a product initial state $\rho_0$ and let us choose $\vsigma_1$ such that $\Tr[\rho_0 \vsigma_1 \rho_0 \vsigma_1] = 0$. 

Given that the overlap between the target and the state around intialization with $\thv^*$ follows
\begin{align}\label{eq:min_overlap}
     F_{\rm target}(\thv^*) := \left| \langle \psi(\thv^*) | \psi_{\rm target}\rangle \right|^2 \geq \frac{1}{2} \;,
\end{align}
and we consider uniformly sampling parameters in a hypercube of width $2r$ around the solution from the previous iteration $\thv^*$, i.e. $\vol(\thv^*,r)$, such that 
\begin{equation}\label{eq:region-for-th4}\frac{1+2F_{\rm target}(\thv^*)}{2F_{\rm target}(\thv^*)}\frac{3r_0^2}{M-1} \geq r^2 \, .
\end{equation}
Then the variance at any iteration of the algorithm is lower bounded as 
\begin{align} \label{eq:prop1-exact-var-a}
		\Var_{\thv\sim\uni(\thv^*, r) }[\mf(\thv)]\geq  \frac{4 r^4}{45} \left(1 - \frac{4r^2}{7} \right)\left[\left( 1 - r_0^2\right) \left(2F_{\rm target}(\thv^*)-1\right)\right]^2 \, .
\end{align}

In addition, by choosing $r$ such that $r\in \Theta\left( \frac{1}{\sqrt{M}}\right)$, we have
\begin{align}
     \Var_{\vtheta \sim\uni(\vtheta^*, r)} \left[ \LC (\vtheta)\right] \in \Omega\left( \frac{1}{M^2}\right) \;.
\end{align}
\end{theorem}
We note that the proof of Theorem~\ref{th:substantial_gradient_app} largely follow the proof steps of Theorem~\ref{thm:variance-lower-bound}. The key difference is here we explicitly assume the large fidelity condition rather than enforcing it with the small time-step condition.
Lastly, we note that while the convexity and adiabatic minimum results are also expected to carry over to other iterative methods, proving this rigorously may depend on the specific details of an iterative method we are interested in. Nevertheless, we argue below that there exist iterative tasks namely unitary learning that all of our three theoretical results can be rigorously shown to apply at the same time. 

\subsubsection{Learning an unknown target unitary via a variational approach}\label{app:extension-variational-u} 
To illustrate how the extension may work specifically in a different context, we consider a task of learning the unknown target unitary $e^{- i H t}$ with some parametrized circuit $U(\thv)$~\cite{cirstoiu2020variational, mizuta2022local,khatri2019quantum}. Here one learns the unitary itself rather than its effect on a fixed initial state. Similar to what have been done so far, we could aim to achieve $U(\thv^*_{\rm opt}) \approx e^{- i H t}$ with some optimal parameters $\thv^*_{\rm opt}$ by breaking the target unitary into smaller time steps $e^{- i H t} = \prod_{k=1}^N e^{- i H \delta t}$ and gradually learning it in an iterative manner. 

Now, we show that our results obtained for learning the dynamics with some fixed input are also directly applied to this iterative approach. To see this, consider one form of the loss function for a given iteration
\begin{equation}
    \LC_{\rm uni}(\thv) = 1-\frac{1}{4^n}\left|\Tr[U^\dagger(\thv) U_{\delta t}(\thv^*)]\right|^2  \;,
\end{equation}
where $U_{\delta t}(\thv^*) = e^{-iH\delta t} U(\thv^*)$ with $\thv^*$ being the optimal parameters from the previous iteration. 

Crucially, as shown in Appendix \ref{app:learning_unitary}, this loss function can be re-written in terms of the loss function in Eq.~\eqref{eq:loss} on the $2n$-qubit composite system. That is, we have
\begin{align}
     \LC_{\rm uni}(\thv) = 1-|\bra{\Phi_+} \widetilde{U}(\thv) e^{ i H\otimes \1 \delta t}  \widetilde{U}(\thv^*)  \ket{\Phi_+}|^2
\end{align}
where $ \widetilde{U}(\thv) = U(\thv) \otimes \1$, and $ \ket{\Phi_+} = \bigotimes_{j=1}^n |\phi_+\rangle_{j}$ is an entangled state on the composite system with $|\phi_+\rangle_j = \frac{1}{\sqrt{2}}(|00\rangle_j + |11\rangle_j)$ being the Bell state on the $j^{\rm th}$ qubits of two subsystems. Hence, by choosing $U(\thv)$ as in Eq.~\eqref{eq:circuit} with $\sigma_1$ such that $\Tr[|\Phi_+\rangle \langle \Phi_+ | \vsigma_1 |\Phi_+ \rangle\langle \Phi_+| \vsigma_1] = 0$ (e.g., a Pauli-X operator on the first qubit $\sigma_1 = X_1$), Theorem~\ref{thm:variance-lower-bound} applies, ensuring the region with substantial gradients around the initialization. Crucially, note that since here our approach is to map the loss into exactly the same fidelity form studied in this work in Eq~\ref{eq:loss} (rather than abstractifying the proof steps), the other theoretical results namely guarantees on convexity and adiabatic minimum also follow in a similar way.

\subsubsection{Learning an unknown target unitary via quantum machine learning}\label{app:extension-qml-u}

Alternative to the approach described in the previous sub-section, learning the unknown unitary can be done through the quantum machine learning (QML). In this setting we do not have a direct access to the unknown unitary, but rather are given a training dataset of input and output states (after the dynamics). 
We then aim to learn the unknown unitary $e^{-iHt}$ with these training states. 
Particularly, in an iterative version, we assume for each iteration a $N_s$-sized training dataset of the form 
\begin{align}
    \SC = \{ U(\thv^*)|\psi_j\rangle, e^{-iH\delta t}U(\thv^*) |\psi_j\rangle \}_{j=1}^{N_s} \;,
\end{align}
where $\{ U(\thv^*)|\psi_j\rangle \}_{j=1}^{N_s}$ is a set of input states for that iteration and $\{ e^{-iH\delta t}U(\thv^*)|\psi_j\rangle \}_{j=1}^{N_s}$ is a set of corresponding output states for that iteration. In addition, we consider that each $|\psi \rangle$ is a product of random single-qubit stabilizer states i.e., $|\psi_j \rangle$ independently and uniformly drawn from $\{ |0\rangle, |1\rangle, |+\rangle, |-\rangle, |y+\rangle, |y-\rangle \}^{\otimes n}$. By choosing $\{ |\psi_j \rangle\}_{j=1}^{N_s}$ this way, it enables the out-of-distribution generalization result~\cite{caro2022outofdistribution}.

The loss function for a given iteration can be expressed as
\begin{align}\label{eq:loss_ml}
    \LC_{\rm QML} (\thv) &= \frac{1}{4N_s} \sum_{j=1}^{N_s} \| U(\thv)|\psi_j\rangle\langle \psi_j| U^\dagger(\thv)\nonumber \\
    & \,\,\,\;\;\;\;\,\,\,\;\; - e^{-iH\delta t} U(\thv^* )|\psi_j\rangle\langle \psi_j| U^\dagger(\thv^*) e^{iH\delta t}\|_1^2  \\
   & = 1 - \frac{1}{N_s} \sum_{j=1}^{N_s} |\langle \psi_j| U^\dagger(\thv) e^{-iH\delta t} U(\thv^*) |\psi_j \rangle |^2  \;,
\end{align}
where $\| \cdot \|_1$ is the Schatten one-norm, and we rewrite the loss in terms of the average fidelity to reach the final line. 

We can transform this loss function into the fidelity-type loss in which the previous results can be applied. By denoting,
\begin{equation}
    \rho_0  = \frac{1}{N_s}\sum_{j=1}^{N_s}\ketbra{\psi_j}\otimes\ketbra{\psi_j}
\end{equation}
the loss can be rewritten as
\begin{align}
    \LC_{\rm QML} (\thv) = 1-N_s \Tr[\widetilde{U}(\thv)\rho_0 \widetilde{U}(\thv)^\dagger e^{-i\delta t H\otimes \1 }\widetilde{U}(\thv^*)\rho_0 \widetilde{U}^\dagger(\thv^*)e^{i\delta tH\otimes\1}]
\end{align}
where $\widetilde{U}(\thv) = U(\thv)\otimes\1$, 
which may appear to be in the identical form of the loss in Eq.~\eqref{eq:loss} (up to a factor $N_s$). However, $\rho_0$ is a mixed state and additional technical steps, which are shown in details in Appendix~\ref{app:learning_unitary_qml}, are required to ensure that Theorem~\ref{thm:variance-lower-bound} can be applied. 
In addition, we show that by choosing $\vsigma_1$ to a Pauli operator the condition $\Tr[\rho_0\sigma_1\rho_0\sigma_1] = 0$ of Theorem~\ref{thm:variance-lower-bound} is fulfilled with high probability exponentially close to $1$. This is indeed a consequence of having $\{ |\psi_j\}_{j=1}^{N_s}$ as products of random single-qubit stabilizer states. Together, we have the guarantee of the substantial gradient region via Theorem~\ref{thm:variance-lower-bound} 
in this setting. Other theoretical results are also applicable by following similar arguments.

\subsection{Proofs of analytical results}\label{app:proof-extension}

\subsubsection{Proof of Theorem \ref{th:substantial_gradient_app}: A substantial gradient region}
\begin{proof}
We start this proof by stating the fact that Proposition \ref{prop:variance-lower-bound} can be slightly modified to work in this general setting. 

\begin{propositionb}[Rephrasing of Proposition~\ref{prop:variance-lower-bound}]\label{prop:variance-lower-bound-2}
Consider the loss function $\mf(\thv)$ as defined in Eq.~\eqref{eq:loss} and with an ansatz of the general form defined in Eq.~\eqref{eq:circuit} with $M$ parameters. The variance of $\mf(\thv)$ over the hypercube parameter region $\vol(\thv^*, r)$ around an optimal solution of the previous iteration $\thv^*$ can be bounded as 
	\begin{equation} \label{eq:prop1-exact-var-b}
		\Var_{\thv\sim\uni(\thv^*, r) }[\mf(\thv)]\geq  \ (c_+ - k_+^2) \min_{\Tilde{\xi}\in [-1,1]} \left( k_+^{M-1} \Delta_{\thv^*}  + (1-k_+^{M-1})\Tilde{\xi} \right)^2 ,
	\end{equation}
	where we have
	\begin{align}
		&c_+     :=  \mathbb{E}_{\alpha \sim\uni(0,r)} [ \cos^4{\alpha}]\;, \\           
		&k_+   := \mathbb{E}_{\alpha \sim\uni(0,r)} [ \cos^2{\alpha}] \; ,                                                               \\
		&\Delta_{\thv^*} := \Tr[(\rho_0 -\vsigma_1 \rho_0 \vsigma_1) U^{\dagger}\left(\vec{\theta^*}\right)\rho_{\rm target} U\left(\vec{\theta^*}\right)] \; .\label{eq:var-delta-key-2}
	\end{align}
	Here $\vsigma_1$ is the Pauli string associated with the first gate in the circuit $U(\thv)$ as defined in Eq.~\eqref{eq:circuit}, $\rho_0 = |\psi_0 \rangle\langle\psi_0|$ is an initial state before the time evolution and $\rho_{\rm target}= \ketbra{\psi_{\rm target}}$ with $H$ being the underlying Hamiltonian of the quantum dynamics. 
\end{propositionb}

The rephrasing of the proposition stresses the fact that this lower-bound works for an arbitrary target state and hence we could replace the time-evolved state $\rho_{\thv^*, \delta t}$ with $\rho_{\rm target}$.
This lower-bound in itself, though, is not meaningful, as it does not elucidate any scaling. Therefore, the remaining part of this proof will be devoted to prove the scaling of the bound. To do so, we will follow the same structure as the proof of Theorem \ref{thm:variance-lower-bound}.

The proof of this theorem is completely equivalent to the proof of Theorem \ref{thm:variance-lower-bound} from Eq.~\eqref{eq:proof-new1} up to Eq.~\eqref{eq:basis-phi2} (simply substituting $\rho_{(\thv^*, \delta t)}\to \rho_{\rm target}$) and hence we will not rewrite it here. 
Instead we continue from that point onward. This means that we have to start by showing an equivalent condition to Eq.~\eqref{eq:proof-coro1-2} but for a general target state. We can do this as follows. 

Under the same assumption as in Theorem \ref{thm:variance-lower-bound} that $\Tr[\rho_0\sigma_1\rho_0\sigma_1]=0$, we can see that
\begin{align}\label{eq:proof-coro1-2-th4}
 \Tr\left[\left( \rho_0 - \vsigma_1 \rho_0 \vsigma_1 \right) U^{\dagger}(\vtheta^*)   \rho_{\rm target}  U(\vtheta^*) \ \right] 
    & = 
   F_{\rm target}(\thv^*) - \Tr\left[|\phi_2 \rangle\langle \phi_2| \rho_{\rm target} \right] \\
    &\geq F_{\rm target}(\thv^*) - \sum_{i=2}^{2^n}  \Tr\left[  |\phi_i \rangle \langle \phi_i |\rho_{\rm target}\right] \\
    & = 2F_{\rm target}(\thv^*) - 1\, , 
\end{align}
where we denote $ F_{\rm target}(\thv^*) := \left| \langle \psi(\thv^*) | \psi_{\rm target}\rangle \right|^2$. The first equality is obtained by writing the first term in the fidelity form and writing the second term in $|\phi_2 \rangle\langle \phi_2|$ in Eq.~\eqref{eq:basis-phi2}, in the first inequality we include terms corresponding to other basis (which holds since $\Tr[\rho |\phi_i \rangle\langle \phi_i|] \geq 0$ for any $\rho$ and $|\phi_i \rangle\langle \phi_i|$). The second equality is from the completeness of the basis $\sum_{i=1}^{2^n} |\phi_i \rangle\langle \phi_i | = \1$.

Similarly to Eq.~\eqref{eq:proof-coro1-3}, we can obtain the condition 
\begin{equation}
    \frac{1}{2F_{\rm target}(\thv^*)}\geq \frac{1}{1 + \Delta_{\thv^*}}
\end{equation}
and therefore we can satisfy the equivalent condition of Eq.~\eqref{eq:condition-perturbation2}. Indeed, similarly to the proof of Theorem \ref{thm:variance-lower-bound} if we enforce that $k_+^{M-1}\geq \frac{1}{2F_{\rm target}(\thv^*)}$ then the equivalent condition to the one in Eq.~\eqref{eq:condition-perturbation2} is achieved. It is important to stress that if we have $F_{\rm target}(\thv^*)<1$ this implies that $\frac{1}{2F_{\rm target}(\thv^*)}>1$. Therefore in that case the condition $k_+^{M-1}\geq \frac{1}{2F_{\rm target}(\thv^*)}$ cannot be fulfilled. This is equivalent to the condition in Eq.~\eqref{eq:min_overlap} required in Theorem \ref{th:substantial_gradient_app}.

Now we can simplify the aforementioned condition, by finding a lower-bound in $k_+^{M-1}$ as done in Eq.~\eqref{eq:proof-coro1-perturb-rearrange} (the bound on $k_{+}^{(M-1)}$ is found in Eq.~\eqref{eq:proof-coro1-0}). That is for $r_0\in(0,1)$
\begin{equation}\label{eq:proof-th4-rearange}
    1 - \frac{(M-1)r^2}{3r^2_0} \geq  \frac{1}{2F_{\rm target}(\thv^*)} \Rightarrow k_+^{M-1}\geq \frac{1}{2F_{\rm target}(\thv^*)}
\end{equation}
By rearranging the inequality in Eq.~\eqref{eq:proof-th4-rearange}, we have the perturbation regime of $r$ to be similar to Eq.~\eqref{eq:proof-coro1-perturb} for Theorem \ref{thm:variance-lower-bound}. That is, in this case we find
\begin{equation}\label{eq:proof-coro1-perturb-th4}
    r^2\leq \frac{1+2F_{\rm target}(\thv^*)}{2F_{\rm target}(\thv^*)}\frac{3r_0^2}{M-1}
\end{equation}
as stated in Eq.~\eqref{eq:region-for-th4}.

Finally, as done in the proof of Theorem \ref{thm:variance-lower-bound}, we can just put this results back into the lower-bound obtained in Proposition \ref{prop:variance-lower-bound-2}. Following from Eq.~\eqref{eq:proof-coro1-1} 
\begin{align} \label{eq:prop1-exact-var-c}
		\Var_{\thv\sim\uni(\thv^*, r) }[\mf(\thv)]\geq & (c_+ - k_+^2)\left( k_+^{M-1} (\Delta_{\thv^*} +1)-1\right)^2 \\
        \geq & (c_+ - k_+^2)\left( k_+^{M-1} 2F_{\rm target}(\thv^*)-1\right)^2 \\
        \geq & (c_+ - k_+^2)\left[\left( 1 - \frac{(M-1)r^2}{3}\right) 2F_{\rm target}(\thv^*)-1\right]^2 \\
        \geq & (c_+ - k_+^2)\left[\left( 1 - r_0^2\right) \left[2F_{\rm target}(\thv^*)-1\right]\right]^2 \\
         \geq &  \frac{4 r^4}{45} \left(1 - \frac{4r^2}{7} \right)\left[\left( 1 - r_0^2\right) \left[2F_{\rm target}(\thv^*)-1\right]\right]^2 \, 
\end{align}
where the second inequality is due to Eq.~\eqref{eq:proof-coro1-2-th4}, the third inequality is by bounding $k_+^{M-1}$ with Eq.~\eqref{eq:proof-bound-kplus} and in the next inequality we explicitly use the perturbation regime of $r$ in Eq.~\eqref{eq:proof-coro1-perturb-th4}. To reach the last inequality, we directly bound $c_+ - k_+^2  = \frac{1}{2r} \int_{-r}^{r} d\alpha \cos^4(\alpha) - \left(\frac{1}{2r} \int_{-r}^{r} d\alpha \cos^2(\alpha) \right)^2 \geq \frac{4 r^4}{45} - \frac{16r^6}{315}$ by expanding it in the series and keeping the terms which result in the lower bound.
\end{proof}

\subsubsection{Proof of rewriting the loss in the learning an unknown unitary via the variational approach.}\label{app:learning_unitary}
In this subsection, we analytically show that the loss in the learning an unknown unitary task can be recast into the form such that the theoretical results developed in this work can be applied. Recall that for a given iteration one valid loss function for this task is of the form
\begin{equation}\label{eq:loss-unitary-learning}
    \LC_{\rm uni}(\thv) = 1-\frac{1}{4^n}\left|\Tr[{U}^\dagger(\thv)U_{\delta t}(\thv^*) ]\right|^2
\end{equation}
where $U_{\delta t}(\thv^*) = e^{-iH\delta t} U(\thv^*)$ is the target unitary for the given iteration. 
Our goal in this sub-section is to rewrite this loss to be in the same form as the loss in Eq.~\eqref{eq:loss} which is mainly studied in this work.

As shown in Ref~\cite{mizuta2022local}, the loss in Eq.~\eqref{eq:loss-unitary-learning} can be written as the fidelity-type loss on the $2n$-qubit composite system $\HC_A \otimes \HC_B$
\begin{equation}
    \LC_{\rm uni}(\thv) = 1- \Tr[\Phi_+^{(A,B)} \left(U(\thv)\otimes \left[e^{-i\delta t H}U^\dagger(\thv^*)\right]^*\right)\Phi_+^{(A,B)}\left(U(\thv^*)\otimes \left[e^{i\delta t H}U^\dagger(\thv)\right]^* \1\right)] \;, \label{eq:loss-unitary-learning-0}
\end{equation}
where both the ansatz and the target unitary are in space $\mathcal{H}_A$, i.e. $U(\thv)\in\mathcal{H}_A$ and $e^{-itH}\in\mathcal{H}_B$. Furthermore $\Phi_+^{(A,B)}\in \mathcal{H}_A\otimes\mathcal{H}_B$ is a state which is of the form
\begin{equation}
    \Phi_+^{(A,B)}  = \bigotimes_{j=1}^n\ketbra{\phi_+^{(A_jB_j)}}
\end{equation}
with $ \ket{\phi_+}$ being Bell state 's entangling the $j$-th qubit of the two subsystems $\mathcal{H}_A,\mathcal{H}_B$
\begin{equation}
    \ket{\phi_+^{(A_jB_j)}} = \frac{1}{\sqrt{2}}\left(\ket{0^{(A_j)}0^{(B_j)}} + \ket{1^{(A_j)}1^{(B_j)}}\right)\, .
\end{equation}

In other words, we can write $\Phi^{(A,B)}_+ = \ketbra{\1}$ where $\ket{\1}$ is the vectorized and normalized version of the $\1$ operator (see Ref. \cite{mele2023introduction}), we can apply the property
\begin{equation}
    A\otimes B^* \ket{\Phi_+} = AB^\dagger\otimes\1\ket{\Phi_+}
\end{equation}
for two general matrices $A\in\mathcal{H}_A,B\in\mathcal{H}_B$ as long as $\dim \mathcal{H}_A = \dim \mathcal{H}_B $. Therefore, we can apply this to Eq.~\eqref{eq:loss-unitary-learning-0} to rewrite the loss function as 
\begin{equation}
    \LC_{\rm uni}(\thv) = 1- \Tr[\Phi_+^{(A,B)} \left(\left[U(\thv)e^{-i\delta t H}U^\dagger(\thv^*)\right]\otimes \1\right)\Phi_+^{(A,B)}\left(\left[U(\thv^*)e^{i\delta t H}U^\dagger(\thv)\right]\otimes \1\right)] \;, \label{eq:loss-unitary-learning-1}
\end{equation}
By further rearranging the trace term in Eq.~\eqref{eq:loss-unitary-learning-1}, the loss can be rewritten as.
\begin{align}
     \mf(\thv) &= 1 - \left|\expval{\left[U(\thv)e^{-i\delta t H}U(\thv^*)\right]\otimes\1} {\Phi_+^{(A,B)}}\right|^2 \\
     &= 1-\left|\expval{\widetilde{U}(\thv)e^{-i\delta t H\otimes\1}\widetilde{U}^\dagger(\thv^*)} {\Phi_+^{(A,B)}}\right|^2
\end{align}
where we have defined $\Tilde{U}(\thv) = U(\thv)\otimes\1\in\mathcal{H}_A\otimes\mathcal{H}_B$. Hence, the loss is now in the same form as in Eq.~\eqref{eq:loss}. Furthermore, the condition for $\Tr[\Phi_+ \vsigma_1 \Phi_+ \vsigma_1] = 0$ is satisfied by choosing $\sigma_1 = \sigma_x$.

\subsubsection{Proof for the task of learning an unknown target unitary via quantum machine learning}\label{app:learning_unitary_qml}
We start by rewriting the loss function in Eq.~\eqref{eq:loss_ml} in a more convenient way on the $2n$-qubit composite system. By using the fact that $\langle \psi_j | \psi_{j'} \rangle = \delta_{j,j'}$ and denoting $\rho_0 = \frac{1}{N_s}\sum_{j=1}^{N_s}\ketbra{\psi_j}\otimes\ketbra{\psi_j}$, we have
\begin{align}
    \LC_{\rm QML} (\thv) =&  1 - \frac{1}{N_s} \sum_{j=1}^{N_s} |\langle \psi_j| U^\dagger(\thv) e^{-iH\delta t} U(\thv^*) |\psi_j \rangle |^2  \\
    =& 1-N_s \Tr[\left(U(\thv)\otimes\1\right)\rho_0 \left(U(\thv)^\dagger\otimes\1\right) \left(\left[e^{-iH\delta t} U(\thv^*)\right]\otimes\1\right) \rho_0 \left(\left[U^\dagger(\thv^*)e^{iH\delta t}\right] \otimes\1 \right) ]\\
    =& 1 - N_s \Tr[\widetilde{U}(\thv) \rho_0 \widetilde{U}^\dagger(\thv) e^{- i \delta t (H \otimes \1)} \widetilde{U}(\thv^*) \rho_0 \widetilde{U}^\dagger(\thv^*) e^{ i \delta t (H \otimes \1)}] \label{eq:loss-qml-fidel}\;,
\end{align}
with $\widetilde{U}(\thv) = U(\thv) \otimes \1$. The loss in Eq.~\eqref{eq:loss-qml-fidel} is now in the same form as the loss in Eq.~\eqref{eq:loss} (up to a factor of $N_s$). However, since $\rho_0$ of the loss here is a mixed state, we have to redo a few proof steps to ensure that Theorem~\ref{thm:variance-lower-bound} is applied. 

First, we note that since the proof of Proposition \ref{prop:variance-lower-bound} does not make use of any pure state properties, it is directly applicable here. To show that Theorem \ref{thm:variance-lower-bound} applies, we have to revisit the steps in Eq.~\eqref{eq:extend1} to Eq.~\eqref{eq:proof-coro1-2}) and show that these steps still hold even with $\rho_0$ being mixed. Indeed, by assuming that $\Tr[\rho_0\sigma_1\rho_0\sigma_1] = 0$, we have the same exact steps
\begin{align}
 \Tr\left[\left( \rho_0 - \vsigma_1 \rho_0 \vsigma_1 \right) U^{\dagger}(\vtheta^*)   \rho_{(\vtheta^*,\delta t)}  U(\vtheta^*) \ \right] 
    & = 
    \Tr[\rho_{(\vtheta^*,0)}, \rho_{(\vtheta^*, \delta t)}] - \Tr\left[\sigma_1\rho_0\sigma_1 U^{\dagger}(\vtheta^*)   \rho_{(\vtheta^*,\delta t)}  U(\vtheta^*) \right] \\
    &\geq \Tr[\rho_{(\vtheta^*,0)}, \rho_{(\vtheta^*, \delta t)}] - \Tr\left[  (\1-\rho_{(\thv^*,0)})\rho_{(\vtheta^*,\delta t)}\right] \\
    & = 2\Tr[\rho_{(\vtheta^*,0)}, \rho_{(\vtheta^*, \delta t)}] - 1 \\
    & \geq 1 - 4 \lambda_{\rm max}^2 \delta t ^2 \;, 
\end{align}
where we note that the first inequality is from the fact that $\1-\rho_{0}-\sigma_1\rho_0\sigma_1$ is a positive semi-definite matrix (which holds because both $\rho_0, \sigma_1\rho_0\sigma_1$ are matrices with eigenvalues between $0$ and $1$ and are orthogonal between them). The rest of the proof follow identically to that of Theorem~\ref{thm:variance-lower-bound}.

\medskip

We now show that by choosing $\sigma_1$ to be any global Pauli operator the assumption that $\Tr[\rho_0\sigma_1\rho_0\sigma_1]=0$ is satisfied with the probability exponentially close to $1$ for products of random single-qubit stabilizer states i.e. $|\psi_j \rangle \sim {\rm Uniform}\left( \{ |0\rangle, |1 \rangle, |+\rangle, |-\rangle, |y+\rangle, |y-\rangle \}^{\otimes n}\right)$. 

Let us denote $|\psi_j \rangle = \bigotimes_{i=1}^n |\psi_j^{(i)}\rangle$ such that $|\psi_j^{(i)}\rangle$ is a random single-qubit stabilizer state on the $i^{\rm th}$ qubit and $\sigma_1 = \bigotimes_{i=1}^n \sigma_1^{(i)}$ such that $\sigma_1^{(i)}$ is a non-trivial single-qubit Pauli operator on the $i^{\rm th}$ qubit. Then, consider the following expression
\begin{align}
    \Tr[\rho_0 \sigma_1 \rho_0 \sigma_1] & = \frac{1}{N_s^2} \sum_{j=1}^{N_s} |\expval{\sigma_1}{\psi_j}|^2  \\
    & = \frac{1}{N_s} \sum_{j=1}^{N_s} \prod_{i=1}^n \left| \langle \psi_j^{(i)} |\sigma_1^{(i)} | \psi_j^{(i)}\rangle \right|^2 \;.
\end{align}
Each $\left| \langle \psi_j^{(i)} |\sigma_1^{(i)} | \psi_j^{(i)}\rangle \right|^2$ is non-zero only if $| \psi_j^{(i)}\rangle$ is an eigenstate of $\sigma_1^{(i)}$. Since each $ \psi_j^{(i)}\rangle$ is drawn uniformly from $\{ |0\rangle, |1 \rangle, |+\rangle, |-\rangle, |y+\rangle, |y-\rangle \}$, we have the probability of  $\left| \langle \psi_j^{(i)} |\sigma_1^{(i)} | \psi_j^{(i)}\rangle \right|^2 \neq 0$ is $1/3$. Hence, the probability of each term in the sum being non-zero vanishes exponentially with the number of qubits i.e.,
\begin{align}
    {\rm Pr}\left(\prod_{i=1}^n \left| \langle \psi_j^{(i)} |\sigma_1^{(i)} | \psi_j^{(i)}\rangle \right|^2 \neq 0 \right) = \frac{1}{3^n} \;.
\end{align}
By using the union bound, we obtain that the probability of the whole sum being zero is
\begin{align}
    {\rm Pr}\left( \Tr[\rho_0 \sigma_1 \rho_0 \sigma_1] = 0 \right) \geq 1 - \frac{N_s}{3^n} \;,
\end{align}
which vanishes exponentially for $N_s \in \OC(\poly(n))$.

\end{document}